\documentclass[11pt,a4paper]{article}
\pdfoutput=1

\usepackage[margin=1in]{geometry}

\usepackage{amsmath} 
\usepackage{amsthm} 
\usepackage{amssymb}	
\usepackage{hyperref}
\usepackage{stmaryrd} 
\usepackage{gensymb} 
\usepackage{enumitem} 

\newtheorem{theorem}{Theorem}
\newtheorem*{thm*}{Theorem}
\newtheorem{proposition}[theorem]{Proposition}
\newtheorem*{prop*}{Proposition}
\newtheorem{lemma}[theorem]{Lemma}
\newtheorem{corollary}[theorem]{Corollary}
\theoremstyle{definition}
\newtheorem{definition}[theorem]{Definition}
\newtheorem*{ex}{Example}
\newtheorem*{rem}{Remark}

\newtheorem*{rep@theorem}{\rep@title}
\newcommand{\newreptheorem}[2]{%
	\newenvironment{rep#1}[1]{%
    \def\rep@title{#2 \ref{##1} (restated)}%
		\begin{rep@theorem}\itshape}%
		{\end{rep@theorem}}}
\newreptheorem{thm}{Theorem}

\newcommand{\abs}[1]{\left| #1 \right|}
\newcommand{\avg}[1]{\left\langle #1 \right\rangle}
\let\ang=\avg

\renewcommand{\AA}{\mathbb{A}}
\newcommand{\CC}{\mathbb{C}}

\newcommand{\QQ}{\mathbb{Q}}
\newcommand{\RR}{\mathbb{R}}
\newcommand{\NN}{\mathbb{N}}
\newcommand{\ZZ}{\mathbb{Z}}
\newcommand{\AAnz}{\AA\setminus\{0\}}

\newcommand{\ket}[1]{\left| #1 \right>} 
 
\renewcommand{\t}[1]{\ensuremath{^{\otimes #1}}}

\DeclareMathOperator{\supp}{supp}
\DeclareMathOperator{\ari}{arity}
\DeclareMathOperator{\wt}{wt}

\newcommand{\GHZ}{\mathrm{GHZ}}

\newcommand{\cA}{\mathcal{A}}

\newcommand{\cE}{\mathcal{E}}
\newcommand{\cF}{\mathcal{F}}
\newcommand{\cG}{\mathcal{G}}
\newcommand{\cH}{\mathcal{H}}
\newcommand{\cL}{\mathcal{L}}
\newcommand{\cM}{\mathcal{M}}
\newcommand{\cO}{\mathcal{O}}

\newcommand{\cS}{\mathcal{B}}
\newcommand{\cT}{\mathcal{T}}
\newcommand{\cU}{\mathcal{U}}

\DeclareMathOperator{\Holant}{Holant}
\newcommand{\hol}{\mathsf{Holant}}
\newcommand{\holp}[2][]{\hol^{ #1 }\left( #2 \right)}
\let\Holp=\holp

\newcommand{\csp}{\#\mathsf{CSP}}
\newcommand{\NCSP}{\#\mathsf{CSP}}
\newcommand{\plcsp}{\mathsf{Pl}\text{-}\#\mathsf{CSP}}
\newcommand{\plhol}{\mathsf{Pl\text{-}Holant}}
\newcommand{\plholp}[2][]{\mathsf{Pl\text{-}Holant}^{ #1 }\left( #2 \right)}

\newcommand{\sP}{\#\textsf{P}}
\newcommand{\numP}{\sP}

\setlist[description]{font=\normalfont\itshape}

\newcommand{\etal}[0]{\emph{et al.}}

\let\dl=\delta
\let\ld=\lambda
\let\om=\omega
\let\Gm=\Gamma
\let\sse=\subseteq
\newcommand{\EQ}{\mathrm{EQ}}
\newcommand{\NEQ}{\mathrm{NEQ}}
\newcommand{\ONE}{\mathrm{ONE}}

\newcommand{\smm}[1]{\left(\begin{smallmatrix} #1 \end{smallmatrix}\right)}
\newcommand{\pmm}[1]{\begin{pmatrix} #1 \end{pmatrix}}

\newcommand{\ba}{\mathbf{a}}
\newcommand{\bb}{\mathbf{b}}
\newcommand{\be}{\mathbf{e}}
\newcommand{\bx}{\mathbf{x}}
\newcommand{\by}{\mathbf{y}}

\def\vc#1#2{#1 _1\zd #1 _{#2}}
\def\zd{,\ldots,}
\newcommand{\allf}{\Upsilon}
\newcommand{\GL}{\operatorname{GL}_2(\AA)}

\newcommand{\tM}{\cM'}

\usepackage[color,leftbars]{changebar}
\newcommand\new[1]{#1}

\usepackage[svgnames]{xcolor} 

\usepackage{tikz}
\usetikzlibrary{shapes.geometric} 
\usetikzlibrary{shapes.misc} 
\usetikzlibrary{shapes.symbols} 
\usetikzlibrary{decorations.pathreplacing} 
\usetikzlibrary{decorations.markings} 
\usetikzlibrary{arrows} 

\pgfdeclarelayer{edgelayer}
\pgfdeclarelayer{nodelayer}
\pgfsetlayers{edgelayer,nodelayer,main}

\tikzstyle{none}=[inner sep=0pt]

\tikzstyle{every picture}=[baseline=-0.1cm,scale=0.5]

\tikzstyle{unitary}=[rectangle,fill=white,draw=black,minimum height=14pt,minimum width=14pt,inner sep=0pt]
\tikzstyle{wideunitary}=[rectangle,fill=white,draw=black,minimum height=14pt,minimum width=24pt,inner sep=0pt]
\tikzstyle{graph}=[ellipse,fill=White,draw=Black,minimum height=0.5cm,minimum width=3cm,inner sep=0pt]
\tikzstyle{biggraph}=[ellipse,fill=White,draw=Black,minimum height=0.75cm,minimum width=6.5cm,inner sep=0pt]
\tikzstyle{scalardiamond}=[diamond,draw=black,fill=white,inner sep=1pt,minimum size=0.4cm]
\tikzstyle{dtr}=[regular polygon,regular polygon sides=3,shape border rotate=180,fill=white,draw=black,minimum size=22pt,inner sep=0pt]
\tikzstyle{utr}=[regular polygon,regular polygon sides=3,fill=white,draw=black,minimum size=22pt,inner sep=0pt]
\tikzstyle{dtrwide}=[isosceles triangle,isosceles triangle stretches,shape border rotate=270,fill=white,draw=black,minimum height=22pt,minimum width=1.5cm,inner sep=0pt]
\tikzstyle{bigcloud}=[cloud,fill=white,draw=black]
\tikzstyle{diagram}=[rectangle,fill=white,draw=black,minimum height=0.5cm,minimum width=1cm,inner sep=0pt]

\tikzstyle{map}=[trapezium,trapezium left angle=90,trapezium right angle=120,fill=White,draw=Black,inner sep=0pt, minimum height=0.5cm]
\tikzstyle{map_dagger}=[trapezium,trapezium left angle=90,trapezium right angle=60,fill=White,draw=Black,inner sep=0pt, minimum height=0.5cm]
\tikzstyle{map_transpose}=[trapezium,trapezium left angle=120,trapezium right angle=90,fill=White,draw=Black,inner sep=0pt, minimum height=0.5cm]
\tikzstyle{transposemap}=[trapezium,trapezium left angle=60,trapezium right angle=90,fill=White,draw=Black,inner sep=0pt, minimum height=0.5cm]
\tikzstyle{solidn}=[circle,fill=Black,draw=Black,line width=0.8 pt,minimum size=5pt,inner sep=0pt]
\tikzstyle{hollown}=[circle,fill=White,draw=Black,line width=0.5 pt,minimum size=5pt,inner sep=0pt]
\tikzstyle{greyn}=[circle,fill=Grey,draw=Black,line width=0.5 pt,minimum size=5pt,inner sep=0pt]

\title{A full dichotomy for Holant\textsuperscript{c}, inspired by quantum computation\footnote{This paper combines extended versions of the conference publications \cite{backens_new_2017} and \cite{backens_complete_2018} with some new results.}}
\author{Miriam Backens}
\date{}

\begin{document}

\maketitle

\begin{abstract}
 \noindent Holant problems are a family of counting problems parameterised by sets of algebraic-complex valued constraint functions, and defined on graphs.
 They arise from the theory of holographic algorithms, which was originally inspired by concepts from quantum computation.
 
 Here, we employ quantum information theory to explain existing results about holant problems in a concise way and to derive two new dichotomies: one for a new family of problems, which we call $\hol^+$, and, building on this, a full dichotomy for $\hol^c$.
 These two families of holant problems assume the availability of certain unary constraint functions -- the two pinning functions in the case of $\hol^c$, and four functions in the case of $\hol^+$ -- and allow arbitrary sets of algebraic-complex valued constraint functions otherwise.
 The dichotomy for $\hol^+$ also applies when inputs are restricted to instances defined on planar graphs.
 In proving these complexity classifications, we derive an original result about entangled quantum states.
\end{abstract}

\section{Introduction}

Quantum computation provided the inspiration for holographic algorithms \cite{valiant_holographic_2008}, which in turn inspired the holant framework for computational counting problems (first introduced in the conference version of \cite{cai_complexity_2014}).
Computational counting problems include a variety of computational problems, from combinatorial problems defined on graphs to the problems of computing partition functions in statistical physics and computing amplitudes in quantum computation.
They are being analysed in different frameworks, including that of counting constraint satisfaction problems (counting CSPs) and that of holant problems.
Computational counting problems are an area of active research, yet so far there appear to have been no attempts to apply knowledge from quantum information theory or quantum computation to their analysis.
Nevertheless, as we show in the following, quantum information theory, and particularly the theory of quantum entanglement, offer promising new avenues of research into holant problems.

A holant problem is parameterised by a set of functions $\cF$; in this paper we consider finite sets of algebraic complex-valued functions of Boolean inputs.
The restriction to finite sets follows the standard in the counting CSP community.
We use it to avoid issues around efficient computability that arise when allowing problems to be parameterised by infinite sets of functions.
In the following, the set of all algebraic complex-valued functions of Boolean inputs is denoted $\allf$.
We also write $\allf_n:=\{f\in\allf\mid\ari(f)=n\}$ for the restriction of $\allf$ to functions of arity $n$.
An instance of the problem $\hol(\cF)$ consist of a multigraph $G=(V,E)$ with vertices $V$ and edges $E$, and a map $\pi$.
This map assigns to each vertex $v\in V$ a function $\pi(v)=f_v\in\cF$.
The map also sets up a bijection between the edges incident on $v$ and the arguments of $f_v$, so the degree of $v$ must equal the arity of $f_v$.
Given the map $\pi$, any assignment $\sigma:E\to\{0,1\}$ of Boolean values to edges induces a weight
\[
 \wt(\sigma) := \prod_{v\in V} f_v(\sigma|_{E(v)}),
\]
where $\sigma|_{E(v)}$ is the restriction of $\sigma$ to the edges incident on $v$.
The desired output of the holant problem is the total weight $\sum_{\sigma:E\to\{0,1\}}\wt(\sigma)$, where the sum is over all assignments $\sigma$.
Formally, we define the problem $\hol(\cF)$ as follows.

\begin{description}[noitemsep]
 \item[Name] $\hol(\cF)$
 \item[Instance] A tuple $(G,\cF,\pi)$.
 \item[Output] $\Holant_\Omega = \sum_{\sigma: E\to\{0,1\}} \prod_{v\in V}f_v(\sigma|_{E(v)})$.
\end{description}

For example, let $N\in\NN_{>0}$, and define $\cM_N:=\{\ONE_n \mid 1\leq n\leq N\}$, where
\[
 \ONE_n(\vc{x}{n}) = \begin{cases} 1 &\text{if } \sum_{k=1}^n x_k=1 \\ 0 &\text{otherwise.} \end{cases} 
\]
Then $\hol(\cM_N)$ corresponds to counting the number of perfect matchings on graphs of maximum degree $N$.
To see this, consider some graph $G=(V,E)$ with maximum degree $N$.
The map $\pi$ must assign to each vertex a function of appropriate arity and since all functions are invariant under permutations of the arguments, this means $\pi$ is fully determined by this requirement.
Consider an assignment $\sigma: E\to\{0,1\}$, this corresponds to a subset of edges $E_\sigma:=\{e\in E\mid \sigma(e)=1\}$.
Note that $\wt(\sigma)\in\{0,1\}$ for all $\sigma$, since each weight is a product of functions taking values in the set $\{0,1\}$.
Now suppose $\wt(\sigma)=1$, then for each $v\in V$, $\sigma$ assigns the value 1 to exactly one edge in $E(v)$.
Thus, $E_\sigma$ is a perfect matching.
Conversely, if $\wt(\sigma)=0$ then there exists a vertex for which $f_v(\sigma|_{E(v)})=0$.
This implies $\sum_{e\in E(v)}\sigma(e)\neq 1$, i.e.\ either 0 or at least 2 of the edges incident on $v$ are assigned 1 by $\sigma$.
Thus, $E_\sigma$ contains either no edges incident on $v$, or at least 2 edges incident on $v$, so $E_\sigma$ is not a perfect matching.
We have shown that $\wt(\sigma)=1$ if $\sigma$ corresponds to a perfect matching and $\wt(\sigma)=0$ otherwise.
The correspondence between assignments $\sigma:E\to\{0,1\}$ and subsets $E_\sigma\sse E$ is a bijection.
Therefore, the total weight -- and hence the output of $\hol(\cF)$ -- is exactly the number of perfect matchings of $G$.

Sometimes it is interesting to restrict holant problems to instances defined on planar graphs.
In that case, the correspondence between the edges incident on a vertex $v$ and the arguments of its assigned function $f_v$ must respect the ordering of the edges in the given planar embedding of the graph and the ordering of the arguments of the function.
This means that choosing an edge to correspond to the first argument of $f_v$ fixes the bijection between edges and arguments for all the remaining edges and arguments.
Given a finite set of functions $\cF$, the problem of computing holant values on planar graphs is denoted $\plhol(\cF)$.
For example, $\plhol(\cM_N)$ is the problem of counting perfect matchings on planar graphs of maximum degree $N$.

The complexity of holant problems has not been fully classified yet, but there exist a number of classifications for subfamilies of holant problems.
These classifications restrict the family of problems in one or more of the following ways:
\begin{itemize}
 \item by considering only symmetric functions, i.e.\ functions that are invariant under any permutation of the arguments (these functions depend only on the Hamming weight of their input), or
 \item by restricting the codomain, e.g. to algebraic real numbers or even non-negative real numbers, or
 \item by considering only sets of functions that contain some specified subset of functions, such as arbitrary unary functions or the two `pinning functions' $\dl_0(x)=1-x$ and $\dl_1(x)=x$.
\end{itemize}
Known classification include a full dichotomy for holant problems where all unary functions are assumed to be available \cite{cai_dichotomy_2011}, and a classification for $\holp[c]{\cF}:=\holp{\cF\cup\{\dl_0,\dl_1\}}$ with symmetric functions \cite{cai_holant_2012}.
There is also a dichotomy for real-valued $\hol^c$, where functions need not be symmetrical but must take values in $\RR$ instead of $\CC$ \cite{cai_dichotomy_2017}.
Both existing results about $\hol^c$ are proved via dichotomies for counting CSPs with complex-valued, not necessarily symmetric functions (see Section~\ref{s:csp} for a formal definition): in the first case, a dichotomy for general $\csp$ \cite{cai_complexity_2014}, and in the second case, a dichotomy for $\csp_2^c$, a subfamily of $\csp$ in which each variable must appear an even number of times and variables can be pinned to 0 or 1 \cite{cai_dichotomy_2017}.
The only broad classifications not assuming availability of certain functions are the dichotomy for symmetric holant \cite{cai_complete_2016} and the dichotomy for non-negative real-valued holant \cite{lin_complexity_2016}.

The problem of computing amplitudes in quantum computation can also be expressed as a holant problem.
This makes certain holant problems examples of \emph{(strong) classical simulation of quantum computations}, an area of active research in the theory of quantum computation.
If a quantum computation can be simulated efficiently on a classical computer, the quantum computer yields no advantage.
If, on the other hand, the complexity of classically simulating a certain family of quantum computations can be shown to satisfy some hardness condition, this provides evidence that quantum computations are indeed more powerful than classical computers.

Indeed, while the origin of holant problems is not related to the notion of classical simulation of quantum computations, quantum computing did provide the inspiration for its origins \cite{valiant_holographic_2008,cai_complexity_2014}.
Nevertheless, so far there have been no attempts to apply knowledge from quantum information theory or quantum computation to the analysis of holant problems.
Yet, as we show in the following, quantum information theory, and particularly the theory of quantum entanglement, offer promising new avenues of research into holant problems.

The complexity of $\holp{\cF}$ may depend on \emph{decomposability} properties of the functions in $\cF$.
A function $f\in\allf_n$ is considered to be decomposable if there exists a permutation $\rho:[n]\to [n]$, an integer $k$ satisfying {$1\leq k<n$}, and functions $f_1,f_2$ such that 
\[
 f(\vc{x}{n}) = f_1(x_{\rho(1)}\zd x_{\rho(k)})f_2(x_{\rho(k+1)}\zd x_{\rho(n)}).
\]
For example, the constant function $f(x_1,x_2)=1$ is equal to $f_1(x_1)f_2(x_2)$, where $f_1(x)=f_2(x)=1$.
Not all functions are decomposable: there are no functions $g_1,g_2\in\allf_1$ such that $\EQ_2(x_1,x_2)=g_1(x_1)g_2(x_2)$.

The functions in $\allf_n$ are in bijection with the algebraic complex-valued vectors of length $2^n$ by mapping each function to a list of its values, or conversely.
Under this map, the notion of a function being non-decomposable corresponds exactly to the quantum-theoretical notion of a vector being \emph{genuinely entangled}.
We therefore draw on the large body of research about the properties of entangled vectors from quantum theory, and apply it to the complexity classification of holant problems.
In the process, we also derive a new result about entanglement previously unknown in the quantum theory literature, this is Theorem~\ref{thm:three-qubit-entanglement}.

For relating the complexity of two computational problems, the main technique used in this paper is that of \emph{gadget reductions}.
An $n$-ary gadget over some set of functions $\cF$ is a fragment of a signature grid using functions from $\cF$, which is connected to the rest of the graph by $n$ external edges.
This signature grid is assigned an effective function in $\allf_n$ by treating the external edges as variables and summing over all possible assignments of Boolean values to the internal edges.
If the function $g$ is the effective function of some gadget over $\cF$, then we say $g$ is \emph{realisable} over $\cF$.
Allowing functions realisable over $\cF$ in addition to the functions in $\cF$ does not affect the complexity of computing the holant \cite{cai_dichotomy_2011}.

Now, the problem $\hol(\cF)$ is known to be polynomial-time computable if $\cF$ contains only functions that decompose as products of unary and binary functions \cite[Theorem~2.1]{cai_dichotomy_2011}.
For $\hol(\cF)$ to be \sP-hard, $\cF$ must thus contain a function that does not decompose in this way.
Borrowing terminology from quantum theory, we say functions that do not decompose as products of unary and binary functions have \emph{multipartite entanglement}.
By applying results from quantum theory, we show how to build a gadget for a non-decomposable ternary function using an arbitrary function with multipartite entanglement and the four unary functions
\[
 \dl_0(x) = 1-x, \qquad
 \dl_1(x) = x, \qquad
 \dl_+(x) = 1, \quad\text{and}\quad
 \dl_-(x) = (-1)^x.
\]
We further apply entanglement theory to analyse a family of symmetric ternary gadgets, and show that given some non-decomposable ternary function, this family lets us realise a symmetric non-decomposable ternary function.
Furthermore, we show how to realise symmetric non-decomposable binary functions with specific properties.
Together, these gadgets allow us to classify the complexity of the problem $\Holp[+]{\cF}:=\holp{\cF\cup\{\dl_0,\dl_1,\dl_+,\dl_-\}}$.

\begin{theorem}[Informal statement of Theorem~\ref{thm:holant_plus}]
 Let $\cF\sse\allf$ be finite.
 Then $\Holp[+]{\cF}$ can be computed in polynomial time if $\cF$ satisfies one of five explicit conditions.
 In all other cases, $\Holp[+]{\cF}$ is \numP-hard.
 The same dichotomy holds when the problem is restricted to instances defined on planar graphs.
\end{theorem}

By combining these entanglement-based techniques with techniques from the real-valued $\hol^c$ classification in \cite{cai_complexity_2017}, we then develop a complexity classification for $\holp[c]{\cF}:=\holp{\cF\cup\{\dl_0,\dl_1\}}$.
The problem $\hol^c$ had first been described about a decade ago, with a full complexity classification remaining open until now.

\begin{theorem}[Informal statement of Theorem~\ref{thm:holant-c}]
 Let $\cF\sse\allf$ be finite.
 Then $\Holp[c]{\cF}$ can be computed in polynomial time if $\cF$ satisfies one of six explicit conditions.
 In all other cases, $\Holp[c]{\cF}$ is \numP-hard.
\end{theorem}

The remainder of this paper is structured as follows.
We introduce preliminary definitions in Section~\ref{s:Holant_problem} and give an overview over known holant complexity classifications in Section~\ref{s:existing_results}.
In Section~\ref{s:quantum_states}, we introduce relevant notions and results from quantum theory, particularly about entanglement.
The complexity classification for $\hol^+$ is derived in Section~\ref{s:Holant_plus} and the complexity classification for $\hol^c$ is derived in Section~\ref{s:dichotomy}.

\section{Preliminaries}
\label{s:Holant_problem}

Holant problems are a framework for computational counting problems on graphs, introduced by Cai \etal{} \cite{cai_complexity_2014}, and based on the theory of holographic algorithms developed by Valiant \cite{valiant_holographic_2008}.

Throughout, we consider algebraic complex-valued functions with Boolean inputs.
Let $\AA$ be the set of algebraic complex numbers.
For any non-negative integer $k$, let $\allf_k$ be the set of functions $\{0,1\}^k\to\AA$, and define $\allf:=\bigcup_{k\in\NN} \allf_k$.
Let $\cF\sse\allf$ be a finite set of functions, and let $G=(V,E)$ be an undirected graph with vertices $V$ and edges $E$.
Throughout, graphs are allowed to have parallel edges and self-loops.
A \emph{signature grid} is a tuple $\Omega=(G,\cF,\pi)$ where $\pi$ is a function that assigns to each $n$-ary vertex $v\in V$ a function $f_v:\{0,1\}^n\to\AA$ in $\cF$, specifying which edge incident on $v$ corresponds to which argument of $f_v$.
The Holant for a signature grid $\Omega$ is:
\begin{equation}\label{eq:Holant_dfn}
 \Holant_\Omega = \sum_{\sigma:E\to\{0,1\}} \prod_{v\in V} f_v(\sigma|_{E(v)}),
\end{equation}
where $\sigma$ is an assignment of Boolean values to each edge and $\sigma|_{E(v)}$ is the application of $\sigma$ to the edges incident on $v$.

\begin{description}[noitemsep]
 \item[Name] $\hol(\cF)$
 \item[Instance] A signature grid $\Omega=(G,\cF,\pi)$.
 \item[Output] $\Holant_\Omega$.
\end{description}

\begin{rem}
 We restrict the definition of holant problems to finite sets of functions, as is common in the counting CSP literature and some of the holant literature.
 This is to avoid issues with the representation of the function values, and it is also relevant for some of the results in Section~\ref{s:interreducing_planar}.
\end{rem}

For any positive integer $n$, let $[n]:=\{1,2\zd n\}$.
Suppose $f\in\allf_n$ and suppose $\rho:[n]\to[n]$ is a permutation, then $f_\rho(\vc{x}{n}):=f(x_{\rho(1)}\zd x_{\rho(n)})$.

A function is called \emph{symmetric} if it depends only on the Hamming weight of the input,
{i.e.\ the number of non-zero bits in the input bit string}.
An $n$-ary symmetric function is often written as $f = [f_0,f_1,\ldots,f_n]$, with $f_k$ for $k\in [n]$ being the value $f$ takes on inputs of Hamming weight $k$.
For any positive integer $k$, define the \emph{equality function} of arity $k$ as $\EQ_k:=[1,0,\ldots,0,1]$, where there are $(k-1)$ zeroes in the expression.
Furthermore, define $\ONE_k:=[0,1,0,\ldots,0]$, where there are $(k-1)$ zeroes at the end.
Let $\dl_0:=[1,0]$ and $\dl_1:=[0,1]$, these are called the \emph{pinning functions} because they allow `pinning' inputs to 0 or 1, respectively.
Finally, define $\dl_+:=[1,1]$ (this is equal to $\EQ_1$), $\dl_-:=[1,-1]$, and $\NEQ=[0,1,0]$.
{We will also occasionally use the unary functions $\dl_i:=[1,i]$ and $\dl_{-i}=[1,-i]$, where $i$ is the imaginary unit, i.e.\ $i^2=-1$.}

An $n$-ary function is called \emph{degenerate} if it is a product of unary functions in the following sense: there exist functions $\vc{u}{n}\in\allf_1$ such that $f(\vc{x}{n}) = u_1(x_1)\ldots u_n(x_n)$.
Any function that cannot be expressed as a product of unary functions in this way is called \emph{non-degenerate}.
Suppose $f\in\allf_n$, $k$ is an integer satisfying $1\leq k< n$, and $\rho:[n]\to [n]$ is a permutation such that $f_\rho(\vc{x}{n})=g(\vc{x}{k})h(x_{k+1}\zd x_n)$ for some functions $g\in\allf_k,h\in\allf_{n-k}$, then $f$ is said to be \emph{decomposable}.
Any way of writing $f$ as a product of two or more functions with disjoint sets of arguments is called a \emph{decomposition} of $f$.
{Borrowing terminology from linear algebra (which will be justified later), we say $g$ and $h$ in the above equation are \emph{tensor factors} of $f$.}
A function that cannot be decomposed in this way is called \emph{non-decomposable}.
Any degenerate function {of arity $n\geq 2$} is decomposable, but not all decomposable functions are degenerate.
For example, $f(x_1,x_2,x_3,x_4)=\EQ_2(x_1,x_2)\NEQ(x_3,x_4)$ is decomposable but not degenerate since \new{$f$} cannot be written as \new{a product} of unary functions.

For any $f\in\allf_n$, the \emph{support} of $f$ is $\supp(f):=\{\bx\in\{0,1\}^n\mid f(\bx)\neq 0\}$.

Suppose $f\in\allf_4$, $g\in\allf_3$, and $h\in\allf_2$.
As a shorthand, let $f_{xyzw}:=f(x,y,z,w)$, $g_{xyz}:=g(x,y,z)$, and $h_{xy}:=h(x,y)$ for any $x,y,z,w\in\{0,1\}$.
We sometimes identify these functions with the following matrices of their values:
\[
 f = \begin{pmatrix} f_{0000}&f_{0001}&f_{0010}&f_{0011} \\  f_{0100}&f_{0101}&f_{0110}&f_{0111} \\  f_{1000}&f_{1001}&f_{1010}&f_{1011} \\  f_{1100}&f_{1101}&f_{1110}&f_{1111} \end{pmatrix}, \quad
 g = \begin{pmatrix}g_{000}&g_{001}&g_{010}&g_{011}\\g_{100}&g_{101}&g_{110}&g_{111}\end{pmatrix} \quad\text{and}\quad
 h = \begin{pmatrix}h_{00}&h_{01}\\h_{10}&h_{11}\end{pmatrix}.
\]
{Here, for binary functions, matrix rows are labelled by the first input and columns are labelled by the second input.
For ternary functions, rows are labelled by the first input and columns by the second and third inputs in lexicographic order.
For functions of arity four, rows are labelled by the first two inputs and columns by the last two inputs, again in lexicographic order.}

For any pair of counting problems $A$ and $B$, we say $A$ reduces to $B$ and write $A\leq_T B$ if there exists a polynomial-time Turing reduction from problem $A$ to problem $B$.
If $A\leq_T B$ and $B\leq_T A$, we say $A$ and $B$ are \emph{interreducible} and write $A\equiv_T B$.

The following result is well-known in the literature, see e.g.~\cite[p.~12]{cai_complexity_2017}.

\begin{lemma}\label{lem:scaling}
 Suppose $\cF\sse\allf$ is finite, $g\in\allf$, and $c\in\AAnz$. Then $\holp{\cF\cup\{c\cdot g\}} \equiv_T \holp{\cF\cup\{g\}}$.
\end{lemma}

Given a bipartite graph $G=(V,W,E)$, with vertex partitions $V$ and $W$, we can define a \emph{bipartite signature grid}.
Let $\cF$ and $\cG$ be two finite subsets of $\allf$ and let $\pi:V\cup W\to\cF\cup\cG$ be a map assigning functions to vertices, with the property that $\pi(v)\in\cF$ for all $v\in V$ and $\pi(w)\in\cG$ for all $w\in W$.
The bipartite signature grid specified in this way is denoted by the tuple $(G,\cF|\cG,\pi)$.
The corresponding \emph{bipartite holant problem} is $\holp{\cF\mid\cG}$.

The following reductions relate bipartite and non-bipartite holant problems.

\begin{proposition}[{\cite[Proposition~2]{cai_holant_2012}}]\label{prop:make_bipartite}
 For any finite $\cF\sse\allf$, we have
 \[
  \holp{\cF}\equiv_T\holp{\cF\mid\{\EQ_2\}}.
 \]
\end{proposition}

\begin{proposition}[{\cite[Proposition~3]{cai_holant_2012}}]\label{prop:bipartite}
 Suppose $\cG_1,\cG_2\sse\allf$ are finite, then
 \[
  \holp{\cG_1\cup\{\EQ_2\}\mid\cG_2\cup\{\EQ_2\}} \equiv_T \holp{\cG_1\cup\cG_2}.
 \]
\end{proposition}

A signature grid $\Omega=(\cF,G,\pi)$ is called \emph{planar} if $G$ is a plane graph and, for each $v$, the arguments of $f_v$ are ordered counterclockwise starting from an edge specified by $\pi$ \cite[Section~2.1]{cai_holographic_2017}.
We denote by $\plhol(\cF)$ the problem $\hol(\cF)$ restricted to planar signature grids.

\subsection{Signature grids in terms of vectors}
\label{s:vector_perspective}

As noted in \cite{cai_valiants_2006}, any function $f\in\allf_n$ can be considered as a vector in $\AA^{2^n}$, which is the list of values of $f$, indexed by $\{0,1\}^n$.
Let $\{\ket{\bx}\}_{\bx\in\{0,1\}^n}$ be an orthonormal basis\footnote{In using this notation for vectors, called \emph{Dirac notation} and common in quantum computation and quantum information theory, we anticipate the interpretation of the vectors associated with functions as quantum states, cf.\ Section \ref{s:quantum_states}.} for $\AA^{2^n}$.
The vector corresponding to the function $f$ is then denoted by $\ket{f} := \sum_{\bx\in\{0,1\}^n} f(\bx)\ket{\bx}$.

Denote by $\otimes$ the Kronecker product of matrices, which we usually call the \emph{tensor product}.
Based on this \new{notion of tensor product}, we define \emph{tensor powers} of a matrix $M$ as follows: $M\t{1}:=M$ and $M\t{k+1}=M\t{k}\otimes M$ for any positive integer $k$.
The operations of tensor product and tensor power can be extended to vectors by considering them as single-column matrices.
Denote by $M^T$ the transpose of the matrix $M$.

Suppose $\Omega=(G,\cF|\cG,\pi)$ is a bipartite signature grid, where $G=(V,W,E)$ has vertex partitions $V$ and $W$.
Then the holant for $\Omega$ can be written as:
\begin{equation}\label{eq:bipartite_Holant_vectors}
 \Holant_\Omega = \left(\bigotimes_{w\in W} \left(\ket{g_w}\right)^T\right) \left(\bigotimes_{v\in V} \ket{f_v}\right) = \left(\bigotimes_{v\in V} \left(\ket{f_v}\right)^T\right) \left(\bigotimes_{w\in W} \ket{g_w}\right),
\end{equation}
where the tensor products are assumed to be ordered such that, in each inner product, two components associated with the same edge meet.

\subsection{Holographic reductions}

\emph{Holographic transformations} are the origin of the name `holant problems'.
Let $\GL$ be the set of all invertible 2 by 2 matrices over $\AA$, and let $\cO:=\{M\in\GL\mid M^TM=\smm{1&0\\0&1}\}$ be the set of orthogonal matrices.

Suppose $M\in\GL$, then for any $f\in\allf_0$ let $M\circ f=f$, and for any $f\in\allf_n$ with $n>0$ let $M\circ f$ denote the function corresponding to the vector $M\t{n}\ket{f}$.
Furthermore, for any set of functions $\cF$, define $M\circ\cF := \{ M\circ f \mid f\in\cF \}$.

\begin{theorem}[Valiant's Holant Theorem \cite{valiant_holographic_2008} as stated in {\cite[Proposition~4]{cai_holant_2012}}]\label{thm:Valiant_Holant}
 For any $M\in\GL$ and any finite sets $\cF,\cG\sse\allf$,
 \[
  \Holp{\cF\mid\cG} \equiv_T \Holp{M\circ\cF\mid (M^{-1})^T\circ\cG}.
 \]
\end{theorem}

\begin{corollary}[{\cite[Proposition~5]{cai_holant_2012}}]\label{cor:orthogonal-holographic}
 Let $O\in\cO$ and let $\cF\sse\allf$ be finite, then
 \[
  \Holp{\cF} \equiv_T \Holp{O\circ\cF}.
 \]
\end{corollary}

Going from a set of functions $\cF\mid\cG$ to $M\circ\cF\mid (M^{-1})^T\circ\cG$ or from $\cF$ to $O\circ\cF$ is a \emph{holographic reduction}.

\subsection{Gadgets and realisability}

A \emph{gadget} over a finite set of functions $\cF$ (also called $\cF$-gate) is a fragment of a signature grid with some `dangling' edges.
Any such gadget can be assigned an effective function $g$.

Formally, let $G=(V,E,E')$ be a graph with vertices $V$, (normal) edges $E$, and dangling edges $E'$, where $E\cap E'=\emptyset$.
Unlike a normal edge, each dangling edge has only one end incident on a vertex in $V$, the other end is dangling.
The gadget is determined by a tuple $\Gamma=(\cF,G,\pi)$, where $\pi:V\to\cF$ assigns a function to each vertex $v$ in such a way that each argument of the function corresponds to one of the edges (normal or dangling) incident on $v$.
Suppose $E'=\{\vc{e}{n}\}$, then the effective function associated with this gadget is
\[
 g_\Gamma(\vc{y}{n}) = \sum_{\sigma:E\to\{0,1\}} \prod_{v\in V} f_v(\hat{\sigma}|_{E(v)}),
\]
where $\hat{\sigma}$ is the extension of $\sigma$ to domain $E\cup E'$ which satisfies $\hat{\sigma}(e_k)=y_k$ for all $k\in[n]$, and $\hat{\sigma}|_{E(v)}$ is the restriction of $\hat{\sigma}$ to edges (both normal and dangling) which are incident on $v$.

If $g$ is the effective function of some gadget over $\cF$, then $g$ is said to be \emph{realisable over $\cF$}.
This notion is sometimes extended to say $g$ is realisable over $\cF$ if there exists a gadget over $\cF$ with effective function $c\cdot g$ for some $c\in\AAnz$.
By Lemma~\ref{lem:scaling}, the extended definition does not affect the validity of the following lemma.

\begin{lemma}[{\cite[p.~1717]{cai_dichotomy_2011}}]\label{lem:realisable}
 Suppose $\cF\sse\allf$ is finite and ${g}$ is realisable over $\cF$.
 Then
 \[
  \Holp{\cF\cup\{{g}\}} \equiv_T \Holp{\cF}.
 \]
\end{lemma}

Following \cite{lin_complexity_2016}, we define $S(\cF) = \{ {g} \mid {g} \text{ is realisable over } \cF \}$ for any set of functions $\cF$.
Then Lemma~\ref{lem:realisable} implies that, for any finite subset $\cF'\sse S(\cF)$, we have $\Holp{\cF'} \leq_T \Holp{\cF}$.

The following lemma will be useful later.
This result is stated e.g.\ in \cite[Lemma~2.1]{cai_dichotomy_2017}, but as it is not proved there, we give a quick proof here.
Note this lemma uses the scaled definition of realisability.

\begin{lemma}\label{lem:decomposable}
 Suppose $f(\vc{x}{n+m})=g(\vc{x}{n})h(x_{n+1}\zd x_{n+m})$, where none of these functions are identically zero.
 Then $g,h\in S(\{f,\dl_0,\dl_1\})$. 
\end{lemma}
\begin{proof} 
 As $g$ is not identically zero, there exists $\ba\in\{0,1\}^n$ such that $g(\ba)\neq 0$.
 But then
 \[
  h(\vc{y}{m}) = \sum_{\vc{x}{n}\in\{0,1\}} f(\vc{x}{n},\vc{y}{m}) \prod_{k\in [n]} \dl_{a_k}(x_k).
 \]
 The right-hand side is the effective function of some gadget over $\{f,\dl_0,\dl_1\}$, which consists of one copy of $f$ connected to $n$ unary functions, so $h\in S(\{f,\dl_0,\dl_1\})$.
 
 An analogous argument with the roles of $g$ and $h$ swapped shows that $g\in S(\{f,\dl_0,\dl_1\})$.
\end{proof}

When considering a bipartite Holant problem $\holp{\cF\mid\cG}$ for some finite $\cF,\cG\sse\allf$, we need to use gadgets that respect the bipartition.
Suppose $\Gm=(\cF|\cG,G,\pi)$ where $G=(V,W,E,E')$ is a bipartite graph with vertex partitions $V$ and $W$, (normal) edges $E$ and dangling edges $E'$, and suppose $\pi:V\cup W\to\cF\cup\cG$ satisfies $\pi(v)\in\cF$ for all $v\in V$ and $\pi(w)\in\cG$ for all $w\in W$.
If furthermore all dangling edges are incident on vertices from $V$, then $\Gm$ is called a \emph{left-hand side (LHS) gadget} over $\cF|\cG$.
Otherwise, if all dangling edges are incident on vertices from $W$, then $\Gm$ is called a \emph{right-hand side (RHS) gadget} over $\cF|\cG$.
The following result is a straightforward extension of Lemma~\ref{lem:realisable}.

\begin{lemma}\label{lem:realisable_planar}
 Let $\cF,\cG\sse\allf$ be two finite sets of functions.
 Suppose $f$ is an LHS gadget over $\cF|\cG$ and $g$ is a RHS gadget over $\cF|\cG$.
 Then
 \[
  \holp{\cF\cup\{f\}\mid\cG} \equiv_T \holp{\cF\mid\cG} \quad\text{and}\quad \holp{\cF\mid\cG\cup\{g\}} \equiv_T \holp{\cF\mid\cG}.
 \]
\end{lemma}

A gadget is called \emph{planar} if it is defined by a plane graph and if the dangling edges, ordered counterclockwise corresponding to the order of the arguments of the effective function, are in the outer face in a planar embedding \cite[Section~2.4]{cai_holographic_2017}.
In reductions between planar holant problems, only planar gadgets may be used.

\subsection{Polynomial interpolation}

Finally, there is the technique of \emph{polynomial interpolation}.
Let $\cF$ be a set of functions and suppose $g$ is a function that cannot be realised over $\cF$.
If, given any signature grid over $\cF\cup\{g\}$, it is possible to set up a family of signature grids over $\cF$ such that the holant for the original problem instance can be determined efficiently from the holant values of the family by solving a system of linear equations, then $g$ is said to be \emph{interpolatable} over $\cF$.
We do not directly use polynomial interpolation here, though the technique is employed by many of the results we build upon.
A rigorous definition of polynomial interpolation can be found in \cite{cai_complexity_2014}.

{\subsection{Linear algebra lemmas for holographic transformations}
\label{s:linear-algebra}

Holographic transformations are, at their core, linear maps.
In this section, we give a few lemmas about decompositions of matrices that will significantly simplify later arguments about these transformations.
In particular, we extend the orthogonal QR decomposition from real to complex matrices and prove two further lemmas building on it.
These result are straightforward and may not be novel, but we have not been able to find a reference for them in the literature, so we provide proofs for completeness.}

Let $K := \left(\begin{smallmatrix}1&1\\i&-i\end{smallmatrix}\right)$, $X := \left( \begin{smallmatrix}0&1\\1&0\end{smallmatrix}\right)$ and $T:=\smm{1&0\\0&\exp(i\pi/4)}$ {where $i$ is the imaginary unit}; these are all elements of $\GL$.

Recall that any real square matrix $M$ can be written as $M=QR$ where $Q$ is an orthogonal matrix and $R$ is upper (or lower) triangular.
The equivalent result for complex matrices requires $Q$ to be unitary instead of orthogonal.
Nevertheless, many complex 2 by 2 matrices do admit a decomposition with a complex orthogonal matrix and an upper or lower triangular matrix.
Where this is not possible, we give an alternative decomposition using the matrix $K$ {defined above} instead of an orthogonal matrix.

\begin{lemma}[Orthogonal QR decomposition for complex matrices]\label{lem:QR_decomposition}
 Let $M$ be an invertible 2 by 2 complex matrix, write {$\cO$} for the set of all 2 by 2 complex orthogonal matrices, and let $K$ be as defined above.
 Then the following hold:
 \begin{itemize}
  \item There exists $Q\in \cO\cup\{K, KX\}$ such that $Q^{-1}M$ is upper triangular.
  \item There exists $Q\in \cO\cup\{K,KX\}$ such that $Q^{-1}M$ is lower triangular.
  \item If $Q^{-1}M$ is neither lower nor upper triangular for any orthogonal $Q$, then $M=KD$ or $M=KXD$, where $D$ is diagonal.
 \end{itemize}
\end{lemma}
\begin{proof}
 Write $M$ as:
 \[
  M = \begin{pmatrix} x&y\\z&w\end{pmatrix}.
 \]
 We assumed $M$ was invertible, so $\det M = xw-yz\neq 0$.
 {Note that $K^{-1} = \frac{1}{2}\smm{1&-i\\1&i}$ and $(KX)^{-1} = \frac{1}{2}\smm{1&i\\1&-i}$.
 
 For a lower triangular decomposition, we want the top right element of $Q^{-1}M$ to vanish, this is $(Q^{-1})_{00} y + (Q^{-1})_{01} w$.
 The elements $y$ and $w$ cannot both be zero since $M$ is invertible.
 Suppose first $y^2+w^2=0$, i.e.\ $w=\pm iy$.
   Then
   \[
    \pmm{1&\pm i\\1&\mp i}\pmm{x&y\\z&\pm i y} = \begin{pmatrix} x\pm iz & 0 \\ x\mp iz & 2y \end{pmatrix}
   \]
   so $Q^{-1}M$ is \new{lower} triangular for some $Q\in\{K,KX\}$.
 If instead $y^2+w^2\neq 0$, then
   \[
    \frac{1}{\sqrt{y^2+w^2}}\pmm{w&-y\\y&w}\pmm{x&y\\z&w} = \frac{1}{\sqrt{y^2+w^2}}\pmm{wx-yz&0\\yx+wz&y^2+w^2},
   \]
   and $\frac{1}{\sqrt{y^2+w^2}}\smm{w&-y\\y&w}\in\cO$.

 Analogously, for an upper triangular decomposition, some $Q\in\{K,KX\}$ works if $x^2+z^2=0$, and some $Q\in\cO$ works otherwise.
 
 We have $Q\in\{K,KX\}$ for both decompositions if and only if} $x^2+z^2=0$ and $y^2+w^2=0$ simultaneously.
 Write $z=\pm ix$.
 Then, by invertibility of $M$, $w=\mp iy$.
 Thus, letting $D=\left(\begin{smallmatrix}x&0\\0&y\end{smallmatrix}\right)$:
 \[
  M = \begin{pmatrix} x&y\\\pm ix&\mp iy\end{pmatrix} = \begin{pmatrix} 1&1\\\pm i&\mp i\end{pmatrix} \begin{pmatrix} x&0\\0&y\end{pmatrix} = \begin{cases} KD &\text{if $\pm$ goes to $+$, or} \\ KXD & \text{if $\pm$ goes to $-$.} \end{cases}
 \]
 This completes the proof.
\end{proof}

\begin{lemma}\label{lem:ATA-D}
 Suppose $M$ is a 2 by 2 invertible complex matrix such that $M^TM=\smm{\ld&0\\0&\mu}$ for some $\ld,\mu\in\CC\setminus\{0\}$.
 Then there exists a 2 by 2 orthogonal matrix $Q$ and a diagonal matrix $D$ such that $M=QD$.
\end{lemma}
\begin{proof}
 {Given invertibility of $M$, the property $(M^TM)_{01} = (M^TM)_{10} = 0$ means that the columns of $M$ are orthogonal to each other (under the `real' inner product).
 Hence there must exists $Q\in\cO$ such that the columns of $M$ are scalings of the columns of $Q$, which implies $M=QD$.}
\end{proof}

Using the complex QR decomposition, we can also consider the solutions of $A^TA\doteq X$, where `$\doteq$' denotes equality up to scalar factor.

\begin{lemma}\label{lem:ATA-X}
 The solutions of:
 \begin{equation}\label{eq:ATA-X}
  A^TA \doteq X
 \end{equation}
 are exactly those matrices $A$ satisfying $A=KD$ or $A=KXD$ for some invertible diagonal matrix $D$.
\end{lemma}
\begin{proof}
 First, we check that matrices of the form $KD$ or $KXD$ for some invertible diagonal matrix $D$ satisfy \eqref{eq:ATA-X}.
 Indeed:
 \begin{equation}
  (KD)^TKD = D^T K^T K D = D \begin{pmatrix}1&i\\1&-i\end{pmatrix} \begin{pmatrix}1&1\\i&-i\end{pmatrix} D = 2 D X D = 2 xy X \doteq X,
 \end{equation}
 where $D = \left(\begin{smallmatrix}x&0\\0&y\end{smallmatrix}\right)$ with $x,y\in\CC\setminus\{0\}$.
 Similarly:
 \begin{equation}
  (KXD)^TKXD = D^T X^T K^T KXD = 2DX^3D = 2DXD = 2xyX \doteq X.
 \end{equation}
 This completes the first part of the proof.

 It remains to be shown that these are the only solutions of $A^TA\doteq X$.
 Assume, for the purposes of deriving a contradiction, that there is a solution that has an orthogonal QR decomposition.
 In particular, suppose $A=QR$ for some upper triangular matrix $R$.
 Then, by orthogonality of $Q$:
 \begin{equation}
  A^T A = R^TQ^TQR = R^T R = \begin{pmatrix} R_{00} & 0 \\ R_{01} & R_{11} \end{pmatrix} \begin{pmatrix} R_{00} & R_{01} \\ 0 & R_{11} \end{pmatrix} = \begin{pmatrix} R_{00}^2 & R_{00}R_{01} \\ R_{00}R_{01} & R_{01}^2+R_{11}^2 \end{pmatrix}.
 \end{equation}
 The only way for the top left component of this matrix to be zero, as required, is if $R_{00}$ is zero. 
 Yet, in that case, the top right and bottom left components of $A^T A$ are zero too, hence $A^T A$ cannot be invertible.
 That is a contradiction because any non-zero scalar multiple of $X$ is invertible.
 
 A similar argument applies if $R$ is lower triangular.
 
 Thus, all solutions of \eqref{eq:ATA-X} have to fall into the third case of Lemma \ref{lem:QR_decomposition}: i.e.\ all solutions must be of the form $A=KD$ or $A=KXD$.
\end{proof}

For any $M\in\GL$, denote by $f_M(x,y):=M_{xy}$ the binary function corresponding to this matrix.
{The following is straightforward, but we give a proof for completeness.}

\begin{lemma}\label{lem:hc_gadget}
 Suppose $M\in\GL$ and $g\in\allf$, then $M\circ g, M^T\circ g\in S(\{g,f_M\})$.
\end{lemma}
\begin{proof}
 Let $n=\ari(g)$.
 We have
 \begin{align*}
  (M\circ g)(\vc{x}{n}) 
  &= \sum_{\vc{y}{n}\in\{0,1\}} \left(\prod_{j=1}^{n}M_{x_j y_j}\right) g(\vc{y}{n}) \\
  &= \sum_{\vc{y}{n}\in\{0,1\}} \left(\prod_{j=1}^{n} f_M(x_j,y_j)\right) g(\vc{y}{n}),
 \end{align*}
 so $M\circ g\in S(\{g,f_M\})$.
 Furthermore, for any $M\in\GL$, $f_{M^T}(x,y)=f_M(y,x)$; therefore
 \[
  (\new{M^T}\circ g)(\vc{x}{n}) = \sum_{\vc{y}{n}\in\{0,1\}} \left(\prod_{j=1}^{n} f_M(y_j,x_j)\right) g(\vc{y}{n})
 \]
 and $M^T\circ g\in S(\{g,f_M\})$.
\end{proof}

\section{Known results about holant problems}
\label{s:existing_results}

We now introduce the existing families of holant problems and their associated dichotomy results.
Gadget constructions (which are at the heart of many reductions) are easier the more functions are known to be available.
As a result, several families of holant problems have been defined, in which certain sets of functions are freely available,
{i.e.\ are always added to the set of functions parameterising the holant problem.
More formally, suppose $\cG\sse\allf$ is a finite set of functions and denote by $\hol^\cG$ the holant problem where functions in $\cG$ are freely available, then $\hol^\cG(\cF) := \hol(\cF\cup \cG)$ for any set of functions $\cF$.
Effectively, the problem $\hol^\cG$ restricts analysis to cases where the set of constraint functions contains $\cG$.}

\subsection{Conservative holant}
\label{s:Holant_star}

Write $\cU:=\allf_1$ for compatibility with earlier notation.
We will not use the notation $\hol^*$ \cite{cai_complexity_2014,cai_dichotomy_2011} here as we do not define holant problems for infinite sets of constraint functions.
Instead, as is common in the counting CSP literature, we refer to holant problems where arbitrary finite subsets of $\cU$ are freely available as `conservative'.

We begin with some definitions. 
Given a bit string $\bx$, let $\bar{\bx}$ be its bit-wise complement and let $\abs{\bx}$ denote its Hamming weight.
Denote by $\avg{\cF}$ the closure of a set of functions $\cF$ under tensor products.
Furthermore, define:
\begin{itemize}
 \item the set of all unary and binary functions
  \[
   \cT := \allf_1\cup\allf_2, 
  \]
 \item the set of functions which are non-zero on at most two \new{complementary} inputs
  \[
   \cE := \{f\in\allf \mid \exists \ba\in\{0,1\}^{\ari(f)} \text{ such that } f(\bx)=0 \text{ if } \bx\notin\{\ba,\bar{\ba}\} \},
  \]
 \item the set of functions which are non-zero only on inputs of Hamming weight at most 1
  \[
   \cM := \{f\in\allf \mid f(\bx)=0 \text{ if } \abs{\bx}> 1\}.
  \]
\end{itemize}
{Note that $\cU\sse\cT$, $\cU\sse\cE$ and $\cU\sse\cM$.}
The following result has been adapted to our notation.

\begin{theorem}[{\cite[Theorem~2.2]{cai_dichotomy_2011}}]\label{thm:Holant-star}
 Suppose $\cF$ is a finite subset of $\allf$.
 If
 \begin{itemize}
  \item $\cF\subseteq\avg{\cT}$, or
  \item there exists $O\in\cO$ such that $\cF\subseteq\avg{O\circ\mathcal{E}}$, or
  \item $\cF\subseteq\avg{K\circ\mathcal{E}}=\avg{KX\circ\mathcal{E}}$, or
  \item $\cF\subseteq\avg{K\circ\mathcal{M}}$ or $\cF\subseteq\avg{KX\circ\mathcal{M}}$,
 \end{itemize}
 then, for any finite subset $\new{\cU'}\sse\cU$, the problem $\hol(\cF,\new{\cU'})$ is polynomial-time computable.
 Otherwise, there exists a finite subset $\new{\cU'}\sse\cU$ such that $\hol(\cF,\new{\cU'})$ is \numP-hard.
 The dichotomy is still true even if the inputs are restricted to planar graphs.
\end{theorem}

{To get an intuition for the polynomial-time computable cases, first note that every tensor product can be thought of as a gadget over its factors, and that tensor closure commutes with holographic transformations.
Furthermore, if a signature grid is not connected, its holant is just the product of the holant values of the individual connected components.

Now, in the first polynomial-time computable case of Theorem~\ref{thm:Holant-star}, $\cF\sse\ang{\cT}$, the signature grid can be transformed to one in which every vertex has degree at most 2 by replacing every function with a disconnected gadget over unary and binary functions.
Then each connected component is a path or cycle, and its holant value can be computed by matrix multiplication.

In the second polynomial-time computable case, $\cF\subseteq\avg{O\circ\mathcal{E}}$, again replace decomposable functions by disconnected gadgets.
Suppose $O$ is the identity matrix, then an assignment of a Boolean value to one edge contributes to at most one non--zero-weight assignment of values to all edges in the same connected component.
Thus the holant value for any connected graph is a sum of at most two terms, which can be computed efficiently, and hence the overall value can be found.
If $O$ is not the identity, apply Corollary~\ref{cor:orthogonal-holographic} to reduce to the identity case.

For the third polynomial-time computable case, $\cF\subseteq\avg{K\circ\mathcal{E}}$, note that given any finite $\cG\sse\allf$,
\[
 \hol\left(K\circ\cG\right) \equiv_T \hol(K\circ\cG\mid\{\EQ_2\}) \equiv_T \hol(\cG\mid\{\NEQ\}) \leq_T \hol(\cG\cup\{\NEQ\}),
\]
where the first equivalence is Proposition~\ref{prop:make_bipartite}, the second is Theorem~\ref{thm:Valiant_Holant}, and the final reduction is because dropping the restriction to bipartite signature grids cannot make the problem easier.
Hence, since $\NEQ\in\cE$, this case reduces to the second one.

For the final polynomial-time computable case, $\cF\subseteq\avg{K\circ\mathcal{M}}$ or $\cF\subseteq\avg{KX\circ\mathcal{M}}$, use the first two steps of the above reduction.
It can be shown that all LHS gadgets over $\cM|\{\NEQ\}$ are in $\ang{\cM}$ (cf.\ e.g.\ \cite[Lemma~43]{backens_holant_2018}).
Thus, vertices from the right-hand side partition can be removed one-by-one, updating the signature grid in the process, until it decomposes into a discrete graph whose holant value can be computed straightforwardly.

The polynomial-time computable sets of Theorem~\ref{thm:Holant-star} will also be related to quantum entanglement in Section~\ref{s:existing_quantum}.
If $\cF$ is not one of the exceptional sets defined above, then the closure of $\cF\cup\cU$ under taking gadgets contains all functions~\cite[Theorem~67]{backens_holant_2018}.
}

\subsection{Counting constraint satisfaction problems}
\label{s:csp}

The family of complex-weighted Boolean $\csp$ (counting constraint satisfaction problems) is defined as follows.
Let $\cF\sse\allf$ be a finite set of functions and let $V$ be a finite set of variables.
A \emph{constraint} $c$ over $\cF$ is a tuple consisting of a function $f_c\in\cF$ and a \emph{scope}, which is a tuple of $\ari(f)$ (not necessarily distinct) elements of $V$.
If $C$ is a set of constraints, then any assignment $\bx:V\to\{0,1\}$ of values to the variables induces a \emph{weight} $\wt_{(V,C)}(\bx):=\prod_{c\in C}f_c(\bx|_c)$, where $\bx|_c$ denotes the restriction of $\bx$ to the scope of $c$.

\begin{description}[noitemsep]
 \item[Name] $\csp(\cF)$
 \item[Instance] A tuple $(V,C)$, where $V$ is a finite set of variables and $C$ is a finite set of constraints over $\cF$.
 \item[Output] The value $Z_{(V,C)} = \sum_{\bx:V\to\{0,1\}} \prod_{c\in C} f_c(\bx|_c)$.
\end{description}

Counting constraint satisfaction problems are closely related to holant problems.
{In particular, a counting constraint satisfaction problem can be thought of as a bipartite holant problem with the constraints as one part and the variables as the other part.
Each vertex corresponding to a variable is assigned the function $\EQ_k$, with $k$ the total number of times that variable appears across the scopes of all constraints.
The straightforward formalisation of this idea would require parameterising the holant problem with an infinite set of functions, the set containing all equality functions of any arity.
Yet the bipartite structure is not necessary: if two variable vertices are adjacent to each other, they can be merged into one, and if two constraint vertices are adjacent to each other, a new variable vertex assigned $\EQ_2$ can be introduced between them.
Now, $S(\{\EQ_3\})$ is exactly the set of all equality functions, hence:}

\begin{proposition}[{\cite[Proposition~1]{cai_holant_2012}}]\label{prop:CSP_holant}
 $\csp(\cF) \equiv_T \Holp{\cF\cup\{ \EQ_3 \}}$.
\end{proposition}

{This proposition shows that any counting constraint satisfaction problem can be expressed as a holant problem.
On the other hand, some holant problems can only be expressed as counting constraint satisfaction problems with the additional restriction that every variable must appear exactly twice across the scopes of all the constraints.
This means the holant framework is more general than counting constraint satisfaction problems.
For example, as is well known, the problem of counting matchings on a graph can be expressed in the holant framework but not in the standard $\csp$ framework as defined above (cf.~\cite[pp.~23:3--4]{backens_holant_2018}).}

The dichotomies for $\csp$ and its variants feature families of tractable functions which do not appear in Theorem~\ref{thm:Holant-star}.

\begin{definition}\label{dfn:affine_function}
 A function $f:\{0,1\}^n\to\AA$ for some non-negative integer $n$ is called \emph{affine} if it has the form:
 \begin{equation}
  f(\bx) = c i^{l(\bx)} (-1)^{q(\bx)} \chi_{A\bx=\bb}(\bx),
 \end{equation}
 where $c\in\AA$, $i^2=-1$, $l:\{0,1\}^n\to\{0,1\}$ is a linear function, $q:\{0,1\}^n\to\{0,1\}$ is a quadratic function, $A$ is an $m$ by $n$ matrix with Boolean entries for some $0\leq m\leq n$, $\bb\in\{0,1\}^m$, and $\chi$ is an indicator function which takes value 1 on inputs satisfying $A\bx=\bb$, and 0 otherwise.
 
 The set of all affine functions is denoted by $\cA$.
\end{definition}

The support of the function $\chi_{A\bx=\bb}$ is an affine subspace of $\{0,1\}^n$, hence the name `affine functions'.
There are different definitions of this family in different parts of the literature, but they are equivalent.
For the reader familiar with quantum information theory, the affine functions correspond -- up to scaling -- to stabiliser states (cf.\ Section \ref{s:existing_quantum} and also independently \cite{cai_clifford_2018}).

\begin{lemma}[{\cite[Lemma~3.1]{cai_clifford_2018}}]\label{lem:cai_clifford}
 If $f(\vc{x}{n}), g(\vc{y}{m})\in\cA$, then so are
 \begin{enumerate}
  \item $(f\otimes g)(\vc{x}{n},\vc{y}{m}) = f(\vc{x}{n}) g(\vc{y}{m})$,
  \item $f(x_{\rho(1)}\zd x_{\rho(n)})$ for any permutation {$\rho:[n]\to [n]$},
  \item $f^{x_j=x_\ell}(\vc{x}{j-1},x_{j+1}\zd x_n) = f(\vc{x}{j-1},x_\ell,x_{j+1}\zd x_n)$, setting the variable $x_j$ to be equal to $x_\ell$, and
  \item $f^{x_j=*}(\vc{x}{j-1},x_{j+1}\zd x_n) = \sum_{x_j\in\{0,1\}} f(\vc{x}{n})$.
 \end{enumerate}
\end{lemma}

\begin{definition}[\new{adapted from \cite[Definition~10]{cai_dichotomy_2017}}]
 A function $f:\{0,1\}^n\to\AA$ for some non-negative integer $n$ is called \emph{local affine} if it satisfies $\left( \bigotimes_{j=1}^{n} T^{a_j} \right) \ket{f} \in \cA$ for any $\ba\in\supp(f)$, where
 \[
  T^1 = T = \begin{pmatrix}1&0\\0&e^{i\pi/4}\end{pmatrix} \qquad\text{and}\qquad T^0 = \begin{pmatrix}1&0\\0&1\end{pmatrix}.
 \]
 The set of all local affine functions is denoted $\cL$.
\end{definition}

Both $\cA$ and $\cL$ are closed under tensor products, i.e.\ $\ang{\cA}=\cA$ and $\ang{\cL}=\cL$.

\begin{theorem}[{\cite[Theorem~3.1]{cai_complexity_2014}}]\label{thm:csp}
 Suppose $\cF\sse\allf$ is finite.
 If $\cF\subseteq\mathcal{A}$ or $\cF\subseteq\avg{\mathcal{E}}$, then $\csp(\cF)$ is computable in polynomial time. Otherwise, $\csp(\cF)$ is \sP-hard.
\end{theorem}

{Unlike a holant problem, the complexity of a counting CSP does not change if the pinning functions $\dl_0$ and $\dl_1$ are added to the set of constraint functions.

\begin{lemma}[{Pinning lemma \cite[Lemma~8]{dyer_complexity_2009}}]\label{lem:csp^c}
 For any finite $\cF\sse\allf$,
 \[
  \csp^c(\cF):=\csp(\cF\cup\{\dl_0,\dl_1\})\equiv_T\csp(\cF).
 \]
\end{lemma}

The dichotomy of Theorem~\ref{thm:csp} also holds for a variant counting constraint satisfaction problem called \#\textsf{R$_3$-CSP} with a restriction on the number of times each variable may appear \cite[Theorem~3.2]{cai_complexity_2014}.

\begin{description}[noitemsep]
 \item[Name] \#\textsf{R$_3$-CSP}$(\cF)$
 \item[Instance] A tuple $(V,C)$, where $V$ is a finite set of variables and $C$ is a finite set of constraints over $\cF$ such that each variable appears at most three times across all scopes.
 \item[Output] The value $Z_{(V,C)} = \sum_{\bx:V\to\{0,1\}} \prod_{c\in C} f_c(\bx|_c)$.
\end{description}

For any finite $\cF\sse\allf$, the problem \#\textsf{R$_3$-CSP}$(\cF)$ is equivalent to the bipartite holant problem $\Holp{\cF\mid\{\EQ_1,\EQ_2,\EQ_3\}}$ \cite[Section~2]{cai_complexity_2014}.
Its} dichotomy follows immediately from that for $\csp$ if $\cF$ contains the binary equality function (or indeed any equality function of arity at least 2), but \new{this interreduction} is non-trivial if $\cF$ does not contain any equality function of arity at least 2.

We will also consider a variant of counting CSPs in which each variable has to appear an even number of times.

\begin{description}[noitemsep]
 \item[Name] $\csp_2(\cF)$
 \item[Instance] A tuple $(V,C)$, where $V$ is a finite set of variables and $C$ is a finite set of constraints over $\cF$ such that each variable appears an even number of times across all scopes.
 \item[Output] The value $Z_{(V,C)} = \sum_{\bx:V\to\{0,1\}} \prod_{c\in C} f_c(\bx|_c)$.
\end{description}

Based on the above definition, we also define $\csp_2^c(\cF):=\csp_2(\cF\cup\{\dl_0,\dl_1\})$.
The dichotomy for this problem differs from that for plain $\csp$.

\begin{theorem}[{\cite[Theorem 4.1]{cai_dichotomy_2017}}]\label{thm:csp_2^c}
 Suppose $\cF\sse\allf$ is finite.
 A $\csp_2^c(\cF)$ problem has a polynomial time algorithm if one of the following holds:
  \begin{itemize}
   \item $\cF\subseteq\avg{\cE}$,
   \item $\cF\subseteq\cA$,
   \item $\cF\subseteq T\circ\cA$, or
   \item $\cF\subseteq\cL$.
  \end{itemize}
 Otherwise, it is \sP-hard.
\end{theorem}

{For any finite $\cF\sse\allf$, we have
 \[
  \csp(\cF\cup\{\dl_+\})\leq_T\csp_2(\cF\cup\{\dl_+\}),
 \]
 where, for each variable $y$ that appears an odd number of times in the original instance, we add a new constraint $(\dl_+,(y))$ to make it appear an even number of times instead.
 Yet if no non-zero scaling of $\dl_+$ is present or realisable (via gadgets or other methods), then there is no general reduction from $\csp(\cF)$ to $\csp_2(\cF)$.
 This lack of a general reduction can be seen for example by noting the differences between Theorem~\ref{thm:csp_2^c} (which classifies the complexity of $\csp_2^c$) and Theorem~\ref{thm:csp} (which by Lemma~\ref{lem:csp^c} also applies to $\csp^c$).}

By analogy with planar holant problems, we also define planar counting CSPs.

\begin{definition}\label{dfn:planar_CSP}
 $\plcsp(\cF):=\plhol\left(\cF\cup\{ \EQ_3 \}\right)$, i.e.\ $\plcsp(\cF)$ is the restriction of $\csp(\cF)$ to planar instances of the corresponding holant problem according to Proposition~\ref{prop:CSP_holant}.
\end{definition}

The complexity dichotomy for $\plcsp$ involves an additional tractable family as compared to the general case.
This new family is called \emph{matchgate functions} and consists of those functions which correspond to computing a weighted sum of perfect matchings.
We denote the set of all matchgate functions by\footnote{Note that we use a different symbol than \cite{cai_holographic_2017} for the set of matchgate functions, to avoid clashing with other established notation.} $\cH\sse\allf$.
As the rigorous definition of this set is somewhat intricate and not required for our work we do not reproduce it here; the interested reader can find it in \cite[pp.~STOC17-65f]{cai_holographic_2017}.
Indeed, the only property of $\cH$ which we will require is the following lemma, which is adapted to avoid having to define concepts not used in this paper.

\begin{lemma}[{\cite[first part of Lemma~2.29]{cai_holographic_2017}}]\label{lem:unary_matchgate}
 If $f\in\allf$ has arity $\leq 3$, then $f\in\cH$ if and only if one of the following parity conditions is satisfied:
 \begin{itemize}
  \item $f(\bx)=0$ whenever $\abs{\bx}$ is even, or
  \item $f(\bx)=0$ whenever $\abs{\bx}$ is odd.
 \end{itemize}
\end{lemma}

For $f\in\allf_1$, this means $f$ is a matchgate function if and only if $f=c\cdot\delta_0$ or $f=c\cdot\delta_1$ for some {$c\in\AA$}.

We can now state the dichotomy for $\plcsp$.

\begin{theorem}[{\cite[Theorem~6.1$'$]{cai_holographic_2017}}]\label{thm:planar_csp}
 Let $\cF$ be any finite set of complex-valued functions in Boolean variables.
 Then $\plcsp(\cF)$ is \sP-hard unless $\cF\sse\cA$, $\cF\sse\ang{\cE}$, or $\cF\sse\smm{1&1\\1&-1}\circ\cH$, in which case the problem is computable in polynomial time.
\end{theorem}

\subsection{Partial results for \textsf{Holant}\texorpdfstring{\textsuperscript{c}}{\textasciicircum c} and \textsf{Holant}}
\label{s:Holant_c}

$\hol^c$ is the holant problem in which the unary functions pinning edges to 0 or 1 are freely available \cite{cai_complexity_2014,cai_holant_2012}, i.e.\ $\Holp[c]{\cF} := \Holp{\cF\cup\{\delta_0,\delta_1\}}$ for any finite $\cF\sse\allf$.
We will give a full dichotomy for this problem in Section~\ref{s:dichotomy}, building on the following dichotomies for symmetric functions and for real-valued functions.

\begin{theorem}[{\cite[Theorem~6]{cai_holant_2012}}]\label{thm:Holant-c-sym}
 Let $\cF\sse\allf$ be a finite set of symmetric functions.
 $\Holp[c]{\cF}$ is \sP-hard unless $\cF$ satisfies one of the following conditions, in which case it is polynomial-time computable:
 \begin{itemize}
  \item $\cF\subseteq\avg{\cT}$, or
  \item there exists $O\in\cO$ such that $\cF\subseteq\avg{O\circ\mathcal{E}}$, or
  \item $\cF\subseteq\avg{K\circ\mathcal{E}}=\avg{KX\circ\mathcal{E}}$, or
  \item $\cF\subseteq\avg{K\circ\mathcal{M}}$ or $\cF\subseteq\avg{KX\circ\mathcal{M}}$, or
  \item there exists $B\in\cS$ such that $\cF\subseteq B\circ\cA$, where:
   \begin{equation}\label{eq:cS_definition}
    \cS = \left\{ M \,\middle|\, M^T\circ \{ \EQ_2, \delta_0, \delta_1 \} \subseteq \mathcal{A} \right\}.
   \end{equation}
 \end{itemize}
\end{theorem}

Note that the first four polynomial-time computable cases are exactly the ones appearing in Theorem~\ref{thm:Holant-star}.

The preceding results all apply to algebraic complex-valued functions, but the following theorem is restricted to algebraic real-valued functions.

\begin{theorem}[{\cite[Theorem 5.1]{cai_dichotomy_2017}}]
\label{thm:real-valued_Holant-c}
 Let $\cF$ be a set of algebraic real-valued functions.
 Then $\Holp[c]{\cF}$ is \sP-hard unless $\cF$ is a tractable family for conservative holant or for $\csp_2^c$.
\end{theorem}

In the case of $\hol$ with no \new{freely available functions}, there exists a dichotomy for complex-valued symmetric functions \cite[Theorem~31]{cai_complete_2016} and a dichotomy for (not necessarily symmetric) functions taking non-negative real values \cite[Theorem~19]{lin_complexity_2016}.
We will not explore those results in any detail here.

\subsection{Results about ternary symmetric functions}
\label{s:results_ternary_symmetric}

The computational complexity of problems of the form $\Holp{\{[y_0,y_1,y_2]\} \mid \{[x_0,x_1,x_2,x_3]\}}$, where $[y_0,y_1,y_2]\in\allf_2$ and $[x_0,x_1,x_2,x_3]\in\allf_3$, has been fully determined.
These are holant problems on bipartite graphs where one partition only contains vertices of degree~2, the other partition only contains vertices of degree~3, and all vertices of the same arity are assigned the same symmetric function.

If $[x_0,x_1,x_2,x_3]$ is degenerate, the problem is tractable by the first case of Theorem \ref{thm:Holant-star}.
If $[x_0,x_1,x_2,x_3]$ is non-degenerate, it can always be mapped to either $[1,0,0,1]$ or $[1,1,0,0]$ by a holographic transformation \cite[Section~3]{cai_holant_2012}, cf.\ also Section~\ref{s:quantum_states} below.
By Theorem \ref{thm:Valiant_Holant}, it thus suffices to consider the cases $\{[y_0,y_1,y_2]\}\mid\{[1,0,0,1]\}$ and $\{[y_0,y_1,y_2]\}\mid\{[1,1,0,0]\}$.

There are also some holographic transformations which leave the function $[1,0,0,1]$ invariant.
{In particular, if} $M=\smm{1&0\\0&\om}$, where $\om^3=1$, i.e.\ $\om$ is a third root of unity, {then $M\circ\EQ_3=\EQ_3$ \cite[Section~4]{cai_holant_2012}}.
By applying Theorem~\ref{thm:Valiant_Holant} with this $M$,
\begin{equation}\label{eq:normalisation}
 \Holp{\{[y_0,y_1,y_2]\}\mid\{[1,0,0,1]\}} \equiv_T \Holp{\{[y_0, \omega y_1, \omega^2 y_2]\}\mid\{[1,0,0,1]\}}.
\end{equation}
This relationship can be used to reduce the number of symmetric binary functions needing to be considered in this section.
Following \cite[Section~4]{cai_holant_2012}, a symmetric binary function $[y_0,y_1,y_2]$ is called \emph{$\omega$-normalised}\footnote{We use the term $\omega$-normalisation to distinguish it from other notions of normalisation, e.g.\ ones relating to the norm of the vector associated with a function.} if
\begin{itemize}
 \item $y_0=0$, or
 \item there does not exist a primitive $(3t)$-th root of unity $\lambda$, where the greatest common divisor $\operatorname{gcd}(t,3)=1$, such that $y_2=\lambda y_0$.
\end{itemize}
Similarly, a unary function $[a,b]$ is called $\omega$-normalised if
\begin{itemize}
 \item $a=0$, or
 \item there does not exist a primitive $(3t)$-th root of unity $\lambda$, where $\operatorname{gcd}(t,3)=1$, such that $b=\lambda a$.
\end{itemize}
If a binary function is not $\omega$-normalised, it can be made so through application of a holographic transformation of the form given in \eqref{eq:normalisation}.
Unary functions will only be required when the binary function has the form $[0,y_1,0]$; in that case the binary function is automatically $\omega$-normalised, and it remains so under a holographic transformation that $\omega$-normalises the unary function.

These definitions allow a complexity classification of all holant problems on bipartite signature grids where there is a ternary equality function on one partition and a non-degenerate symmetric binary function on the other partition.

\begin{theorem}[{\cite[Theorem~5]{cai_holant_2012}}]\label{thm:GHZ-state}
 Let $\mathcal{G}_1,\mathcal{G}_2\sse\allf$ be finite and let $[y_0,y_1,y_2]\in\allf_2$ be an $\omega$-normalised and non-degenerate function.
 In the case of $y_0=y_2=0$, further assume that $\mathcal{G}_1$ contains a unary function $[a,b]$ which is $\omega$-normalised and satisfies $ab\neq 0$.
 Then:
 \[
  \Holp{\{[y_0,y_1,y_2]\}\cup\mathcal{G}_1 \mid \{[1,0,0,1]\}\cup\mathcal{G}_2} \equiv_T \csp(\{[y_0,y_1,y_2]\}\cup\mathcal{G}_1\cup\mathcal{G}_2).
 \]
\end{theorem}

\begin{theorem}[{\cite[Theorem~4]{cai_holant_2012}}]\label{thm:W-state}
 $\Holp{\{[y_0,y_1,y_2]\}\mid\{[x_0,x_1,x_2,x_3]\}}$ is \sP-hard unless the functions $[y_0,y_1,y_2]$ and $[x_0,x_1,x_2,x_3]$ satisfy one of the following conditions, in which case the problem is polynomial-time computable:
 \begin{itemize}
  \item $[x_0,x_1,x_2,x_3]$ is degenerate, or
  \item there is $M\in\GL$ such that:
   \begin{itemize}
    \item $[x_0,x_1,x_2,x_3]=M\circ[1,0,0,1]$ and $M^T\circ[y_0,y_1,y_2]$ is in $\mathcal{A}\cup\avg{\cE}$,
    \item $[x_0,x_1,x_2,x_3]=M\circ[1,1,0,0]$ and $[y_0,y_1,y_2]=(M^{-1})^T\circ[0,a,b]$ for some $a,b\in\AA$.
   \end{itemize}
 \end{itemize}
\end{theorem}

Here, we have combined the last two cases of \cite[Theorem~4]{cai_holant_2012} into one case, since one can be mapped to the other by a bit flip: a holographic transformation using the matrix $X$.

We will also use a result about planar holant problems involving the ternary equality function.
As the notation used in the original statement of this theorem differs significantly from the notation used in this paper, we first state the original theorem and then prove a corollary which translates the theorem into our notation.

\begin{theorem}[{\cite[Theorem~7]{kowalczyk_holant_2016}}]\label{thm:kowalczyk-cai}
 Let $a,b\in\AA$ and define $X:=ab$, $Y:=a^3+b^3$.
 The problem $\plhol\left(\{[a,1,b]\}\mid \{\EQ_3\}\right)$ is \numP-hard for all $a,b\in\AA$ except in the following cases, for which the problem is polynomial-time computable:
 \begin{enumerate}
  \item $X=1$
  \item $X=Y=0$
  \item $X=-1$ and $Y\in\{0,\pm 2i\}$, or
  \item $4X^3=Y^2$.
 \end{enumerate}
\end{theorem}

\begin{corollary}\label{cor:pl-holant_binary}
 Suppose $g\in\allf_2$ is symmetric.
 The problem $\plhol\left(\{g\}\mid \{\EQ_3\}\right)$ is \numP-hard except in the following cases, for which the problem is polynomial-time computable:
 \begin{enumerate}
  \item\label{c:gen_eq} $g\in\ang{\cE}$,
  \item\label{c:affine} $\smm{1&0\\0&\ld}\circ g\in\cA$ for some $\ld\in\AA$ such that $\ld^3=1$, or
  \item\label{c:matchgate} $g=c\cdot [a,1,b]$, where $a,b,c\in\AAnz$ and $a^3=b^3$.
 \end{enumerate}
\end{corollary}

\begin{rem}
 Note that the three exceptional cases of Corollary~\ref{cor:pl-holant_binary} overlap, e.g.\ $g=\EQ_2$ satisfies both Case~\ref{c:gen_eq} and Case~\ref{c:affine}, and $g=[1,1,1]$ satisfies all three cases.
 
 We also include binary functions of the form $[a,0,b]$ in the corollary for completeness. 
\end{rem}

\begin{proof}[Proof of Corollary~\ref{cor:pl-holant_binary}.]
 First, we show that the problem is tractable in the exceptional cases.
 Let $M=\smm{1&0\\0&\ld}$ with $\ld^3=1$ be such that $g':=M\circ g$ is $\om$-normalised.
 Then
 \begin{align*}
  \plhol\left(\{g\}\mid \{\EQ_3\}\right)
  &\equiv_T \plhol\left(\{g'\}\mid \{\EQ_3\}\right) \\
  &\leq_T \holp{\{g', \EQ_3\}} \\
  &\leq_T \csp(\{g'\}),
 \end{align*}
 where the first step is by Theorem~\ref{thm:Valiant_Holant}, the second step is by forgetting about the bipartition and the planarity constraint, and the third step is by Proposition~\ref{prop:CSP_holant}.
 Hence if $\csp(\{g'\})$ is polynomial-time computable, then so is $\plhol\left(\{g\}\mid \{\EQ_3\}\right)$.
 By Theorem~\ref{thm:csp}, $\csp(\{g'\})$ is polynomial-time computable if $g'\in\cA$ or $g'\in\ang{\cE}$.
 Since $M$ is diagonal, $g'\in\ang{\cE}$ is equivalent to $g\in\ang{\cE}$, so we have tractability for Case~\ref{c:gen_eq}.
 Furthermore, $g'\in\cA$ is equivalent to $\smm{1&0\\0&\ld}\circ g\in\cA$ for some $\ld^3=1$, which establishes tractability for Case~\ref{c:affine}.
 
 This leaves Case~\ref{c:matchgate}.
 Assume $g=c\cdot [a,1,b]$ for some $a,b,c\in\AAnz$ with $a^3=b^3$, and define $g'':=[a,1,b]$.
 The functions $g$ and $g''$ differ only by a non-zero factor, so by a straightforward extension of Lemma~\ref{lem:scaling} to bipartite signature grids we have
 \begin{equation}\label{eq:scaling}
  \plhol\left(\{g\}\mid \{\EQ_3\}\right) \equiv_T \plhol\left(\{g''\}\mid \{\EQ_3\}\right).
 \end{equation}
 Now, $a^3=b^3$ implies
 \[
  0 = (a^3-b^3)^2 = (a^3+b^3)^2 - 4(ab)^3,
 \]
 hence $g''$ satisfies condition~4 of Theorem~\ref{thm:kowalczyk-cai}.
 This establishes tractability for Case~\ref{c:matchgate}.
 
 It remains to prove the hardness part of the theorem.
 If $g=[a,0,b]$ for some $a,b\in\AA$, then $g\in\ang{\cE}$ and $\plhol\left(\{g\}\mid \{\EQ_3\}\right)$ is polynomial-time computable by the above arguments.
 Thus, from now on, we may assume $g=c\cdot [a,1,b]$, where $a,b,c\in\AA$ and $c\neq 0$.
 
 Let $g'':=[a,1,b]$ as before, then again \eqref{eq:scaling} holds.
 We will show that if $g''$ satisfies one of the tractability conditions of Theorem~\ref{thm:kowalczyk-cai}, then $g$ satisfies one of Cases~\ref{c:gen_eq}--\ref{c:matchgate}. 
 Consider each tractability condition of Theorem~\ref{thm:kowalczyk-cai} in turn.
 \begin{enumerate}
  \item $X=1$ is equivalent to $ab=1$.
   This is exactly the condition for the function $[a,1,b]$ to be degenerate.
   But then $g=c\cdot [a,1,b]$ is degenerate as well, so $g\in\ang{\cE}$ and $g$ satisfies Case~\ref{c:gen_eq}.
  \item $X=Y=0$ is equivalent to $ab=0$ and $a^3+b^3=0$.
   Together, the two equalities imply that $a=b=0$.
   Thus $g''=\NEQ$, which implies $g=c\cdot\NEQ\in\ang{\cE}$, so $g$ satisfies Case~\ref{c:gen_eq}.
  \item $X=-1$ implies that $a,b\neq 0$.
   We can therefore rewrite $X=ab=-1$ to $b=-a^{-1}$, which in turn implies $Y=a^3-a^{-3}$.
   This case therefore reduces to $a^3-a^{-3}\in\{0,\pm 2i\}$.
   We distinguish subcases.
   \begin{itemize}
    \item Suppose $a^3-a^{-3}=0$.
     Then $a^6=1$, i.e.\ $a=e^{ik\pi/3}$ for some $k\in\{0,1,2,3,4,5\}$.
      \begin{itemize}
       \item If $k=0$, let $\ld=1$, then $\ld^3=1$ and  $\smm{1&0\\0&\ld}\circ g=c\cdot[1,1,-1]\in\cA$.
       \item If $k=1$, let $\ld=e^{4i\pi/3}$, then $\ld^3=1$ and
        \[
         \smm{1&0\\0&\ld}\circ g = c\cdot [e^{i\pi/3}, \ld, -\ld^2 e^{-i\pi/3}] = c\cdot [e^{i\pi/3}, e^{4i\pi/3}, -e^{7i\pi/3}] = e^{4i\pi/3}c\cdot [-1,1,1]
        \]
        so $\smm{1&0\\0&\ld}\circ g\in\cA$.
       \item If $k=2$, let $\ld=e^{2i\pi/3}$, then $\ld^3=1$ and
        \[
         \smm{1&0\\0&\ld}\circ g = c\cdot [e^{2i\pi/3}, \ld, -\ld^2 e^{-2i\pi/3}] = c\cdot [e^{2i\pi/3}, e^{2i\pi/3}, -e^{2i\pi/3}] = e^{2i\pi/3}c\cdot [1,1,-1]
        \]
        so $\smm{1&0\\0&\ld}\circ g\in\cA$.
      \end{itemize}
      The remaining subcases are similar.
    \item Suppose $a^3-a^{-3}= 2i$.
     Then $a^6-1=2ia^3$, or equivalently $a^6 - 2ia^3 -1=(a^3 - i)^2 = 0$.
     Thus $a = e^{(4k+1)i\pi/6}$ for some $k\in\{0,1,2\}$.
      \begin{itemize}
       \item If $k=0$, let $\ld=e^{2i\pi/3}$, then $\ld^3=1$ and
        \[
         \smm{1&0\\0&\ld}\circ g
         = c\cdot [e^{i\pi/6}, \ld, -\ld^2 e^{-i\pi/6}]
         = c\cdot [e^{i\pi/6}, e^{2i\pi/3}, -e^{7i\pi/6}]
         = e^{i\pi/6}c\cdot [1,-1,1]
        \]
        so $\smm{1&0\\0&\ld}\circ g\in\cA$.
       \item If $k=1$, let $\ld=e^{4i\pi/3}$, then $\ld^3=1$ and
        \[
         \smm{1&0\\0&\ld}\circ g
         = c\cdot [e^{5i\pi/6}, \ld, -\ld^2 e^{-5i\pi/6}]
         = c\cdot [e^{5i\pi/6}, e^{4i\pi/3}, -e^{11i\pi/6}]
         = e^{5i\pi/6}c\cdot [1,-1,1]
        \]
        so $\smm{1&0\\0&\ld}\circ g\in\cA$.
       \item If $k=2$, let $\ld=1$, then $\ld^3=1$ and
        \[
         \smm{1&0\\0&\ld}\circ g
         = c\cdot [e^{9i\pi/6}, \ld, -\ld^2 e^{-9i\pi/6}]
         = c\cdot [e^{3i\pi/2}, 1, -e^{-3i\pi/2}]
         = c\cdot [-i,1,-i]
        \]
        so $\smm{1&0\\0&\ld}\circ g\in\cA$.
     \end{itemize}
    \item Suppose $a^3-a^{-3}= -2i$.
     This subcase is analogous to the subcase $a^3-a^{-3}= 2i$.
   \end{itemize}
   In each subcase, we were able to find $\ld\in\AA$ such that $\ld^3=1$ and $\smm{1&0\\0&\ld}\circ g\in\cA$, i.e.\ the function $g$ satisfies Case~\ref{c:affine}.
  \item $4X^3=Y^2$ implies $4(ab)^3=(a^3+b^3)^2$, which is equivalent to $(a^3-b^3)^2=0$.
   Hence this condition immediately implies that $g$ satisfies Case~\ref{c:matchgate}.
 \end{enumerate}
 This completes the case distinction.
 
 We have shown that if $g''=[a,1,b]$ satisfies any of the tractability conditions of Theorem~\ref{thm:kowalczyk-cai}, then $g=c\cdot g''$ satisfies one of Cases~\ref{c:gen_eq}--\ref{c:matchgate}.
 Hence, conversely, if $g=c\cdot [a,1,b]$ and $g$ does not satisfy any of Cases~\ref{c:gen_eq}--\ref{c:matchgate}, then $g''=[a,1,b]$ does not satisfy any of the tractability conditions of Theorem~\ref{thm:kowalczyk-cai}.
 In that case, Theorem~\ref{thm:kowalczyk-cai} implies that $\plhol\left(\{g''\}\mid \{\EQ_3\}\right)$ is \numP-hard.
 Thus, by \eqref{eq:scaling}, $\plhol\left(\{g\}\mid \{\EQ_3\}\right)$ is \numP-hard.
\end{proof}

\subsection{Results about functions of arity 4}

Besides the above results about ternary functions, we will also make use of the following result about realising or interpolating the arity-4 equality function from a more general function of arity 4.

\begin{lemma}[{\cite[Part of Lemma~2.40]{cai_holographic_2017}}]\label{lem:interpolate_equality4}
 Suppose $\cF\sse\allf$ is finite and contains a function $f$ of arity 4 with matrix
 \[
  \begin{pmatrix} a&0&0&b \\ 0&0&0&0 \\ 0&0&0&0 \\ c&0&0&d \end{pmatrix}
 \]
 where $\smm{a&b\\c&d}$ has full rank.
 Then $\mathsf{Pl\text{-}Holant}(\{\EQ_4\}\cup\cF) \leq_T \mathsf{Pl\text{-}Holant}(\cF)$.
\end{lemma}

{The lemma can of course also be used in the non-planar setting.}

Functions in $\cE$ are often called `generalised equality functions'.
Recall from Section~\ref{s:csp} that $\csp_2$ is a counting CSP in which each variable has to appear an even number of times.

\begin{lemma}[{\cite[Lemma~5.2]{cai_dichotomy_2017}}]\label{lem:generalised_equality4}
 Suppose $\cF\sse\allf$ is finite and contains a generalised equality function $f$ of arity 4, then
 \[
  \Holp{\cF}\equiv_T\csp_2(\cF).
 \]
\end{lemma}

\begin{rem}
 Note that the statement of Lemma~5.2 in \cite{cai_dichotomy_2017} has `$\hol^c$' instead of plain `$\hol$', but the proof does not use pinning functions, so this is presumably a typo.
\end{rem}

\section{The quantum state perspective}
\label{s:quantum_states}

In Section \ref{s:vector_perspective}, we introduced the idea of considering functions as complex vectors.
This perspective is not only useful for proving Valiant's Holant Theorem (which is at the heart of the theory of holant problems), it also gives a connection to the theory of quantum computation.

In quantum computation and quantum information theory, the basic system of interest is a \emph{qubit} (quantum bit), which takes the place of the usual bit in standard computer science.
The state of a qubit is described by a vector\footnote{Strictly speaking, vectors only describe \emph{pure} quantum states: there are also \emph{mixed} states, which need to be described differently; but we do not consider those here.} in the two-dimensional complex Hilbert space $\CC^2$.
State spaces compose by tensor product, i.e.\ the state of $n$ qubits is described by a vector in $\left(\CC^2\right)\t{n}$, which is isomorphic to $\CC^{2^n}$.
Here, $\otimes$ denotes the tensor product of Hilbert spaces and tensor powers are defined analogously to tensor powers of matrices in Section~\ref{s:vector_perspective}.
Thus, the vector associated with an $n$-ary function can be considered to be a quantum state of $n$ qubits.
The vectors describing quantum states are usually required to have norm 1, but for the methods used here, multiplication by a non-zero complex number does not make a difference, so we can work with states having arbitrary norms.

Let $\{\ket{0},\ket{1}\}$ be an orthonormal basis for $\CC^2$.
This is usually called the \emph{computational basis}.
The induced basis on $\left(\CC^2\right)\t{n}$ is labelled by $\{\ket{\bx}\}_{\bx\in\{0,1\}^n}$ as a short-hand, e.g.\ we write $\ket{00\ldots 0}$ instead of $\ket{0}\otimes\ket{0}\otimes\ldots\otimes\ket{0}$.
This is exactly the same as the basis introduced in Section \ref{s:vector_perspective}.

Holographic transformations also have a natural interpretation in quantum information theory: going from an $n$-qubit state $\ket{f}$ to $M\t{n}\ket{f}$, where $M$ is some invertible 2 by 2 matrix, is a `stochastic local operation with classical communication' (SLOCC) \cite{bennett_exact_2000,dur_three_2000}.
These are physical operations that can be applied locally (without needing access to more than one qubit at a time) using classical (i.e.\ non-quantum) communication between the sites where the different qubits are held, and which succeed with non-zero probability.
{If two quantum states are equivalent under SLOCC, they can be used for the same quantum information tasks, albeit potentially with different probabilities of success.
Two $n$-qubit states $\ket{\psi}$ and $\ket{\phi}$ are equivalent under SLOCC if and only if there exist invertible complex 2 by 2 matrices $M_1,M_2,\ldots, M_n$ such that $\ket{\psi} = (M_1\otimes M_2\otimes \ldots \otimes M_n) \ket{\phi}$ \cite[Section~II.A]{dur_three_2000}.
In particular, SLOCC operations do not need to be symmetric under permutations.}

From now on, we will sometimes mix standard holant terminology (or notation) and quantum terminology (or notation).

\subsection{Entanglement and its classification}
\label{s:entanglement}

One major difference between quantum theory and preceding theories of physics (known as `classical physics') is the possibility of \emph{entanglement} in states of multiple systems.

\begin{definition}\label{def:entanglement}
 A state of multiple systems is \emph{entangled} if it cannot be written as a tensor product of states of individual systems.
\end{definition}

\begin{ex}
 In the case of two qubits,
 \begin{equation}
  \ket{00}+\ket{01}+\ket{10}+\ket{11}
 \end{equation}
 is a product state -- it can be written as $(\ket{0}+\ket{1})\otimes(\ket{0}+\ket{1})$.
 On the other hand, consider the state
 \begin{equation}
  \ket{00}+\ket{11}.
 \end{equation}
 It is impossible to find single-qubit states $\ket{f},\ket{g}\in\CC^2$ such that $\ket{f}\otimes\ket{g} = \ket{00}+\ket{11}$.
 {In function notation, this can be seen by noting that if $h(x,y)=f(x)g(y)$, then 
 \[
  h(0,0)h(1,1)-h(0,1)h(1,0) = f(0)g(0)f(1)g(1) - f(0)g(1)f(1)g(0) = 0,
 \]
 whereas for $h(x,y)=\EQ_2(x,y)$, we have $h(0,0)h(1,1)-h(0,1)h(1,0) = 1$.}
 Thus, $\ket{00}+\ket{11}$ is entangled.
\end{ex}

Where a state involves more than two systems, it is possible for some of the systems to be entangled with each other and for other systems to be in a product state with respect to the former.
We sometimes use the term \emph{genuinely entangled state} to refer to a state in which no subsystem is in a product state with the others.
The term \emph{multipartite entanglement} refers to entangled states in which more than two qubits are mutually entangled.

{Under the bijection between functions in vectors described in Section~\ref{s:vector_perspective}, a state vector is entangled if and only if the corresponding function is non-degenerate.
A state is genuinely entangled if and only if the corresponding function is non-decomposable.
Finally, a state has multipartite entanglement if and only if the corresponding function has a non-decomposable factor of arity at least three.
In other words, entangled states correspond to functions in $\allf\setminus\ang{\cU}$ and multipartite entangled states correspond to functions in $\allf\setminus\ang{\cT}$.}

Entanglement is an important resource in quantum computation, where it has been shown that quantum speedups are impossible without the presence of unboundedly growing amounts of entanglement \cite{jozsa_role_2003}.
Similarly, it is a resource in quantum information theory \cite{nielsen_quantum_2010}, featuring in protocols such as quantum teleportation \cite{bennett_teleporting_1993} and quantum key distribution \cite{ekert_quantum_1991}.
Many quantum information protocols have the property that two quantum states can be used to perform the same task if one can be transformed into the other by SLOCC, motivating the following equivalence relation.

\begin{definition}\label{def:SLOCC-equivalence}
 Two $n$-qubit states are \emph{equivalent under SLOCC} if one can be transformed into the other using SLOCC.
 More formally: suppose $\ket{f},\ket{g}\in(\CC^2)\t{n}$ are two $n$-qubit states. Then $\ket{f} \sim_{SLOCC} \ket{g}$ if and only if there exist invertible complex 2 by 2 matrices $M_1,M_2,\ldots, M_n$ such that $\left(M_1\otimes M_2\otimes \ldots \otimes M_n\right) \ket{f} = \ket{g}$.
\end{definition}
The equivalence classes of this relation are called \emph{entanglement classes} or \emph{SLOCC classes}.
{This definition is justified because SLOCC does not affect the decomposition of a state into tensor factors.
To see this, suppose $\ket{f}, \ket{g}$ are two $n$-qubit states satisfying Definition~\ref{def:SLOCC-equivalence}.
Furthermore, suppose $\ket{f}$ decomposes as $\ket{f_1^k}\otimes\ket{f_{k+1}^n}$ for some $k$ with $1\leq k<n$, where $\ket{f_1^k}$ is a state on the first $k$ qubits and $\ket{f_{k+1}^n}$ is a state on the last $(n-k)$ qubits.
Then
\begin{align*}
 \ket{g}
 &= \left(M_1\otimes M_2\otimes \ldots \otimes M_n\right) \ket{f} \\
 &= \left(M_1\otimes M_2\otimes \ldots \otimes M_n\right) \left(\ket{f_1^k}\otimes\ket{f_{k+1}^n}\right) \\
 &= \left((M_1\otimes\ldots\otimes M_k)\ket{f_1^k}\right) \otimes \left((M_{k+1}\otimes\ldots\otimes M_n)\ket{f_{k+1}^n}\right),
\end{align*}
so $\ket{g}$ decomposes in the same way as $\ket{f}$.
Since the matrices $M_1,M_2,\ldots, M_n$ are invertible, the converse also holds.
Hence $\ket{g}$ can be decomposed as a tensor product according to some partition if and only if $\ket{f}$ can be decomposed according to the same partition.
Therefore SLOCC does not affect entanglement.
Due to the correspondences between vectors and functions outlined in the first part of the section, the same holds for holographic transformations.}

For two qubits, there is only one class of entangled states: all entangled two-qubit states are equivalent to $\ket{00}+\ket{11}$ (the vector corresponding to $\EQ_2$) under SLOCC.
For three qubits, there are two classes of genuinely entangled states, the GHZ class and the $W$ class \cite{dur_three_2000}.
The former contains states that are equivalent under SLOCC to the GHZ state:
\begin{equation}
 \ket{\GHZ} := \frac{1}{\sqrt{2}}(\ket{000}+\ket{111}),
\end{equation}
the latter those equivalent to the $W$ state:
\begin{equation}
 \ket{W} := \frac{1}{\sqrt{3}}(\ket{001}+\ket{010}+\ket{100}).
\end{equation}
Note that, up to scalar factor, $\ket{\GHZ}$ is the vector corresponding to $\EQ_3$ and $\ket{W}$ is the vector corresponding to $\ONE_3$.

We say that a function has GHZ type if it is equivalent to the GHZ state under local holographic transformations and that a function has $W$ type if it is equivalent to the GHZ state under local holographic transformations.
Local holographic transformations include invertible scaling, so non-zero scalar factors do not affect the entanglement classification.
In the holant literature, GHZ-type functions have been called the \emph{generic case} and $W$-type functions have been called the \emph{double-root case}, cf.\ \cite[Section~3]{cai_holant_2012}.

The two types of genuinely entangled ternary functions can be distinguished as follows.

\begin{lemma}[\cite{li_simple_2006}]\label{lem:li}
 Let $f$ be a ternary function and write $f_{klm}:=f(k,l,m)$ for all $k,\ell,m \in\{0,1\}$.
 Then $f$ has GHZ type if the following polynomial in the function values is non-zero:
 \begin{equation}\label{eq:GHZ_polynomial}
  (f_{000}f_{111} - f_{010}f_{101} + f_{001}f_{110} - f_{011}f_{100})^2 - 4(f_{010}f_{100}-f_{000}f_{110})(f_{011}f_{101}-f_{001}f_{111}).
 \end{equation}
 The function $f$ has $W$ type if the polynomial \eqref{eq:GHZ_polynomial} is zero, and furthermore each of the following three expressions is satisfied:
 \begin{align}
  (f_{000}f_{011}\neq f_{001}f_{010}) &\vee (f_{101}f_{110}\neq f_{100}f_{111}) \label{eq:W1} \\
  (f_{001}f_{100}\neq f_{000}f_{101}) &\vee (f_{011}f_{110}\neq f_{010}f_{111}) \label{eq:W2} \\
  (f_{011}f_{101}\neq f_{001}f_{111}) &\vee (f_{010}f_{100}\neq f_{000}f_{110}). \label{eq:W3}
 \end{align}
 If the polynomial \eqref{eq:GHZ_polynomial} is zero and at least one of the three expressions evaluates to false, $f$ is decomposable.
 In fact, for any decomposable $f$, at least two of the expressions are false.
\end{lemma}

The above lemma can be specialised to symmetric functions as follows.

\begin{lemma}\label{lem:li_symmetric}
 Let $f$ be a ternary symmetric function and write $f=[f_0,f_1,f_2,f_3]$.
 Then $f$ has GHZ type if the following polynomial in the function values is non-zero:
 \begin{equation}\label{eq:GHZ_polynomial_symmetric}
  (f_0 f_3 - f_1 f_2)^2 - 4(f_1^2 - f_0 f_2)(f_2^2 - f_1 f_3) \neq 0.
 \end{equation}
 The function $f$ has $W$ type if the above polynomial is zero and furthermore
 \[
  (f_1^2 \neq f_0 f_2) \vee (f_2^2 \neq f_1 f_3).
 \]
 If the polynomial in \eqref{eq:GHZ_polynomial_symmetric} is zero and the above expression evaluates to false, $f$ is decomposable.
\end{lemma}

For joint states of more than three qubits, there are infinitely many SLOCC classes.
It is possible to partition these into families which share similar properties.
Yet, so far, there is no consensus on how to partition the classes: there are different schemes for partitioning even the four-qubit entanglement classes, yielding different families \cite{verstraete_four_2002,lamata_inductive_2007,backens_inductive_2016}.

A \emph{generalised GHZ state} is a vector corresponding (up to scalar factor) to $\EQ_k$ for some integer $k\geq 3$.
A \emph{generalised $W$ state} is a vector corresponding (up to scalar factor) to $\ONE_k$ for some integer $k\geq 3$.

\subsection{The existing results in the quantum picture}
\label{s:existing_quantum}

{Using the correspondence between vectors and functions,} several of the existing dichotomies have straightforward descriptions in the quantum picture.
The tractable cases of Theorem~\ref{thm:Holant-star} can be described as follows:

\begin{itemize}
 {
 \item The case $\cF\subseteq\avg{\cT}$ corresponds to vectors with no multipartite entanglement.
 Unbounded multipartite entanglement is needed for quantum computation to offer any advantage over non-quantum computation \cite{jozsa_role_2003}, so it makes sense that its absence would lead to a holant problem that is polynomial-time computable.
 \item In the cases $\cF\subseteq\avg{O\circ\cE}$ or $\cF\subseteq\avg{K\circ\cE}$, assuming $\cF\nsubseteq\avg{\cT}$, there is GHZ-type multipartite entanglement.
 To see this, note first that if $f\in\ang{\cE}\setminus\avg{\cT}$ is non-decomposable, then $f$ must have arity at least 3 and be non-zero on exactly two \new{complementary} inputs.
 Suppose $f$ has arity $n$, and let $\ba\in\new{\{0,1\}^n}$ be such that $f(\ba)\neq 0$. \new{Without loss of generality, assume $a_1=0$ (if $a_1=1$, replace $\ba$ by $\bar{\ba}$)}.
 Then
 \[
  \ket{f} = \left(\pmm{f(\ba)&0\\0&f(\bar{\ba})} \otimes \bigotimes_{k=2}^n X^{a_n}\right) \ket{\EQ_n},
 \]
 where $X^0$ is the identity matrix; hence $\ket{f}\sim_{SLOCC}\ket{\EQ_n}$.
 Further holographic transformations do not affect the equivalence under SLOCC.
 As formally shown in \cite[Lemma~46]{backens_holant_2018}, the sets $\avg{O\circ\cE}$ and $\avg{K\circ\cE}$ are already closed under taking gadgets, so every non-decomposable function of arity $n$ is SLOCC-equivalent to $\EQ_n$.
 This means it is impossible to realise $W$-type multipartite entanglement \new{from $\cF$} via gadgets, which again indicates these cases are insufficient to describe full quantum computation.
 \item Finally, in the case $\cF\subseteq\avg{K\circ\cM}$ or $\cF\subseteq\avg{KX\circ\cM}$, again assuming $\cF\nsubseteq\avg{\cT}$,  there is $W$-type multipartite entanglement.
 To see this, note first that if $f\in\ang{\cM}\setminus\ang{\cT}$ is non-decomposable, then $f$ must have arity at least 3.
 Suppose $n:=\ari(f)\geq 3$ and there exists an index $k\in[n]$ such that $f$ is 0 on the bit string that has a 1 only on position $k$ and zeroes elsewhere.
 Then by the definition of $\cM$, the function $f$ is 0 whenever input $k$ is non-zero; hence $f$ can be decomposed as $f(\vc{x}{n}) = \dl_0(x_k) f'(\vc{x}{k-1},x_{k+1}\zd x_{n})$.
 Thus, any non-decomposable $f\in\ang{\cM}$ has support on all bit strings of Hamming weight exactly 1.
 Now suppose $f\in\cM$ is a non-decomposable function of arity $n$, then if $f(0\ldots 0)=0$, we have
 \[
  \ket{f} = \left(\bigotimes_{k=1}^n \pmm{1&0\\0&f(\be_k)}\right) \ket{\ONE_n},
 \]
 where $\be_k$ for $1\leq k\leq n$ is the $n$-bit string that contains a single 1 in position $k$ and zeroes elsewhere.
 If $f(0\ldots 0)\neq 0$, then
 \[
  \ket{f} = f(0\ldots 0) \left(\bigotimes_{k=1}^n \pmm{1&\frac{1}{n}\\0&\frac{f(\be_k)}{f(0\ldots 0)}}\right) \ket{\ONE_n}.
 \]
 In both cases $\ket{f}\sim_{SLOCC}\ket{\ONE_n}$, i.e.\ any non-decomposable function in $\cM$ has $W$-type entanglement.
 Further holographic transformations do not affect SLOCC-equivalence.
 Again, it was shown in \cite[Lemma~46]{backens_holant_2018} that  $\ang{K\circ\cM}$ and $\ang{KX\circ\cM}$ are closed under taking gadgets, so every non-decomposable function of arity $n$ they contain is SLOCC-equivalent to $\ONE_n$. In particular, it is impossible to realise GHZ-type multipartite entanglement.
 }
\end{itemize}

The family of affine functions (cf.\ Definition~\ref{dfn:affine_function}) also has a natural description in quantum information theory: the quantum states corresponding to affine functions are known as \emph{stabiliser states} \cite{dehaene_clifford_2003}.
These states and the associated operations play an important role in the context of quantum error-correcting codes \cite{gottesman_heisenberg_1998} and are thus at the core of most attempts to build large-scale quantum computers \cite{devitt_quantum_2013}.
The fragment of quantum theory consisting of stabiliser states and operations that preserve the set of stabiliser states can be efficiently simulated on a classical computer \cite{gottesman_heisenberg_1998}; this result is known as the Gottesman-Knill theorem.
The connection between affine functions and stabiliser quantum mechanics has recently been independently noted and explored in \cite{cai_clifford_2018}.

The examples in this section show that holant problems and quantum information theory are linked not only by quantum algorithms being an inspiration for holographic ones: instead, many of the known tractable set of functions of various holant problems correspond to state sets that are of independent interest in quantum computation and quantum information theory.
The one exception are local affine functions, which seem not to have been described in the quantum literature, possibly because this set does not contain any interesting unitary operations.

The restriction to algebraic numbers is not a problem from the quantum perspective, not even when considering the question of universal quantum computation: there exist (approximately) universal sets of quantum operations where each operation can be described using algebraic complex coefficients.
One such example is the Clifford+T gate set \cite{boykin_universal_1999,giles_exact_2013}, which is generated by the operators
\[
 \pmm{1&0\\0&e^{i\pi/4}}, \qquad
 \frac{1}{\sqrt{2}}\pmm{1&1\\1&-1}, \quad\text{and}\quad
 \pmm{1&0&0&0\\0&1&0&0\\0&0&0&1\\0&0&1&0}.
\]

{\subsection{Affine functions, holographic transformations, and entanglement}

The lemmas in this section, like those in Section~\ref{s:linear-algebra} are straightforward.
They will be useful in the complexity classification proofs later.

The first results is well known in quantum theory.
It is also} closely related to the theory of functional clones \cite{bulatov_expressibility_2013,bulatov_functional_2017} and holant clones \cite{backens_holant_2018}, though we will not introduce the full formalisms of those frameworks here.
{Instead,} we give a proof in the language of gadgets.

\begin{lemma}\label{lem:affine_closed}
 The set of affine functions is closed under taking gadgets, i.e.\ $S(\cA)=\cA$.
\end{lemma}
\begin{proof}
 Suppose $G=(V,E,E')$ is a graph with vertices $V$, (normal) edges $E$, and dangling edges $E'=\{\vc{e}{n}\}$, where $E\cap E'=\emptyset$.
 Let $\Gamma=(\cA,G,\pi)$ be a gadget with effective function
 \[
  g_\Gamma(\vc{y}{n}) = \sum_{\sigma:E\to\{0,1\}} \prod_{v\in V} f_v(\hat{\sigma}|_{E(v)}),
 \]
 where $\hat{\sigma}$ is the extension of $\sigma$ to domain $E\cup E'$ which satisfies $\hat{\sigma}(e_k)=y_k$ for all $k\in[n]$, and $\hat{\sigma}|_{E(v)}$ is the restriction of $\hat{\sigma}$ to edges (both normal and dangling) which are incident on $v$.
 
 We prove $g_\Gamma\in\cA$ by induction on the number of normal edges $m:=\abs{E}$.
 
 The base case $m=0$ implies that all edges are dangling and $g_\Gamma = \bigotimes_{v\in V} f_v$.
 Then by associativity of $\otimes$ and by repeated application of Lemma~\ref{lem:cai_clifford}~(1), we have $g_\Gamma\in\cA$.
 
 For the inductive step, assume the desired property holds if there are $m$ normal edges.
 Consider a gadget $\Gamma=(\cA,G,\pi)$ with $m+1$ normal edges and $n$ dangling edges.
 Pick some $e=\{u,v\}\in E$ and `cut it', i.e.\ replace it by two dangling edges $e_{n+1},e_{n+2}$, where $e_{n+1}$ is incident on $u$ and $e_{n+2}$ is incident on $v$.
 Let $\bar{E}=E\setminus\{e\}$ and let $E''=E'\cup\{e_{n+1},e_{n+2}\}$.
 The resulting graph is $G'=(V,\bar{E},E'')$.
 Since $G'$ has the same vertices as $G$ and each vertex has the same degree in both graphs, $\Gamma'=(\cA,G',\pi)$ is a valid gadget, where $\pi$ is the same map as before.
 Then
 \[
  g_{\Gamma'}(\vc{y}{n+2}) = \sum_{\sigma:\bar{E}\to\{0,1\}} \prod_{v\in V} f_v(\hat{\sigma}'|_{E(v)}),
 \]
 where $\hat{\sigma}'$ is the extension of $\sigma$ to domain $\bar{E}\cup E''$ which satisfies $\hat{\sigma}'(e_k)=y_k$ for all $k\in[n+2]$.
 Now $\Gamma'$ is a gadget with $m$ normal edges, so by the inductive hypothesis, $g_{\Gamma'}(\vc{y}{n+2})\in\cA$.
 
 But $g_\Gamma(\vc{y}{n}) = \sum_{y_{n+1}\in\{0,1\}} g_{\Gamma'}(\vc{y}{n},y_{n+1},y_{n+1})$, i.e.\ $g_\Gamma = (g_{\Gamma'}^{y_{n+2}=y_{n+1}})^{y_{n+1}=*}$.
 Thus, by Lemma~\ref{lem:cai_clifford}~(3) and (4), we have $g_\Gamma\in\cA$.
\end{proof}

\begin{lemma}\label{lem:cS_cA-group}
 The set $\cS_\cA := \{L\in\GL\mid f_L\in\cA\}$ is a group under matrix multiplication.
\end{lemma}
\begin{proof}
 {Closure under matrix multiplication follows directly from Lemma~\ref{lem:affine_closed}.
 The identity matrix corresponds to $\EQ_2$, which is affine, so $\cS_\cA$ contains the identity.
 For closure under inverse, note that by Definition~\ref{dfn:affine_function} any matrix $A\in\cS_\cA$ corresponds to a function
 \[
  f_A(x,y) = ci^{\ell(x,y)}(-1)^{q(x,y)}\chi(x,y),
 \]
 where $c\in\AAnz$ by invertibility of $A$, $\ell,q:\{0,1\}^2\to\{0,1\}$ with $\ell$ being linear and $q$ being quadratic, and $\chi\in\{[1,1,1],\,[1,0,1],\,[0,1,0]\}$ is the indicator function for an affine support.
 (With support on some other affine subspace of $\{0,1\}^2$, $A$ could not be invertible.)
 Constant terms in $\ell$ or $q$ can be absorbed into $c$, and for Boolean variables, $x^2=x$, so non--cross-terms from $q$ can be absorbed into $\ell$.
 Thus, without loss of generality, $\ell(x,y)=\ld x + \mu y$ and $q(x,y)= \kappa xy$ for some $\ld,\mu\in\{0,1,2,3\}$ and $\kappa\in\{0,1\}$.
 Let $P:=\smm{1&0\\0&i}$.
 
 Now if $\chi=[1,1,1]$, then
   \[
    A = c\pmm{1 & i^{\mu} \\ i^{\ld} & (-1)^\kappa i^{\ld+\mu}} = c P^\ld \pmm{1&1\\1&(-1)^\kappa} P^\mu.
   \]
 For $A$ to be invertible, $\kappa$ must be 1.
 But then $A^{-1} = c^{-1} P^{4-\mu} \smm{1&1\\1&-1} P^{4-\ld} \in \cS_\cA$.
 On the other hand, if $\chi=[1,0,1]$ or $\chi=[0,1,0]$, then
   \[
    A = c\pmm{1 & 0 \\ 0 & (-1)^\kappa i^{\ld+\mu}} \qquad\text{or}\qquad
    A = c\pmm{0 & i^{\mu} \\ i^{\ld} & 0}.
   \]
 Again, in both cases $A^{-1}$ has the same form as $A$, so $A^{-1}\in\cS_\cA$.
 Hence $\cS_\cA$ is a group.}
\end{proof}

Up to scaling, the unary elements of $\cA$ are
 \[
  \dl_0:=[1,0],\quad \dl_1:=[0,1],\quad \dl_+:=[1,1],\quad \dl_-:=[1,-1],\quad \dl_i:=[1,i],\quad\text{and}\quad \dl_{-i}:=[1,-i].
 \]
We say the pairs $\{\dl_0,\dl_1\}$, $\{\dl_+,\dl_-\}$, and $\{\dl_i,\dl_{-i}\}$ are orthogonal pairs\footnote{This is because the two corresponding vectors are orthogonal \new{under the complex inner product}.}, any other pair of distinct unary functions $u,v\in\cA$ is called non-orthogonal.
{It is straightforward to see that if $u,u^\perp$ are an orthogonal pair and $v\in\cA$ is a unary function that is not a scaling of $u$ or $u^\perp$, then $v(x)=\alpha\cdot u(x) + \beta\cdot u^\perp(x)$ where $\alpha^4=\beta^4$.}

\begin{lemma}\label{lem:cS-cA}
 Suppose $M\in\cS=\left\{L\in\GL \,\middle|\, L^T\circ\{\EQ_2,\dl_0,\dl_1\}\sse\cA \right\}$ and $M^T\circ\dl_+,M^T\circ\dl_-\in\cA$.
 Then $M\in\cS_\cA=\{L\in\GL\mid f_L\in\cA\}$.
\end{lemma}
\begin{proof}
 {The property $M\in\cS$ implies there are $u,v\in\{\dl_0,\dl_1,\dl_+,\dl_-,\dl_i,\dl_{-i}\}$ and $\ld,\mu\in\AAnz$ such that $M^T\circ\dl_0=\ld\cdot u$ and $M^T\circ\dl_1=\mu\cdot v$.
 But $\dl_\pm(x)=\dl_0(x)\pm\dl_1(x)$, so by linearity $(M^T\circ\dl_\pm)(x) = \ld\cdot u(x) \pm\mu\cdot v(x)$.
 
 First, suppose $u$ and $v$ are orthogonal. 
 Then $M^T = M'\smm{\ld&0\\0&\mu}$ where
 \[
  M' \in \left\{ \pmm{1&0\\0&1}, \pmm{0&1\\1&0}, \pmm{1&1\\1&-1}, \pmm{1&1\\-1&1}, \pmm{1&1\\i&-i} , \pmm{1&1\\-i&i} \right\},
 \]
 since this set contains all matrices that map $\dl_0$ and $\dl_1$ to a pair of orthogonal functions in $\cA$, including permutations.
 It is straightforward to check that if $u$ and $v$ are orthogonal, then $M^T\circ\dl_\pm\in\cA$ if and only if $\ld^4=\mu^4$.
 But then $M^T$ is a product of two matrices corresponding to functions in $\cA$, so $M\in\cS_\cA$ by Lemma~\ref{lem:affine_closed}.
 
 Now suppose $u$ and $v$ are not orthogonal; denote by $u^\perp$ the function which forms an orthogonal pair with $u$.
 Then there exist $\alpha,\beta\in\AAnz$ with $\alpha^4=\beta^4$ such that $v(x) = \alpha\cdot u(x) + \beta\cdot u^\perp(x)$.
 Thus
 \[
  (M^T\circ\dl_\pm)(x) = (\ld\pm\mu\alpha)\cdot u(x) \pm\mu\beta\cdot u^\perp(x).
 \]
 These two functions are in $\cA$ if and only if both $(\ld+\mu\alpha)^4=\mu^4\beta^4$ and $(\ld-\mu\alpha)^4=\mu^4\beta^4$.
 That means
 \[
  (\ld+\mu\alpha)^4 = (\ld-\mu\alpha)^4
  \quad\Longleftrightarrow\quad \ld\mu\alpha(\ld^2 + \mu^2\alpha^2) = 0
  \quad\Longleftrightarrow\quad \ld = \pm i \mu\alpha
 \]
 Then $\mu^4\beta^4 = (1+i)^4\mu^4\alpha^4$, so since all of the numbers are non-zero, we have $\beta^4 = -4\alpha^4$.
 This contradicts the assumption $\alpha^4=\beta^4$.
 Therefore this case cannot happen and we always have $M\in\cS_\cA$ by the previous case.}
\end{proof}

{The following makes more precise some of the arguments about entanglement types in Section~\ref{s:existing_quantum} in the context of ternary functions, and extends the argument to functions in $\cA$.}

\begin{lemma}\label{lem:family-types}
 Suppose $f\in\ang{\cE}$ is a non-decomposable ternary function, then $f$ has GHZ type.
 Similarly, suppose $g\in\ang{\cM}$ is a non-decomposable ternary function, then $g$ has $W$ type.
 Finally, suppose $h\in\cA$ is a non-decomposable ternary function, then $h$ has GHZ type.
\end{lemma}
\begin{proof}
 Suppose $f\in\ang{\cE}$ is a non-decomposable ternary function.
 Non-decomposability implies $f\in\cE$, so there exists $\ba\in\{0,1\}^3$ such that $f(\bx)=0$ unless $\bx\in\{\ba,\bar{\ba}\}$.
 Since $f$ is non-decomposable, $f(\ba)$ and $f(\bar{\ba})$ must both be non-zero.
 We can thus find matrices $A,B,C\in\{I,X\}$ such that $(A\otimes B\otimes C)\ket{f}= f_{\ba}\ket{000}+f_{\bar{\ba}}\ket{111}$, which is clearly a GHZ-type state.
 But this is an SLOCC operation, which does not affect the entanglement class, so $f$ has GHZ type.
 
 Now suppose $g\in\ang{\cM}$ is a non-decomposable ternary function.
 Non-decomposability implies $g\in\cM$, hence $g(\bx)=0$ whenever $\abs{\bx}>1$.
 The polynomial in \eqref{eq:GHZ_polynomial} becomes
 \begin{multline*}
  (g_{000}g_{111} - g_{010}g_{101} + g_{001}g_{110} - g_{011}g_{100})^2 - 4(g_{010}g_{100}-g_{000}g_{110})(g_{011}g_{101}-g_{001}g_{111}) \\
  = (0 - 0 + 0 - 0)^2 - 4(g_{010}g_{100}-0)(0-0)
  = 0.
 \end{multline*}
 Yet $g$ is non-decomposable by assumption, therefore Lemma~\ref{lem:li} implies that $g$ must have $W$~type.
 
 Finally, suppose $h\in\cA$ is a non-decomposable ternary function.
 Then by Definition~\ref{dfn:affine_function}, $h(\bx) = c i^{l(\bx)} (-1)^{q(\bx)} \chi_{A\bx=\bb}(\bx)$ where $c\in\AA\setminus\{0\}$ is a constant, $l$ is a linear Boolean function, $q$ is a quadratic Boolean function, and $\chi$ is a 0-1 valued indicator function for an affine subspace of $\{0,1\}^3$.
 
 By~\cite{montanaro_hadamard_2006}, for any function $h'\in\allf_n$ there exist matrices $\vc{M}{n}\in\{I,H\}$, where $H=\smm{1&1\\1&-1}$ is the Hadamard matrix, such that $(M_1\otimes\ldots\otimes M_n)\ket{h'}$ is everywhere non-zero.
 Both $I$ and $H$ correspond to affine functions, so by Lemmas~\ref{lem:hc_gadget} and~\ref{lem:affine_closed}, if $h'\in\cA$ then the function corresponding to $(M_1\otimes\ldots\otimes M_n)\ket{h'}$ is also in $\cA$.
 Hence, since SLOCC operations do not affect the entanglement class, we may assume without loss of generality that $h$ has full support by replacing it with the function transformed according to \cite{montanaro_hadamard_2006} if necessary.
 Then $\chi$ is the constant-1 function and can be ignored.
 {Now, a SLOCC transformation by $P := \smm{1&0\\0&i}$ on argument $x_k$ contributes a factor $i^{x_k}$ to the overall function.
 Thus, by such transformations, we can make $l$ trivial and remove all terms of the form $x_k^2$ from $q$ without changing the entanglement.
 It thus suffices to consider the function $h'(x_1,x_2,x_3) = (-1)^{\gamma_{12} x_1 x_2 + \gamma_{13} x_1 x_3 + \gamma_{23} x_2 x_3}$, where $\gamma_{12},\gamma_{13},\gamma_{23}\in\{0,1\}$.
 Then the first term of the polynomial in \eqref{eq:GHZ_polynomial} becomes
 \begin{multline*}
  (h_{000}h_{111} - h_{010}h_{101} + h_{001}h_{110} - h_{011}h_{100})^2  \\
  = c^4 \Big((-1)^{\gamma_{12}+\gamma_{13}+\gamma_{23}} - (-1)^{\gamma_{13}} + (-1)^{\gamma_{12}} - (-1)^{\gamma_{23}}\Big)^2
 \end{multline*}
 The second term becomes
 \begin{align*}
  -4(h_{010}h_{100}-h_{000}h_{110})(h_{011}h_{101} &-h_{001}h_{111}) \\
  &=- 4c^4 \Big(1-(-1)^{\gamma_{12}}\Big) \Big((-1)^{\gamma_{13}+\gamma_{23}}-(-1)^{\gamma_{12}+\gamma_{13}+\gamma_{23}}\Big) \\
  &=- 4c^4 (-1)^{\gamma_{13}+\gamma_{23}} \Big(1-(-1)^{\gamma_{12}}\Big)^2
 \end{align*}
 Thus, if $\gamma_{12}=0$, the polynomial in \eqref{eq:GHZ_polynomial} is equal to
 \[
  c^4 \Big((-1)^{\gamma_{13}+\gamma_{23}} - (-1)^{\gamma_{13}} + 1 - (-1)^{\gamma_{23}}\Big)^2,
 \]
 which is $16c^4$ if $\gamma_{13}=\gamma_{23}=1$, and 0 otherwise.
 If $\gamma_{12}=1$, the polynomial becomes
 \[
  c^4 \Big(-(-1)^{\gamma_{13}+\gamma_{23}} - (-1)^{\gamma_{13}} - 1 - (-1)^{\gamma_{23}}\Big)^2 - 16 c^4 (-1)^{\gamma_{13}+\gamma_{23}},
 \]
 which is 0 if $\gamma_{13}=\gamma_{23}=0$, and non-zero otherwise.
 Hence the function has GHZ-type if and only if $\gamma_{12}+\gamma_{13}+\gamma_{23}\geq 2$.
 
 It remains to see what happens if $\gamma_{12}+\gamma_{13}+\gamma_{23} < 2$.
 Now, the condition \eqref{eq:W1}, $(h_{000}h_{011}\neq h_{001}h_{010}) \vee (h_{101}h_{110}\neq h_{100}h_{111})$, becomes
 \[
  \Big( (-1)^{\gamma_{23}} \neq 1 \Big)
  \vee \Big((-1)^{\gamma_{12}+\gamma_{13}} \neq (-1)^{\gamma_{12}+\gamma_{13}+\gamma_{23}} \Big),
 \]
 which reduces to the single inequality $\gamma_{23}\neq 0$.
 Similarly, \eqref{eq:W2} becomes $\gamma_{13}\neq 0$ and \eqref{eq:W3} becomes $\gamma_{12}\neq 0$.
 Hence $h$ either has GHZ type or it decomposes.} 
 Thus, any non-decomposable ternary function in $\cA$ has GHZ type.
\end{proof}

\section{\textsf{Holant}\texorpdfstring{\textsuperscript{+}}{\textasciicircum +}}
\label{s:Holant_plus}

Before deriving the dichotomy for $\hol^c$, we consider a new family of holant problems, called $\hol^+$, which fits between conservative holant problems and $\hol^c$: 
It has four freely available functions, which are all unary and include the pinning functions.
Using results from quantum information theory, these four functions can be shown to be sufficient for constructing the gadgets required to apply the dichotomies in Section \ref{s:results_ternary_symmetric}.
Formally, for any finite $\cF\sse\allf$:
\begin{equation}
 \Holp[+]{\cF} := \Holp{\cF\cup\{\dl_0,\dl_1,\dl_+,\dl_-\}}.
\end{equation}

Note that the vectors $\ket{0}$ and $\ket{1}$ corresponding to $\dl_0$ and $\dl_1$ are orthogonal to each other.
{Similarly,} the vectors $\ket{+}$ and $\ket{-}$ corresponding, {up to scalar factor,} to $\dl_+$ and $\dl_-$ are orthogonal to each other.
In quantum theory, the set $\{\ket{+},\ket{-}\}$ is known as the \emph{Hadamard basis} of $\CC^2$, since these vectors are related to the computational basis vectors by a Hadamard transformation: $\{\ket{+},\ket{-}\}\doteq H\circ\{\ket{0},\ket{1}\}$, where $H = \frac{1}{\sqrt{2}}\left(\begin{smallmatrix}1&1\\1&-1\end{smallmatrix}\right)$.
{Hence $\ket{+}=\frac{1}{\sqrt{2}}\left(\ket{0}+\ket{1}\right)$ and $\ket{-}=\frac{1}{\sqrt{2}}\left(\ket{0}-\ket{1}\right)$, i.e.\ the Hadamard basis vectors differ from the vectors corresponding to $\delta_+$ and $\delta_-$ by a factor of $\frac{1}{\sqrt{2}}$, which does not affect any of the following arguments.

In the next subsection, we first state and extend a result from quantum theory that is used in the later proofs.
Specifically, we prove that if $\cF\nsubseteq\ang{\cT}$, then  $S(\cF\cup\{\dl_0,\dl_1,\dl_+,\dl_-\})$ contains a non-decomposable ternary function.
It is vital for the proof to have all four unary functions available, e.g.\ if $\cF=\{\EQ_4\}$, it would be impossible to produce a non-decomposable ternary function using only pinning.

In Section~\ref{s:symmetrising_ternary}, we furthermore show that, under some mild assumptions on $\cF$, the set $S(\cF\cup\{\dl_0,\new{\dl}_1,\dl_+,\dl_-\})$ actually contains a \emph{symmetric} non-decomposable ternary function.
In Section~\ref{s:binary}, we exhibit gadget constructions for certain binary functions and show that the assumptions of the previous section are satisfied if $\cF$ is not one of the exceptional cases of Theorem~\ref{thm:Holant-star}.
All the gadgets in these subsections are planar.
To ensure that the full complexity classification works for planar holant problems, we next give a reduction between planar holant problems and planar counting CSPs in Section~\ref{s:interreducing_planar}.
This is based on results sketched in the literature, but to our knowledge the full proof has not been written out before.
Finally, in Section~\ref{s:hardness}, we combine all of the parts to prove the complexity classification for $\hol^+$, which holds even when restricted to the planar case.
}

\subsection{Why these free functions?}

The definition of $\hol^+$ is motivated by the following results from quantum theory.
We first state the results in quantum terminology and translate them into holant terminology at the end of the section.

\begin{theorem}[{\cite[Lemma on p.~296]{popescu_generic_1992},\cite{gachechiladze_addendum_2016}}]\label{thm:popescu-rohrlich}
 Let $\ket{\Psi}$ be an $n$-system genuinely entangled quantum state. For any two of the $n$ systems, there exists a projection, onto a tensor product of states of the other $(n-2)$ systems, that leaves the two systems in an entangled state.
\end{theorem}
Here, `projection' means a (partial) inner product between $\ket{\Psi}$ and the tensor product of single-system states.
{The $n$ systems do not have to be qubits, but for the purposes of the following arguments, it suffices to think of them as $n$ qubits.}
The original proof of this statement in \cite{popescu_generic_1992} was flawed but it was recently corrected \cite{gachechiladze_addendum_2016}.
The following corollary is not stated explicitly in either paper, but can be seen to hold by inspecting the proof in \cite{gachechiladze_addendum_2016}.

\begin{corollary}\label{cor:popescu-rohrlich_restricted}
 {Let $\ket{\Psi}$ be an $n$-qubit genuinely entangled quantum state. For any two of the $n$ qubits, there exists a projection, onto a tensor product of computational and Hadamard basis states of the other $(n-2)$ qubits, that leaves the remaining two qubits in an entangled state.}
\end{corollary}

In other words, Theorem \ref{thm:popescu-rohrlich} holds when the systems are restricted to qubits and the projectors are restricted to products of computational and Hadamard basis states.
Here, it is crucial to have projectors taken from two bases that are linked by the Hadamard transformation: the proof applies only in that case.
{Intuitively, the corollary states that if $n$ parties share a genuinely entangled $n$-qubit state, then this can be converted into an entangled 2-qubit state shared by two of the parties using local projections\footnote{In the real world, these projections correspond to post-selected measurements, so without post-selection the protocol may fail in some runs.}.

In holant terminology, the corollary corresponds to the following proposition about producing binary non-decomposable functions from a higher-arity non-decomposable function via gadgets with unary functions.}

\begin{proposition}[Restatement of Corollary~\ref{cor:popescu-rohrlich_restricted}]\label{prop:popescu-rohrlich_gadget}
 Let $f$ be a non-decomposable function of arity {$n\geq 2$}.
 Suppose $j,k\in [n]$ with $j<k$.
 Then there exist $u_m\in\{\dl_0,\dl_1,\dl_+,\dl_-\}$ for all $m\in [n]\setminus\{j,k\}$ such that the following binary function is non-decomposable:
 \[
  g(x_j,x_k) = \sum_{x_s\in\{0,1\}\text{ for } s\in [n]\setminus\{j,k\}} f(\vc{x}{n}) \prod_{m\in [n]\setminus\{j,k\}} u_m(x_m).
 \]
 This $g$ is the effective function of the gadget in Figure~\ref{fig:pr}a.
\end{proposition}

\begin{rem}
 Proposition~\ref{prop:popescu-rohrlich_gadget} can be considered an alternative definition of what it means to be non-decomposable.
 To see this, suppose the function $f\in\allf_n$ is decomposable, i.e.\ there exists $1\leq k<n$ and functions $f_1,f_2$ such that:
 \[
  f(\vc{x}{n}) = f_1(x_{\rho(1)}\zd x_{\rho(k)})f_2(x_{\rho(k+1)}\zd x_{\rho(n)}).
 \]
 Choose one argument from each partition, say $x_{\rho(1)}$ and $x_{\rho(n)}$; the argument is analogous for any other choice.
 Then for all $u_m\in\{\dl_0,\dl_1,\dl_+,\dl_-\}$ we have
 \begin{align*}
  g(x_{\rho(1)},x_{\rho(n)})
  &= \sum_{x_s\in\{0,1\}\text{ for } s\in [n]\setminus\{\rho(1),\rho(n)\}} f(\vc{x}{n}) \prod_{m\in [n]\setminus\{\rho(1),\rho(n)\}} u_m(x_m) \\
  &= \left( \sum_{x_{\rho(2)}\zd x_{\rho(k)}\in\{0,1\}} f_1(x_{\rho(1)}\zd x_{\rho(k)}) \prod_{m=2}^k u_m(x_{\rho(m)}) \right) \\
  &\qquad\times \left( \sum_{x_{\rho(k+1)}\zd x_{\rho(n-1)}\in\{0,1\}} f_2(x_{\rho(k+1)}\zd x_{\rho(n)}) \prod_{m=k+1}^{n-1} u_m(x_{\rho(m)}) \right)
 \end{align*}
 which is clearly decomposable.
 Thus if the conclusion of Proposition~\ref{prop:popescu-rohrlich_gadget} holds for some $f$, then $f$ must be non-decomposable.
\end{rem}

We extend this proposition as follows.
Note that all gadgets are planar.

\begin{figure}
 \centering
 (a) \begin{tikzpicture}[scale=1.25]
	\begin{pgfonlayer}{nodelayer}
		\node [style=none] (0) at (1.75, 1.5) {};
		\node [style=none] (1) at (-1.75, 1.5) {};
		\node [style=greyn] (2) at (0, -1) {};
		\node [style=none] (3) at (0, -1.5) {$f$};
		\node [style=none] (5) at (-3.5, 0.5) {$\ldots$};
		\node [style=none] (6) at (-4.25, 1) {$u_1$};
		\node [style=greyn] (7) at (-4.25, 0.5) {};
		\node [style=greyn] (8) at (-2.75, 0.5) {};
		\node [style=none] (9) at (-2.75, 1) {$u_{j-1}$};
		\node [style=none] (10) at (0, 0.5) {$\ldots$};
		\node [style=none] (11) at (-0.75, 1) {$u_{j+1}$};
		\node [style=greyn] (12) at (-0.75, 0.5) {};
		\node [style=greyn] (13) at (0.75, 0.5) {};
		\node [style=none] (14) at (0.75, 1) {$u_{k-1}$};
		\node [style=none] (15) at (3.5, 0.5) {$\ldots$};
		\node [style=none] (16) at (2.75, 1) {$u_{k+1}$};
		\node [style=greyn] (17) at (2.75, 0.5) {};
		\node [style=greyn] (18) at (4.25, 0.5) {};
		\node [style=none] (19) at (4.25, 1) {$u_n$};
	\end{pgfonlayer}
	\begin{pgfonlayer}{edgelayer}
		\draw [bend right=15] (7) to (2);
		\draw [bend right=15] (8) to (2);
		\draw [bend right] (1.center) to (2);
		\draw [bend right=15] (12) to (2);
		\draw [bend right=15] (2) to (13);
		\draw [bend left] (0.center) to (2);
		\draw [bend right=15] (2) to (17);
		\draw [bend right=15] (2) to (18);
	\end{pgfonlayer}
\end{tikzpicture} \hfill (b) \begin{tikzpicture}[scale=1.25]
	\begin{pgfonlayer}{nodelayer}
		\node [style=none] (0) at (0, 1.5) {};
		\node [style=none] (1) at (-3.25, 1.5) {};
		\node [style=greyn] (2) at (0, -1) {};
		\node [style=none] (3) at (0, -1.5) {$f$};
		\node [style=none] (5) at (-5, 0.5) {$\ldots$};
		\node [style=none] (6) at (-5.75, 1) {$u_1$};
		\node [style=greyn] (7) at (-5.75, 0.5) {};
		\node [style=greyn] (8) at (-4.25, 0.5) {};
		\node [style=none] (9) at (-4.25, 1) {$u_{j-1}$};
		\node [style=none] (10) at (-1.5, 0.5) {$\ldots$};
		\node [style=none] (11) at (-2.25, 1) {$u_{j+1}$};
		\node [style=greyn] (12) at (-2.25, 0.5) {};
		\node [style=greyn] (13) at (-0.75, 0.5) {};
		\node [style=none] (14) at (-0.75, 1) {$u_{k-1}$};
		\node [style=none] (15) at (1.5, 0.5) {$\ldots$};
		\node [style=none] (16) at (0.75, 1) {$u_{k+1}$};
		\node [style=greyn] (17) at (0.75, 0.5) {};
		\node [style=greyn] (18) at (5.75, 0.5) {};
		\node [style=none] (19) at (5.75, 1) {$u_n$};
		\node [style=none] (20) at (3.25, 1.5) {};
		\node [style=greyn] (21) at (2.25, 0.5) {};
		\node [style=none] (22) at (2.25, 1) {$u_{\ell-1}$};
		\node [style=none] (23) at (5, 0.5) {$\ldots$};
		\node [style=none] (24) at (4.25, 1) {$u_{\ell+1}$};
		\node [style=greyn] (25) at (4.25, 0.5) {};
	\end{pgfonlayer}
	\begin{pgfonlayer}{edgelayer}
		\draw [bend right=15] (7) to (2);
		\draw [bend right=15] (8) to (2);
		\draw [in=165, out=-90] (1.center) to (2);
		\draw (12) to (2);
		\draw (2) to (13);
		\draw (0.center) to (2);
		\draw (2) to (17);
		\draw [bend right=15] (2) to (18);
		\draw (2) to (21);
		\draw [in=-90, out=15] (2) to (20.center);
		\draw [bend right=15] (2) to (25);
	\end{pgfonlayer}
\end{tikzpicture}
 \caption{(a) The gadget from Proposition~\ref{prop:popescu-rohrlich_gadget} and (b) the gadget from Theorem~\ref{thm:three-qubit-gadget}.}
 \label{fig:pr}
\end{figure}
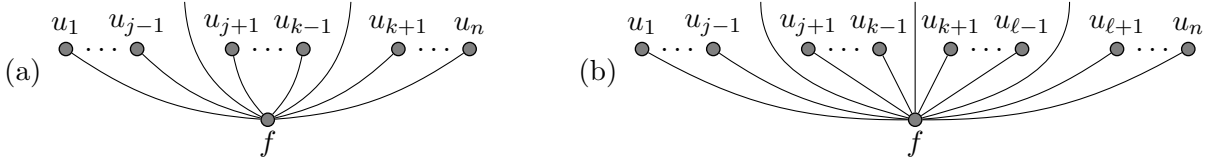

\begin{theorem}\label{thm:three-qubit-gadget}
 Let $f$ be a non-decomposable function of arity $n\geq 3$.
 Then there exist $j,k,\ell\in [n]$ with $j<k<\ell$, and $u_m\in\{\dl_0,\dl_1,\dl_+,\dl_-\}$ for all $m\in [n]\setminus\{j,k,\ell\}$, such that the following ternary function is non-decomposable:
 \[
  g(x_j,x_k,x_\ell) = \sum_{x_s\in\{0,1\}\text{ for } s\in [n]\setminus\{j,k,\ell\}} f(\vc{x}{n}) \prod_{m\in [n]\setminus\{j,k,\ell\}} u_m(x_m).
 \]
 This $g$ is the effective function of the gadget in Figure~\ref{fig:pr}b.
\end{theorem}
{\begin{proof}
 The result is proved by induction on $n$.
 If $n=3$, $f$ itself is the desired non-decomposable ternary function; this is the base case.
 Now suppose the result holds for all $n$ satisfying $3\leq n\leq N$.

 We prove the result for $n=N+1$ by contradiction, i.e.\ we begin by assuming that for some non-decomposable function $f$ of arity $n=N+1$ there does not exist any choice $j,k,\ell\in [n]$ with $j<k<\ell$, and $u_m\in\{\dl_0,\dl_1,\dl_+,\dl_-\}$ for all $m\in [n]\setminus\{j,k,\ell\}$ such that the function $g$ defined in the theorem statement is non-decomposable.
 Note the function $f$ cannot be identically zero since such functions are trivially decomposable.
 
 First, consider the family of gadgets that arise by composing one input of $f$ with one of the allowed unary functions:
 \[
  h_{j,u}(\vc{x}{j-1},x_{j+1}\zd x_{n}) := \sum_{x_j\in\{0,1\}} f(\vc{x}{n}) u(x_j)
 \]
 where $j\in [n]$ and $u\in\{\dl_0,\dl_1,\dl_+,\dl_-\}$.
 \new{Note that
 \begin{align*}
  f(\vc{x}{n}) &= h_{j,\dl_0}(\vc{x}{j-1},x_{j+1}\zd x_{n})\dl_0(x_j) + h_{j,\dl_1}(\vc{x}{j-1},x_{j+1}\zd x_{n})\dl_1(x_j) \\
  &= h_{j,\dl_+}(\vc{x}{j-1},x_{j+1}\zd x_{n})\dl_+(x_j) + h_{j,\dl_-}(\vc{x}{j-1},x_{j+1}\zd x_{n})\dl_-(x_j)
 \end{align*}
 for any $j\in[n]$.
 Hence, if $h_{j,u}$ was identically zero for some $j$ and $u$, then $f$ would be decomposable.
 But we assumed $f$ was non-decomposable, therefore the functions $h_{j,u}$ cannot be identically zero.}
 
 Furthermore, if one of the functions $h_{j,u}$ has a non-decomposable tensor factor of arity at least $3$, then we can remove the other tensor factors by Lemma~\ref{lem:decomposable}, replace $f$ with the resulting function of arity between 3 and $N$ (inclusive), and be done by the inductive hypothesis.
 Thus, for all $j$ and $u$, we must have $h_{j,u}\in\ang{\cT}$, i.e.\ $h_{j,u}$ decomposes as a tensor product of unary and binary functions.
 
 Now by Proposition~\ref{prop:popescu-rohrlich_gadget}, we can find $u_m\in\{\dl_0,\dl_1,\dl_+,\dl_-\}$ for all $m\in [n]\setminus\{1,2\}$ such that
 \[
  b(x_1,x_2) := \sum_{x_s\in\{0,1\}\text{ for } s\in [n]\setminus\{1,2\}} f(\vc{x}{n}) \prod_{m\in [n]\setminus\{1,2\}} u_m(x_m)
 \]
 is non-decomposable.
 Since $b$ arises from $h_{n,u_n}$ by contraction with unary functions, the arguments $x_1$ and $x_2$ must appear in the same tensor factor of $h_{n,u_n}$.
 Yet $h_{n,u_n}\in\ang{\cT}$ by the argument of the previous paragraph, so we must have $h_{n,u_n}(\vc{x}{n-1}) = b(x_1,x_2) h(x_3\zd x_{n-1})$ for some $h\in\ang{\cT}$, which is not identically zero.
 Similarly, by Proposition~\ref{prop:popescu-rohrlich_gadget}, we can find $v_m\in\{\dl_0,\dl_1,\dl_+,\dl_-\}$ for all $m\in [n]\setminus\{2,3\}$ such that
 \[
  b'(x_2,x_3) := \sum_{x_s\in\{0,1\}\text{ for } s\in [n]\setminus\{2,3\}} f(\vc{x}{n}) \prod_{m\in [n]\setminus\{2,3\}} v_m(x_m)
 \]
 is non-decomposable.
 Then, analogous to the above, $h_{1,v_1}(x_2\zd x_n) = b'(x_2,x_3) h'(x_4\zd x_n)$ for some $h'\in\ang{\cT}$, which is not identically zero.
 
 Now consider the gadget
 \[
  f'(x_2\zd x_{n-1}) := \sum_{x_1,x_n\in\{0,1\}} f(\vc{x}{n}) v_1(x_1) u_n(x_n).
 \]
 If we perform the sum over $x_n$ first, we find
 \begin{equation}\label{eq:f-prime1}
  f'(x_2\zd x_{n-1}) = \sum_{x_1\in\{0,1\}} h_{n,u_n}(\vc{x}{n-1}) v_1(x_1) = \sum_{x_1\in\{0,1\}} b(x_1,x_2) h(x_3\zd x_{n-1}) v_1(x_1),
 \end{equation}
 which is not identically zero since $h$ is not, and $v'(x_2) := \sum_{x_1\in\{0,1\}} b(x_1,x_2) v_1(x_1)$ being identically zero would imply $b$ is decomposable.
 By inspection, the arguments $x_2$ and $x_3$ appear in different tensor factors of $f'$.
 
 If, on the other hand, we perform the sum over $x_1$ first, we find
 \begin{equation}\label{eq:f-prime2}
  f'(x_2\zd x_{n-1}) = \sum_{x_n\in\{0,1\}} h_{1,v_1}(x_2\zd x_n) u_n(x_n) = \sum_{x_n\in\{0,1\}} b'(x_2,x_3) h'(x_4\zd x_n) u_n(x_n).
 \end{equation}
 This could be identically zero if $h'(x_4\zd x_n) = h''(x_4\zd x_{n-1}) u_n^{\perp}(x_n)$ for some $h''$, where the function $u_n^\perp$ satisfies $\sum_{x_n\in\{0,1\}}u_n^{\perp}(x_n) u_n(x_n)=0$.
 Yet from \eqref{eq:f-prime1} we deduced that $f'$ is not identically zero, so this cannot happen.
 Then, inspection of \eqref{eq:f-prime2} shows that the arguments $x_2$ and $x_3$ appear in the same non-decomposable tensor factor of $f'$.
 This contradicts the finding from \eqref{eq:f-prime1} that they appear in different tensor factors.
 
 Hence the assumption must have been wrong and we have $h_{1,v_1}\notin\ang{\cT}$ or $h_{n,u_n}\notin\ang{\cT}$.
 Thus, by Lemma~\ref{lem:decomposable} and the induction hypothesis, we can realise the desired non-decomposable ternary function.
\end{proof}

In quantum terminology, this corresponds to the following theorem.}

\begin{theorem}[Restatement of Theorem~\ref{thm:three-qubit-gadget}]\label{thm:three-qubit-entanglement}
 Let $\ket{\Psi}$ be an $n$-qubit genuinely entangled state with $n\geq 3$. There exists some choice of three of the $n$ qubits and a projection of the other $(n-3)$ qubits onto a tensor product of computational and Hadamard basis states that leaves the three qubits in a genuinely entangled state.
\end{theorem}

This result, which was not previously known in the quantum information theory literature, is stronger than Corollary~\ref{cor:popescu-rohrlich_restricted} in that we construct entangled three-qubit states rather than two-qubit ones.
On the other hand, our result may not hold for arbitrary choices of three qubits: all we show is that there exists some choice of three qubits for which it does hold.

The original proof of this theorem in an earlier version of this paper was long and involved; this new shorter proof was suggested by Gachechiladze and G\"{u}hne \cite{gachechiladze_personal_2017}.

\subsection{Symmetrising ternary functions}
\label{s:symmetrising_ternary}

The dichotomies given in Section \ref{s:results_ternary_symmetric} apply to symmetric ternary non-decomposable functions.
The functions constructed according to Theorem \ref{thm:three-qubit-gadget} are ternary and non-decomposable, but they are not generally symmetric.
Yet, these general ternary non-decomposable functions can be used to realise symmetric ones, possibly with the help of an additional binary non-decomposable function.
We prove this by distinguishing cases according to whether the ternary non-decomposable function constructed using Theorem \ref{thm:three-qubit-gadget} is in the GHZ or the $W$ entanglement class (cf.\ Section~\ref{s:entanglement}).

Consider a function $f\in\allf_3$ which is in the GHZ class.
By definition, there exist matrices $A,B,C\in\GL$ such that $\ket{f} = (A\otimes B\otimes C)\ket{\GHZ}$, i.e.
\[
 f(x_1,x_2,x_3) = \sum_{y_1,y_2,y_3\in\{0,1\}} A_{x_1, y_1} B_{x_2, y_2} C_{x_3, y_3} \EQ_3(y_1,y_2,y_3).
\]
We can thus draw $f$ as the `virtual gadget' shown in Figure \ref{fig:virtual_gadget}.
The `boxes' denoting the matrices are non-symmetric to indicate that $A,B,C$ are not in general symmetric.
The white dot is assigned $\EQ_3$.
This notation is not meant to imply that the binary functions associated with $A,B,C$ or the ternary equality function are available on their own.
Instead, thinking of the function as such a composite will simply make future arguments more straightforward.

Similarly, if $f$ is in the $W$ class then there exist matrices $A,B,C\in\GL$ such that $\ket{f} = (A\otimes B\otimes C)\ket{W}$, or equivalently,
\[
 f(x_1,x_2,x_3) = \sum_{y_1,y_2,y_3\in\{0,1\}} A_{x_1, y_1} B_{x_2, y_2} C_{x_3, y_3} \ONE_3(y_1,y_2,y_3).
\]
In this case, $f$ can again be represented as a virtual gadget, but now the white dot in Figure~\ref{fig:virtual_gadget} is assigned the function $\ONE_3$.

In both the GHZ and the $W$ case, three vertices assigned $f$ can be connected to form the rotationally symmetric gadget shown in Figure \ref{fig:symmetrising_GHZ}a.
In fact, since $f$ is a function over the Boolean domain, the effective function  $g$ of that gadget is fully symmetric: its value depends only on the Hamming weight of the inputs.
On the other hand, $g$ may be decomposable (in fact, any symmetric decomposable function must be degenerate) and it may be the all-zero function.
For a general non-symmetric $f$ there are three such symmetric gadgets that can be constructed by `rotating' $f$, i.e.\ by replacing $f(x_1,x_2,x_3)$ with $f(x_2,x_3,x_1)$ or $f(x_3,x_1,x_2)$.
Rotating $f$ in this way does not affect the planarity of the gadget.
The idea leads to the following lemmas.

\begin{figure}
 \centering
 \begin{tikzpicture}
	\begin{pgfonlayer}{nodelayer}
		\node [style=hollown] (0) at (0, -1) {};
		\node [style=none] (1) at (-2, 1.25) {};
		\node [style=map] (2) at (-2, 0.25) {$C$};
		\node [style=map] (3) at (0, 0.25) {$B$};
		\node [style=map] (4) at (2, 0.25) {$A$};
		\node [style=none] (5) at (0, 1.25) {};
		\node [style=none] (6) at (2, 1.25) {};
	\end{pgfonlayer}
	\begin{pgfonlayer}{edgelayer}
		\draw (0) to (3);
		\draw (1.center) to (2);
		\draw [bend right, looseness=1.00] (2) to (0);
		\draw [bend right, looseness=1.00] (0) to (4);
		\draw (4) to (6.center);
		\draw (3) to (5.center);
	\end{pgfonlayer}
\end{tikzpicture}
 \caption{A `virtual gadget' for a non-decomposable ternary function. The white vertex represents either $\EQ_3$ or $\ONE_3$ and the boxes represent (not necessarily symmetric) binary functions corresponding to the matrices $A$, $B$, and $C$, respectively.}
 \label{fig:virtual_gadget}
\end{figure}
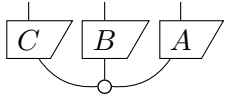

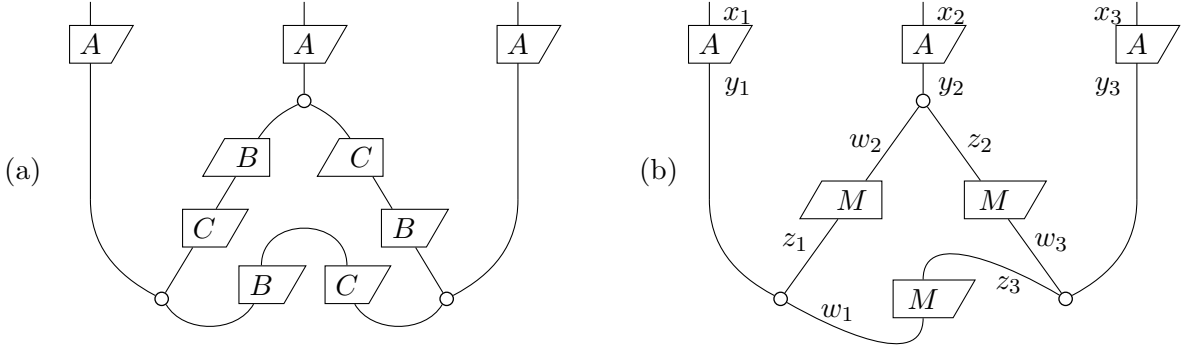
\begin{figure}
 \centering
 (a) \; \begin{tikzpicture}[scale=1.5]
	\begin{pgfonlayer}{nodelayer}
		\node [style=hollown] (0) at (0, 1.25) {};
		\node [style=hollown] (1) at (-2.5, -2.25) {};
		\node [style=hollown] (2) at (2.5, -2.25) {};
		\node [style=none] (3) at (0, 3) {};
		\node [style=none] (4) at (-3.75, -0.5) {};
		\node [style=none] (5) at (3.75, -0.5) {};
		\node [style=none] (6) at (-3.75, 3) {};
		\node [style=none] (7) at (3.75, 3) {};
		\node [style=map] (8) at (-3.75, 2.25) {$A$};
		\node [style=map] (9) at (-1.75, -1) {$C$};
		\node [style=map] (10) at (-0.75, -2) {$B$};
		\node [style=map] (11) at (0.75, -2) {$C$};
		\node [style=map] (12) at (1.75, -1) {$B$};
		\node [style=map] (13) at (3.75, 2.25) {$A$};
		\node [style=map] (14) at (0, 2.25) {$A$};
		\node [style=transposemap] (15) at (-1, 0.25) {$B$};
		\node [style=transposemap] (16) at (1, 0.25) {$C$};
	\end{pgfonlayer}
	\begin{pgfonlayer}{edgelayer}
		\draw [in=-90, out=152, looseness=1.00] (1) to (4.center);
		\draw [in=-90, out=28, looseness=1.00] (2) to (5.center);
		\draw (3.center) to (0);
		\draw (6.center) to (4.center);
		\draw (7.center) to (5.center);
		\draw (1) to (9);
		\draw (2) to (12);
		\draw (9) to (15);
		\draw [bend left=15, looseness=1.00] (15) to (0);
		\draw [bend left=15, looseness=1.00] (0) to (16);
		\draw (16) to (12);
		\draw [bend right=60, looseness=1.00] (1) to (10);
		\draw [bend left=60, looseness=1.00] (2) to (11);
		\draw [bend left=90, looseness=1.50] (10) to (11);
	\end{pgfonlayer}
\end{tikzpicture} \qquad
 (b) \begin{tikzpicture}[scale=1.5]
	\begin{pgfonlayer}{nodelayer}
		\node [style=hollown] (0) at (0, 1.25) {};
		\node [style=hollown] (1) at (-2.5, -2.25) {};
		\node [style=hollown] (2) at (2.5, -2.25) {};
		\node [style=none] (3) at (0, 3) {};
		\node [style=none] (4) at (-3.75, -0.5) {};
		\node [style=none] (5) at (3.75, -0.5) {};
		\node [style=none] (6) at (-3.75, 3) {};
		\node [style=none] (7) at (3.75, 3) {};
		\node [style=map] (8) at (-3.75, 2.25) {$A$};
		\node [style=transposemap] (9) at (-1.25, -0.5) {$M$};
		\node [style=map] (11) at (0, -2.25) {$M$};
		\node [style=map] (13) at (3.75, 2.25) {$A$};
		\node [style=map] (14) at (0, 2.25) {$A$};
		\node [style=map] (16) at (1.25, -0.5) {$M$};
		\node [style=none] (17) at (-3.25, 2.75) {$x_1$};
		\node [style=none] (18) at (-3.25, 1.5) {$y_1$};
		\node [style=none] (19) at (-2.25, -1.25) {$z_1$};
		\node [style=none] (20) at (-1.5, -2.5) {$w_1$};
		\node [style=none] (21) at (1.5, -2) {$z_3$};
		\node [style=none] (22) at (2.25, -1.25) {$w_3$};
		\node [style=none] (24) at (3.25, 1.5) {$y_3$};
		\node [style=none] (25) at (3.25, 2.75) {$x_3$};
		\node [style=none] (26) at (0.5, 1.5) {$y_2$};
		\node [style=none] (27) at (0.5, 2.75) {$x_2$};
		\node [style=none] (28) at (1, 0.5) {$z_2$};
		\node [style=none] (29) at (-1, 0.5) {$w_2$};
	\end{pgfonlayer}
	\begin{pgfonlayer}{edgelayer}
		\draw [in=-90, out=152] (1) to (4.center);
		\draw [in=-90, out=28] (2) to (5.center);
		\draw (3.center) to (0);
		\draw (6.center) to (4.center);
		\draw (7.center) to (5.center);
		\draw (1) to (9);
		\draw (0) to (16);
		\draw [in=90, out=149] (2) to (11);
		\draw (9) to (0);
		\draw (16) to (2);
		\draw [in=330, out=-90] (11) to (1);
	\end{pgfonlayer}
\end{tikzpicture}
 \caption{(a) A symmetric gadget constructed from three copies of the ternary function from Figure~\ref{fig:virtual_gadget}. (b) A simplified version of the same gadget, where $M:=C^TB$. The variable names next to the edges are those used in \eqref{eq:symmetrising}.}
 \label{fig:symmetrising_GHZ}
\end{figure}

\begin{lemma}\label{lem:GHZ_symmetrise}
 Suppose $f\in\allf_3$ has GHZ type, i.e.\ $\ket{f}=(A\otimes B\otimes C)\ket{\GHZ}$ for some matrices $A,B,C\in\GL$.
 Then there exists a non-decomposable symmetric ternary function $g\in S(\{f\})$, which is furthermore realisable by a planar gadget.
\end{lemma}
\begin{proof}
 The function $f$ can be represented by a virtual gadget as in Figure~\ref{fig:virtual_gadget}, where the white vertex is assigned $\EQ_3$ and the matrices $A$, $B$, and $C$ are those appearing in the statement of the lemma.
 We will realise a symmetric ternary function by using the triangle gadget in Figure~\ref{fig:symmetrising_GHZ}a.
 Note that there are three different planar versions of this gadget by cyclically permuting the inputs of $f$ (and thus the roles of $A$, $B$, and $C$) in the gadget.
 The desired result will follow from arguing that either at least one of these three gadgets is non-decomposable and therefore yields the desired function $g$, or $f$ is already symmetric.
 In the latter case, we simply take $g:=f$.

 Consider the gadget in Figure~\ref{fig:symmetrising_GHZ}a.
 To simplify the argument, we will not parameterise each of the three matrices $A$, $B$, and $C$ individually; instead we let
 \begin{equation}\label{eq:M}
  M := C^T B = \begin{pmatrix}a&b\\c&d\end{pmatrix},
 \end{equation}
 yielding the gadget in Figure~\ref{fig:symmetrising_GHZ}b.
 The effective function of that gadget, which we will denote $h(x_1,x_2,x_3)$ is equal to
 \begin{multline}
  \sum A_{x_1,y_1} A_{x_2,y_2} A_{x_3,y_3} \EQ_3(y_1,z_1,w_1) \EQ_3(y_2,z_2,w_2) \EQ_3(y_3,z_3,w_3) M_{z_1,w_2} M_{z_2,w_3} M_{z_3,w_1} \\
  = \sum_{y_1,y_2,y_3\in\{0,1\}} A_{x_1,y_1} A_{x_2,y_2} A_{x_3,y_3} M_{y_1,y_2} M_{y_2,y_3} M_{y_3,y_1}, \label{eq:symmetrising}
 \end{multline}
 where the first sum is over all $y_1,y_2,y_3,z_1,z_2,z_3,w_1,w_2,w_3\in\{0,1\}$.
 Let the function $h'$ be such that $h=A\circ h'$, then, by invertibility of $A$,
 \[
  h'(y_1,y_2,y_3) = M_{y_1,y_2} M_{y_2,y_3} M_{y_3,y_1}.
 \]
 Plugging in the values from \eqref{eq:M}, we find that $h' = [a^3, abc, bcd, d^3]$.
 Since $h$ and $h'$ are connected by a holographic transformation, they must be in the same entanglement class: in particular, $h$ is non-decomposable if and only if $h'$ is non-decomposable.
 Hence we may work with $h'$ instead of $h$ for the remainder of the proof.
 
 Recall from Section~\ref{s:entanglement} that a symmetric ternary function is non-decomposable if and only if it has either GHZ type or $W$ type.
 By Lemma~\ref{lem:li_symmetric}, $h'$ has GHZ type if and only if:
 \[
  P_{h'}:=a^2 d^2 (ad+3bc)(ad-bc)^3\neq 0.
 \]
 It has $W$ type if $P_{h'} = 0$ and furthermore
 \begin{equation}\label{eq:W-conditions-h'}
  \left( (abc)^2 \neq a^3 bcd \right) \vee \left( (bcd)^2 \neq abcd^3 \right).
 \end{equation}
 If neither of these conditions is satisfied, $h'$ is decomposable (and in fact degenerate).
 
 Now, as $M$ is invertible, we have $ad-bc\neq 0$.
 Thus, $h'$ fails to have GHZ type only if at least one of $a$, $d$, or $(ad+3bc)$ is zero.
 We consider the cases individually.
 \begin{itemize}
  \item Suppose $a=0$.
   Then $P_{h'}=0$ and \eqref{eq:W-conditions-h'} becomes $(0\neq 0)\vee(bcd\neq 0)$.
   Hence $h'$ has $W$ type if $bcd\neq 0$ and is degenerate otherwise.
   Since $M$ is invertible, $a=0$ implies $bc\neq 0$.
   Thus, $h'$ is degenerate (in fact, it is identically zero) if $a=d=0$ and it has $W$ type otherwise.
  \item Suppose $a\neq 0$ and $d=0$.
   Then $P_{h'}=0$ and \eqref{eq:W-conditions-h'} becomes $(abc\neq 0)\vee (0\neq 0)$.
   Hence $h'$ has $W$ type if $abc\neq 0$. 
   But invertibility of $M$ together with $d=0$ implies that $bc\neq 0$, so in this case $a,b,c$ are all guaranteed to be non-zero.
   Therefore $h'$ always has $W$ type.
  \item Suppose $a,d\neq 0$ and $ad+3bc=0$, i.e.\ $bc=-\frac{1}{3}ad$.
   Then $P_{h'}=0$ and by substituting for $bc$, \eqref{eq:W-conditions-h'} becomes
   \[
    \left( \frac{1}{9} a^4 d^2 \neq -\frac{1}{3} a^4 d^2 \right) \vee \left(\frac{1}{9} a^2 d^4 \neq -\frac{1}{3} a^2 d^4 \right),
   \]
   which is true for all $a,d\neq 0$.
   Therefore $h'$ always has $W$ type.
 \end{itemize}
 By combining the three cases, we find that $h'$ is degenerate if and only if $a=d=0$, or equivalently if and only if $(C^TB)_{00}=(C^TB)_{11}=0$.
 
 As noted before, there are three different planar gadgets that can be constructed from the same non-symmetric ternary function, by `rotating' it.
 Assume that all three  gadget constructions yield a decomposable function.
 For the original gadget, this means that $(C^TB)_{00}=(C^TB)_{11}=0$.
 For the rotated versions of the gadget, the assumption furthermore implies that
 \[
  (B^TA)_{00}=0=(B^TA)_{11} \quad\text{and}\quad (A^TC)_{00}=0=(A^TC)_{11}.
 \]
 These equalities indicate that $C^TB$, $B^TA$, and $A^TC$ are purely off-diagonal matrices.
 Additionally, since $A,B,C$ are invertible, $C^TB$, $B^TA$, and $A^TC$ must also be invertible.
 Hence there exist invertible diagonal matrices $D_1,D_2,D_3\in\GL$ such that:
 \begin{equation}\label{eq:GHZ_fail_conditions}
  B^TA=XD_1, \qquad C^TB=XD_2, \quad\text{and}\quad A^TC = XD_3.
 \end{equation}
 By rearranging the second equation of \eqref{eq:GHZ_fail_conditions}, we find $B = (C^T)^{-1} X D_2$.
 Now, by transposing, rearranging, and then inverting, the third equation of \eqref{eq:GHZ_fail_conditions} is equivalent to
 \[
  C^T A = D_3 X \quad\Longleftrightarrow\quad C^T = D_3 X A^{-1} \quad\Longleftrightarrow\quad (C^T)^{-1} = A X D_3^{-1},
 \]
 so $B = A X D_3^{-1} X D_2$.
 Similarly, transposing and then rearranging the second equation of \eqref{eq:GHZ_fail_conditions} yields $C = (B^T)^{-1} D_2 X$.
 By rearranging and inverting, the first equation of \eqref{eq:GHZ_fail_conditions} is equivalent to
 \[
  B^T = X D_1 A^{-1} \quad\Longleftrightarrow\quad (B^T)^{-1} = A D_1^{-1} X,
 \]
 so $C = A D_1^{-1} X D_2 X$.
 
 Both $D_1^{-1} X D_2 X$ and $X D_3^{-1} X D_2$ are diagonal matrices; write them as
 \begin{equation}
  D_C = D_1^{-1} X D_2 X = \begin{pmatrix}\gamma_0&0\\0&\gamma_1\end{pmatrix} \quad\text{and}\quad D_B = X D_3^{-1} X D_2 = \begin{pmatrix}\beta_0&0\\0&\beta_1\end{pmatrix}
 \end{equation}
 for some $\beta_0,\beta_1,\gamma_0,\gamma_1\in\AA\setminus\{0\}$.
 Then $B=A D_B$ and $C=A D_C$, so:
 \begin{equation}
  \ket{f} = (A\otimes B\otimes C)\ket{\GHZ} = A\t{3} (I\otimes D_B\otimes D_C) \ket{\GHZ} = A\t{3}\left( \beta_0\gamma_0\ket{000}+\beta_1\gamma_1\ket{111} \right),
 \end{equation}
 where $I$ is the 2 by 2 identity matrix.
 Hence the assumption that all three gadgets are decomposable implies that $f=A\circ [\beta_0\gamma_0, 0, 0, \beta_1\gamma_1]$, i.e.\ the assumption implies that $f$ is already symmetric.
 
 We have shown that either one of the three possible planar triangle gadgets yields a symmetric non-degenerate function, or $f$ itself is already symmetric.
 Thus there always exists a non-degenerate symmetric ternary function in $S(\{f\})$ which is realisable by a planar gadget.
\end{proof}

\begin{lemma}\label{lem:W_symmetrise}
 Suppose $f\in\allf_3$ has $W$ type, i.e.\ $\ket{f}=(A\otimes B\otimes C)\ket{W}$ for some matrices $A,B,C\in\GL$.
 {If $f\notin (K\circ\cM)\cup(KX\circ\cM)$, there exists a non-degenerate symmetric ternary function $g\in S(\{f\})$.
 The function}
 $g$ has GHZ type and is realisable by a planar gadget.
\end{lemma}
\begin{proof}
 Since $f$ has $W$ type, we can write $\ket{f}=(A\otimes B\otimes C)\ket{W}$: i.e.\ the function $f$ can be thought of as the `virtual gadget' given in Figure~\ref{fig:virtual_gadget}, where the white dot is now assigned $\ONE_3$.
 We can therefore combine three copies of $f$ into the triangle gadget given in Figure~\ref{fig:symmetrising_GHZ}a, analogous to Lemma~\ref{lem:GHZ_symmetrise}.

 There are two cases, depending on whether $(C^T B)_{0,0}$ vanishes.

 \textit{Case~1}: {By basic linear algebra, if $(C^T B)_{0,0}\neq 0$, we can apply the PLDU decomposition\footnote{The PLDU decomposition (permutation, lower triangular, diagonal, upper triangular) is a variant on the more common PLU decomposition. In contrast to the latter, it requires the diagonal terms in the lower and upper triangular matrices to be 1 and includes an additional diagonal matrix in the product.} to $C^T B$ with a trivial permutation $P$, i.e.:
 \[
  C^T B = LDU = \pmm{1&0\\a&1}\pmm{b&0\\0&c}\pmm{1&d\\0&1},
 \]
 where $a,b,c,d\in\AA$ with $b,c\neq 0$ and $L=\smm{1&0\\a&1}$ is lower triangular, $D=\smm{b&0\\0&c}$ is diagonal, and $U=\smm{1&d\\0&1}$ is upper triangular.
 This is illustrated in Figure~\ref{fig:sym-W}a.

 Now, note that each of the white dots assigned $\ONE_3$ is connected to a copy of $U$ and a transposed copy of $L$.
 It is straightforward to show that
 \begin{equation}\label{eq:triangular-W}
  (I\otimes U\otimes L^T)\ket{W} = (a+d)\ket{000} + \ket{W} = \left( \pmm{1&a+d\\0&1}\otimes I \otimes I\right)\ket{W},
 \end{equation}
 so the diagram can be transformed to the one of Figure~\ref{fig:sym-W}b, where $A':=A\smm{1&a+d\\0&1}$.
 
 \begin{figure}
  \centering
  (a) \begin{tikzpicture}[scale=1.5]
	\begin{pgfonlayer}{nodelayer}
		\node [style=hollown] (0) at (0, 1.25) {};
		\node [style=hollown] (1) at (-2.5, -2.25) {};
		\node [style=hollown] (2) at (2.5, -2.25) {};
		\node [style=none] (3) at (0, 3) {};
		\node [style=none] (4) at (-3.75, -0.5) {};
		\node [style=none] (5) at (3.75, -0.5) {};
		\node [style=none] (6) at (-3.75, 3) {};
		\node [style=none] (7) at (3.75, 3) {};
		\node [style=map] (8) at (-3.75, 2.25) {$A$};
		\node [style=map] (13) at (3.75, 2.25) {$A$};
		\node [style=map] (14) at (0, 2.25) {$A$};
		\node [style=transposemap] (11) at (0.75, -2) {$L$};
		\node [style=unitary] (17) at (0, -1) {$D$};
		\node [style=map] (10) at (-0.75, -2) {$U$};
		\node [style=transposemap] (9) at (-1.75, -1.25) {$L$};
		\node [style=unitary] (18) at (-1.25, -0.375) {$D$};
		\node [style=transposemap] (15) at (-0.75, 0.5) {$U$};
		\node [style=map] (16) at (0.75, 0.5) {$L$};
		\node [style=unitary] (19) at (1.25, -0.375) {$D$};
		\node [style=map] (12) at (1.75, -1.25) {$U$};
	\end{pgfonlayer}
	\begin{pgfonlayer}{edgelayer}
		\draw [in=-90, out=152, looseness=1.00] (1) to (4.center);
		\draw [in=-90, out=28, looseness=1.00] (2) to (5.center);
		\draw (3.center) to (0);
		\draw (6.center) to (4.center);
		\draw (7.center) to (5.center);
		\draw (1) to (9);
		\draw (2) to (12);
		\draw (9) to (15);
		\draw [bend left=15, looseness=1.00] (15) to (0);
		\draw [bend left=15, looseness=1.00] (0) to (16);
		\draw (16) to (12);
		\draw [bend right=60, looseness=1.00] (1) to (10);
		\draw [bend left=60, looseness=1.00] (2) to (11);
		\draw [bend left=90, looseness=1.50] (10) to (11);
	\end{pgfonlayer}
\end{tikzpicture} \qquad (b) \begin{tikzpicture}[scale=1.5]
	\begin{pgfonlayer}{nodelayer}
		\node [style=hollown] (0) at (0, 1.25) {};
		\node [style=hollown] (1) at (-2.5, -2.25) {};
		\node [style=hollown] (2) at (2.5, -2.25) {};
		\node [style=none] (3) at (0, 3) {};
		\node [style=none] (4) at (-3.75, -0.5) {};
		\node [style=none] (5) at (3.75, -0.5) {};
		\node [style=none] (6) at (-3.75, 3) {};
		\node [style=none] (7) at (3.75, 3) {};
		\node [style=map] (8) at (-3.75, 2.25) {$A'$};
		\node [style=map] (13) at (3.75, 2.25) {$A'$};
		\node [style=map] (14) at (0, 2.25) {$A'$};
		\node [style=unitary] (17) at (0, -2.25) {$D$};
		\node [style=unitary] (18) at (-1.25, -0.375) {$D$};
		\node [style=unitary] (19) at (1.25, -0.375) {$D$};
	\end{pgfonlayer}
	\begin{pgfonlayer}{edgelayer}
		\draw [in=-90, out=152, looseness=1.00] (1) to (4.center);
		\draw [in=-90, out=28, looseness=1.00] (2) to (5.center);
		\draw (3.center) to (0);
		\draw (6.center) to (4.center);
		\draw (7.center) to (5.center);
		\draw (1) to (18);
		\draw (2) to (19);
		\draw [bend left=15, looseness=1.00] (18) to (0);
		\draw [bend left=15, looseness=1.00] (0) to (19);
		\draw (1) to (17);
		\draw (2) to (17);
	\end{pgfonlayer}
\end{tikzpicture}
  \caption{(a) The symmetrisation gadget after the matrix $C^T B$ has been converted into $LDU$ decomposition, where $L$ is lower triangular, $D$ is diagonal (and hence symmetric), and $U$ is upper triangular.
  (b) If the white vertices are assigned the function $\ONE_3$, the triangular matrices can be replaced by a different matrix on the outer leg, which is absorbed into $A$ to form $A'$.}
 \label{fig:sym-W}
 \end{figure}
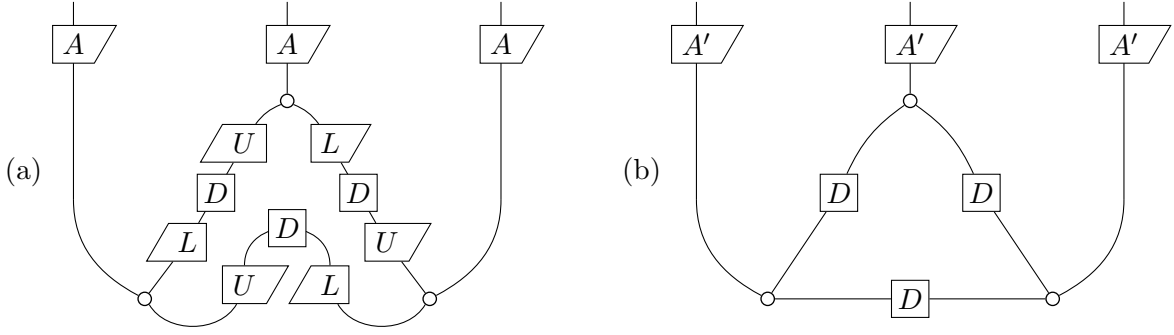
 
 The effective function of the gadget in Figure~\ref{fig:sym-W}b, which we will denote $g(x_1,x_2,x_3)$, equals
 \[
  \sum A_{x_1,y_1}' A_{x_2,y_2}' A_{x_3,y_3}' \ONE_3(y_1,z_1,w_1) \ONE_3(y_2,z_2,w_2)
  \ONE_3(y_3,z_3,w_3) D_{z_1,w_2} D_{z_2,w_3} D_{z_3,w_1},
 \]
 where the sum is over all $y_1,y_2,y_3,z_1,z_2,z_3,w_1,w_2,w_3\in\{0,1\}$.
 Let $g':=(A')^{-1}\circ g$, then
 \[
  g'(y_1,y_2,y_3)
  = \sum \ONE_3(y_1,z_1,z_3) \ONE_3(y_2,z_2,z_1) \ONE_3(y_3,z_3,z_2) D_{z_1,z_1} D_{z_2,z_2} D_{z_3,z_3},
 \]
 where the sum is now over $z_1,z_2,z_3\in\{0,1\}$ and we have used the property that $D$ is diagonal.

 Recall that $\ONE_3$ is the `perfect matchings' constraint.
 This means the gadget for $g'$ is a symmetric matchgate on three vertices, where every internal edge that is {in the matching} contributes a factor $c$ to the weight, and every internal edge that is not {in the matching} contributes a factor $b$ to the weight.
 Clearly, $g'(1,1,1)=b^3$ since if the three external edges are {in the matching}, there can be no internal edges {in the matching}.
 As a matchgate, $g'$ must satisfy a parity condition: in particular, given the odd number of vertices, it is 0 on inputs of even Hamming weight.
 The remaining value on inputs of Hamming weight 1, having one external edge and thus one internal edge to make a matching, is $b^2c$.
 Combining these, we have $g' = [0,\; b^2c,\; 0,\; b^3]$.
 Note that $\ket{g}$ and $\ket{g'}$ are by definition equivalent under SLOCC, i.e.\ $g$ and $g'$ are in the same entanglement class.
 We may therefore reason about $g'$ instead of $g$.

 By Lemma~4.5, $g'$ has GHZ type if and only if
 \[
  0\neq (0 - 0)^2 - 4\left((b^2c)^2 - 0\right)\left(0 - b^2c b^3\right) = 4 b^9 c^3.
 \]
 Yet $b,c\neq 0$, so $g'$ always has GHZ type.
 The gadget is planar.}

 \textit{Case~2}: If $(C^T B)_{0,0} = 0$, the PLDU decomposition would require $P=X$ and the argument of Case~1 does not work.
 Yet we will now show that we can avoid this situation: under the assumptions made in the lemma, there is always at least one way of placing a rotated copy of $f$ in the symmetrisation gadget that leads to Case~1.

 Assume for a contradiction that we have $(C^T B)_{0,0} = 0$ for all three rotations of~$f$.

 We will first consider just one of the resulting symmetrisation gadgets. Instead of decomposing as above, let $M := C^T B = \smm{0&a\\b&c}$, for some $a,b,c\in\AA$ with $a,b\neq 0$.
 The corresponding diagram is Figure~\ref{fig:symmetrising_GHZ}b.
 The function of the gadget can thus be written as
 \begin{multline*}
  g(x_1,x_2,x_3) := \sum A_{x_1,y_1} A_{x_2,y_2} A_{x_3,y_3} \ONE_3(y_1,z_1,w_1) \\ \ONE_3(y_2,z_2,w_2) \ONE_3(y_3,z_3,w_3) M_{z_1,w_2} M_{z_2,w_3} M_{z_3,w_1},
 \end{multline*}
 where the sum is over all $y_1,y_2,y_3,z_1,z_2,z_3,w_1,w_2,w_3\in\{0,1\}$.
 Define $g' = A^{-1}\circ g$ so that
 \[
  g'(y_1,y_2,y_3)
  = \sum \ONE_3(y_1,z_1,w_1) \ONE_3(y_2,z_2,w_2) \ONE_3(y_3,z_3,w_3) M_{z_1,w_2} M_{z_2,w_3} M_{z_3,w_1}
 \]
 where the sum is now over all $z_1,z_2,z_3,w_1,w_2,w_3\in\{0,1\}$.

 Let us consider the conditions under which this function takes non-zero values.
 Since $M_{0,0}=0$, all non-zero terms in the sum must satisfy
 \[
  (z_1=1 \vee w_2=1) \wedge (z_2=1 \vee w_3=1) \wedge (z_3=1 \vee w_1=1).
 \]
 At the same time, by the perfect matching constraints, all non-zero terms must satisfy
 \[
  (z_1\neq 1 \vee w_1\neq 1) \wedge (z_2\neq 1 \vee w_2\neq 1) \wedge (z_3\neq 1 \vee w_3\neq 1).
 \]
 Combining these, all non-zero terms in the sum must satisfy either $z_1=z_2=z_3=1$ and $w_1=w_2=w_3=0$, or $z_1=z_2=z_3=0$ and $w_1=w_2=w_3=1$.
 Then furthermore, $g'$ is non-zero only if $y_1=y_2=y_3=0$ since the perfect matching constraints in fact require that exactly one of $y_k,z_k,w_k$ is 1 for each $k$.
 Therefore $g'(y_1,y_2,y_3) = (a^3+b^3) \dl_0(y_1)\dl_0(y_2)\dl_0(y_3)$ and the gadget is degenerate.

 Now, by cyclic permutation, for each initial function $f$, there are three possible planar symmetric gadgets.
 The symmetrisation procedure fails only if all three of these gadgets are degenerate.
 We will now analyse under which conditions that may happen (and show that those conditions contradict the assumptions of the lemma).

 Recall the original function has the form $(A\otimes B\otimes C)\ket{W}$.
 By the PLDU decomposition, each matrix $N\in\{A,B,C\}$ satisfies $N = P_N L_N D_N U_N$, where $P_N \in\{I,X\}$ is a permutation, $L_N$ is lower triangular, $D_N$ is diagonal, and $U_N$ is upper triangular; with both $L_N$ and $U_N$ having 1s on the diagonal.
 As shown in \eqref{eq:triangular-W}, the upper triangular components can all be moved to the dangling leg and do not need to be considered.
 Moreover, the diagonal components will not affect whether $M_{0,0}$ is zero.

 Thus it suffices to consider just $P_N L_N$ for each $N$; let $d_N\in\AA$ and write $L_N := \smm{1&0\\d_N&1}$.
 To determine whether $M_{0,0}$ vanishes, we need only consider the top left element of
 \[
  \pmm{1&d_C\\0&1} P_C P_B \pmm{1&0\\d_B&1} =
  \begin{cases}
   \pmm{1 + d_B d_C & d_C \\ d_B & 1} &\text{if } P_C P_B = I \\
   \pmm{d_B + d_C & 1 \\ 1 & 0} &\text{if } P_C P_B = X
  \end{cases}
 \]
 where $P_C P_B = I \Leftrightarrow P_C = P_B$ and $P_C P_B = X \Leftrightarrow P_C \neq P_B$.

 All three gadgets being degenerate implies one of the following two subcases:
 \begin{itemize}
  \item $P_A = P_B = P_C$ and $0 = 1 + d_A d_B = 1 + d_B d_C = 1 + d_A d_C$, which implies $d_A = d_B = d_C = \pm i$, or
  \item $P_A = P_B \neq P_C$ and $0 = 1 + d_A d_B = d_B + d_C = d_A + d_C$ (up to permutations of $A,B,C$), which implies $d_A = d_B = - d_C = \pm i$.
 \end{itemize}
 Note that $K$ and $KX$ have the following PLDU decompositions with $P\in\{I,X\}$:
 \begin{equation}
  \begin{pmatrix} 1 & 1 \\ \pm i & \mp i \end{pmatrix} = \begin{pmatrix} 1 & 0 \\ \pm i & 1 \end{pmatrix} \begin{pmatrix} 1 & 0 \\ 0 & \mp 2i \end{pmatrix} \begin{pmatrix} 1 & 1 \\ 0 & 1 \end{pmatrix} = X \begin{pmatrix} 1 & 0 \\ \mp i & 1 \end{pmatrix} \begin{pmatrix} \pm i & 0 \\ 0 & 2 \end{pmatrix} \begin{pmatrix} 1 & -1 \\ 0 & 1 \end{pmatrix}.
 \end{equation}
 Thus, if for example $P_N=I$ and $d_N = i$, we have
 \[
  N = \pmm{1&0\\i&1} D_N U_N = K \pmm{1&-1\\0&1} \pmm{1&0\\0&i/2} D_N U_N = K U'_N
 \]
 where $U'_N$ is a product of upper triangular (and diagonal) matrices and thus is upper triangular itself.
 Similarly, if $P_N = X$ and $d_N = -i$, we have
 \[
  N = X\pmm{1&0\\-i&1} D_N U_N = K \pmm{1&1\\0&1} \pmm{-i&0\\0&1/2} D_N U_N = K U''_N
 \]
 where again $U''_N$ is upper triangular.
 Analogous arguments apply for the other combinations of permutations and values $\pm i$.

 Note that, if $P_A = P_B \neq P_C$ and $d_A = d_B = - d_C$ (or some permutation thereof), the difference in the permutation matrix and the difference in sign cancel out in the sense that, in the above rewriting process, we get either $K$ for all three matrices or $KX$ for all three matrices.
 Hence if all three gadgets are degenerate, then $\ket{f} = K\t{3} (U_A'\otimes U_B' \otimes U_C')\ket{W}$ or $\ket{f} = (KX)\t{3} (U_A'\otimes U_B' \otimes U_C')\ket{W}$ for some upper triangular matrices $U_A',U_B',U_C'$.
 Now,
 \begin{multline*}
  \left(\pmm{a_{00}&a_{01}\\0&a_{11}}\otimes I\otimes I\right) (b_{000}\ket{000} + b_{001}\ket{001} + b_{010}\ket{010} + b_{100}\ket{100}) \\
  = (a_{00}b_{000} + a_{01}b_{111})\ket{000} + a_{00}b_{001}\ket{001} + a_{00}b_{010}\ket{010} + a_{11}b_{100}\ket{100}
 \end{multline*}
 and similarly for applying upper triangular matrices in other places.
 Thus, for any upper triangular matrices $U_1,U_2,U_3$, we have
 $(U_1\otimes U_2\otimes U_3)\ket{W}\in\cM$.
 This means that if all three gadgets are degenerate, then $f\in (K\circ\cM)\cup(KX\circ\cM)$, contradicting the assumption of the lemma.

 In other words, if the assumption $f\notin (K\circ\cM)\cup(KX\circ\cM)$ holds, for at least one way to place a rotated copy of $f$ in the symmetrisation gadget, Case~1 applies and hence the symmetrisation procedure works for this orientation of $f$.
 This concludes Case~2.
\end{proof}

{The following lemma shows an analogous result for the case where the ternary function $f\in K\circ\cM$ or $f\in KX\circ\cM$.
In that case, some additional support from a binary function outside the respective set is needed; we will show in the next section that such a binary function can be realised unless all functions are in $\ang{K\circ\cM}$ or $\ang{KX\circ\cM}$.}

\begin{lemma}\label{lem:W_symmetrise-K}
 {Let $\tM$ be one of $K\circ\cM$ and $KX\circ\cM$.
 Suppose $f\in\allf_3\cap\tM$ and $h\in\allf_2\setminus\ang{\tM}$, then there exists a non-degenerate symmetric ternary function $g\in S(\{f,h\})$.
 In both} cases, $g$ has GHZ type and is realisable by a planar gadget.
\end{lemma}
\begin{proof}
 {Note that by Lemma~\ref{lem:family-types} and the property that holographic transformations do not affect the entanglement class, $f$ has $W$ type in both cases, i.e.\ $\ket{f}=(A\otimes B\otimes C)\ket{W}$ for some matrices $A,B,C\in\GL$.

 Now first suppose} $f\in K\circ\cM$ and $h\notin\ang{K\circ\cM}$.
 Note that any binary function $h$ with this property is non-degenerate (and thus non-decomposable), since $\allf_1\sse K\circ\cM$.
 Let $f':=K^{-1}\circ f\in\cM$ and $h':=K^{-1}\circ h\notin\cM$, then we can write
 \[
  f' = \begin{pmatrix}f_{000}&f_{001}&f_{010}&0\\f_{100}&0&0&0\end{pmatrix} \quad\text{and}\quad
  h' = \begin{pmatrix}h_{00}&h_{01}\\h_{10}&h_{11}\end{pmatrix},
 \]
 where $f_{001}f_{010}f_{100}\neq 0$ by non-decomposability of $f'$, $h_{00}h_{11}-h_{01}h_{10}\neq 0$ by non-de\-com\-posa\-bi\-li\-ty of $h'$, and $h_{11}\neq 0$ because $h'\notin \cM$.
 
 \begin{figure}
  \centering
  \begin{tikzpicture}
	\begin{pgfonlayer}{nodelayer}
		\node [style=solidn] (0) at (0, -0.25) {};
		\node [style=solidn] (1) at (-0.75, 0.25) {};
		\node [style=none] (2) at (-1, 0.75) {};
		\node [style=none] (3) at (0, 0.75) {};
		\node [style=none] (4) at (1, 0.75) {};
	\end{pgfonlayer}
	\begin{pgfonlayer}{edgelayer}
		\draw [bend right=15, looseness=1.00] (2.center) to (1);
		\draw [bend right=15, looseness=1.00] (1) to (0);
		\draw (0) to (3.center);
		\draw [bend right, looseness=1.00] (0) to (4.center);
	\end{pgfonlayer}
\end{tikzpicture}
  \caption{Gadget for constructing a ternary function that is not in $K\circ\cM$ (or $KX\circ\cM$). The degree-3 vertex is assigned the function $f$ and the degree-2 vertex is assigned the function $h$, with the second input of $h$ connected to the first input of $f$.}
  \label{fig:not_KcM}
 \end{figure}
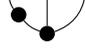
 
 Consider the gadget in Figure~\ref{fig:not_KcM}, where the second input of $h$ is connected to the first input of $f$.
 The effective function associated with this gadget is
 \begin{align*}
  f''(x_1,x_2,x_3) &= \sum_{y\in\{0,1\}} h(x_1,y) f(y,x_2,x_3) \\
  &= \sum_{y,\vc{z}{5}\in\{0,1\}} K_{x_1,z_1}K_{y,z_2} h'(z_1,z_2) K_{y,z_3} K_{x_2,z_4} K_{x_3,z_5} f'(z_3,z_4,z_5) \\
  &= \sum_{\vc{z}{5}\in\{0,1\}} K_{x_1,z_1} K_{x_2,z_4} K_{x_3,z_5} h'(z_1,z_2) f'(z_3,z_4,z_5) \sum_{y\in\{0,1\}} K_{y,z_2} K_{y,z_3} \\
  &= \sum_{z_1,z_4,z_5\in\{0,1\}} K_{x_1,z_1} K_{x_2,z_4} K_{x_3,z_5} \sum_{z_2,z_3\in\{0,1\}} 2\cdot h'(z_1,z_2) f'(z_3,z_4,z_5) \NEQ(z_2,z_3).
 \end{align*}
 Let $f''':=\frac{1}{2}(K^{-1}\circ f'')$ or equivalently $f''=2(K\circ f''')$, then
 \[
  f''' = \begin{pmatrix} f_{000}h_{01}+f_{100}h_{00} & f_{001}h_{01} & f_{010}h_{01} & 0 \\ f_{000}h_{11}+f_{100}h_{10} & f_{001}h_{11} & f_{010}h_{11} & 0 \end{pmatrix}.
 \]
 Note that $f''\in K\circ\cM$ if and only if $f'''\in\cM$, and $f''\in KX\circ\cM$ if and only if $f'''\in X\circ \cM$.
 We now show $f''\notin (K\circ\cM)\cup (KX\circ\cM)$.
 
 Assume for a contradiction that $f''\in K\circ\cM$, i.e.\ $f'''\in\cM$.
 This implies that $0=f'''(1,0,1)=f_{001}h_{11}$ and $0=f'''(1,1,0)=f_{010}h_{11}$.
 But, as stated above, $f$ being non-decomposable and $h$ not being in $K\circ\cM$ imply that $f_{001},f_{010}$ and $h_{11}$ must be non-zero.
 Therefore $f'''\notin\cM$ and $f''\notin K\circ\cM$.
 
 Similarly, assume for a contradiction that $f''\in KX\circ\cM$, i.e.\ $f'''\in X\circ\cM$.
 This implies
 \begin{align*}
  0 &= f'''(0,0,0) = f_{000}h_{01}+f_{100}h_{00} \\
  0 &= f'''(0,0,1) = f_{001}h_{01} \\
  0 &= f'''(0,1,0) = f_{010}h_{01} \\
  0 &= f'''(1,0,0) = f_{000}h_{11}+f_{100}h_{10}.
 \end{align*}
 Now, the condition $f_{001}h_{01}=0$ implies $h_{01}=0$ since $f_{001}\neq 0$ by non-decomposability of $f$.
 Then the first equality reduces to $0=f_{100}h_{00}$, which implies $h_{00}=0$ since $f_{100}\neq 0$ by non-decomposability of $f$.
 Yet $h_{01}=h_{00}=0$ would imply that $h$ is degenerate, contradicting the assumption of the lemma.
 Thus $f'''\notin X\circ\cM$ and $f''\notin KX\circ\cM$.
 
 Hence, we may replace $f$ by $f''$ and proceed as in Lemma~\ref{lem:W_symmetrise}.
 Since $f''\in S(\{f,h\})$ and $f''$ is realisable by a planar gadget, the function $g\in S(\{f''\})$ that results from applying this lemma to $f''$ is in $S(\{f,h\})$ and $g$ is again realisable by a planar gadget.
 
 {Now suppose $f\in KX\circ\cM$ and $h\notin\ang{KX\circ\cM}$. 
 As before, let $f':=K^{-1}\circ f$ and $h':=K^{-1}\circ  h$, then $f'\in X\circ\cM$ and $h'\notin\ang{X\circ\cM}$.
 The properties of these functions differ from the ones in the previous case by bit flips on all inputs.
 Otherwise the argument is the same.} 
\end{proof}

\subsection{Realising binary functions}
\label{s:binary}

We have shown in the previous section that it is possible to realise a non-decomposable ternary symmetric function from an arbitrary non-decomposable ternary function under some mild further assumptions.
Now, we show that if the full set of functions $\cF$ is not a subset of $\ang{K\circ\cM}$, there exists a planar gadget in $S(\cF\cup\{\dl_0,\dl_1,\dl_+,\dl_-\})$ whose effective function $g$ is binary, symmetric, non-decomposable and satisfies $g\notin\ang{K\circ\cM}$.
An analogous result holds with $KX$ instead of $K$.

\begin{lemma}\label{lem:not_in_cM}
 Suppose $f\in\allf_n$ satisfies $f(1,1,a_3\zd a_n)\neq 0$ for some $a_3\zd a_n\in\{0,1\}$.
 Suppose furthermore that there exist functions $u_3\zd u_n\in\cU$ such that setting 
 \[
  f'(x_1,x_2) := \sum_{x_3\zd x_n\in\{0,1\}} f(\vc{x}{n}) \prod_{j=3}^n u_j(x_j)
 \]
 implies $f'(0,0)f'(1,1)-f'(0,1)f'(1,0)\neq 0$.
 Then $f\notin\ang{\cM}$.
 Conversely, if $f\in\allf_n\setminus\ang{\cM}$, then there exists some permutation $\rho:[n]\to [n]$ such that $f_\rho$ satisfies the above properties.
\end{lemma}
\begin{proof}
 Consider the first part of the lemma and assume for a contradiction that $f\in\ang{\cM}$.
 Then there exists a decomposition of $f$ where each factor is in $\cM$.
 In particular, the first and second arguments of $f$ must belong to different factors: if they belonged to the same factor, that factor would be a function which is non-zero on an input of Hamming weight at least 2, so it could not be in $\cM$.
 Hence there exists some permutation $\rho:\{3,4\zd n\}\to\{3,4\zd n\}$ and $k\in\{3,4\zd n\}$ such that
 \[
  f_\rho(\vc{x}{n}) = f_1(x_1,x_{k+1}\zd x_n)f_2(x_2\zd x_k),
 \]
 where $f_1,f_2$ may be further decomposable.
 But then
 \[
  f'(x_1,x_2) = \sum_{x_3\zd x_n\in\{0,1\}} f_1(x_1,x_{k+1}\zd x_n)f_2(x_2\zd x_k) \prod_{j=3}^n u_j(x_j) = w_1(x_1) w_2(x_2),
 \]
 where
 \begin{align*}
  w_1(x_1) &:= \sum_{x_{k+1}\zd x_n\in\{0,1\}} f_1(x_1,x_{k+1}\zd x_n) \prod_{j=k+1}^n u_j(x_j) \\
  w_2(x_2) &:=\sum_{x_3\zd x_k\in\{0,1\}} f_2(x_2\zd x_k) \prod_{j=3}^k u_j(x_j).
 \end{align*}
 This implies that $f'(0,0)f'(1,1)-f'(0,1)f'(1,0)=0$, contradicting the assumption in the lemma.
 Thus $f$ cannot be in $\ang{\cM}$.
 
 For the second part of the lemma, assume $f\in\allf_n\setminus\ang{\cM}$.
 In particular, $f$ is not the all-zero function, so it has a decomposition.
 If all the factors in the decomposition of $f$ are in $\cM$, then $f\in\ang{\cM}$.
 Hence there exists a non-decomposable factor $g$ of $f$ which is not in $\cM$.
 Without loss of generality, assume that $g$ contains exactly the first $k$ arguments of $f$ for some $k\in [n]$, otherwise permute the arguments.
 As $\cM$ contains all unaries, $g$ must be a function of arity $k\geq 2$.
 Furthermore, there must exist some bit string $\ba\in\{0,1\}^k$ of Hamming weight at least 2 such that $g(\ba)\neq 0$.
 Without loss of generality, assume $a_1=a_2=1$, otherwise permute the arguments.
 Then $f$ satisfies the first property.
 If $f=g$, i.e.\ $f$ itself is non-decomposable, the second property is satisfied by Proposition~\ref{prop:popescu-rohrlich_gadget}.
 Otherwise, by Lemma~\ref{lem:decomposable}, there exist $u_{k+1}\zd u_n\in\{\dl_0,\dl_1\}$ such that
 \[
  g(\vc{x}{k}) = \sum_{x_{k+1}\zd x_n\in\{0,1\}} f(\vc{x}{n}) \prod_{j=k+1}^n u_j(x_j).
 \]
 We then find that the second property is satisfied by combining this equation with an application of Proposition~\ref{prop:popescu-rohrlich_gadget} to $g$.
\end{proof}

\begin{lemma}\label{lem:binary_notin_KcircM}
 {Let $\tM$ be one of $K\circ\cM$ and $KX\circ\cM$.}
 Suppose $f\in\allf\setminus\ang{\tM}$.
 Then there exists $g\in S(\{f,\dl_0,\dl_1,\dl_+,\dl_-\})$ which is binary, non-decomposable, and satisfies $g\notin\ang{\tM}$.
 Furthermore, this function is realised by a planar gadget.
\end{lemma}
\begin{proof}
 {First, suppose $\tM = K\circ\cM$.}
 Note that $\allf_1\sse K\circ\cM\sse\ang{K\circ\cM}$, so $f\notin\ang{K\circ\cM}$ implies $\ari(f)\geq 2$.
 Furthermore, $\ang{\allf_1}\sse\ang{K\circ\cM}$, thus any binary function which is not in $\ang{K\circ\cM}$ must be non-decomposable.

 We will show that we can realise a gadget like that in Figure~\ref{fig:pr}a, whose effective function $g$ is non-decomposable and satisfies $g\notin\ang{K\circ\cM}$.
 As the gadget consists of a number of unary functions connected to one larger-arity function, planarity is assured even if we arbitrarily permute the arguments of $f$ at intermediate stages of the proof.
 
 The result will be proved by induction on the arity of $f$.
 
 The base case is $\ari(f)=2$, in which case we may take $g=f$.
 This function is binary, non-decomposable, a gadget in $S(\{f,\dl_0,\dl_1,\dl_+,\dl_-\})$ and it is not in $\ang{K\circ\cM}$.
 Furthermore, the gadget is trivially planar, so $g$ has all the required properties.
 
 Now suppose the desired result has been proved for functions of arity at most $n$, and suppose $\ari(f)=n+1$.
 We distinguish cases according to the decomposability of $f$.
 
 \textbf{Case~1}: Suppose $f$ is decomposable and consider some decomposition of $f$.
 If all factors are in $\ang{K\circ\cM}$, then $f$ is in $\ang{K\circ\cM}$, so there must be a factor $f'\notin\ang{K\circ\cM}$.
 By Lemma~\ref{lem:decomposable}, $f'\in S(\{f,\dl_0,\dl_1,\dl_+,\dl_-\})$.
 Thus, we are done by the inductive hypothesis.
 
 \textbf{Case~2}: Suppose $f$ is non-decomposable.
 We will show that there exists $u\in\{\dl_0,\dl_1,\dl_+,\dl_-\}$ and $k\in [n+1]$ such that $\sum_{x_k\in\{0,1\}}f(\vc{x}{n+1})u(x_k)$ is not in $\ang{K\circ\cM}$.
 It is simpler to work with $\ang{\cM}$ than $\ang{K\circ\cM}$, so let $h=K^{-1}\circ f$; then $h\notin\ang{\cM}$.
 Thus, by Lemma~\ref{lem:not_in_cM}, there exists a permutation $\rho:[n+1]\to [n+1]$ such that $h_\rho(1,1,a_3\zd a_{n+1})\neq 0$ for some $a_3\zd a_{n+1}\in\{0,1\}$ and there exist functions $u_3\zd u_{n+1}\in\cU$ such that letting
 \begin{equation}\label{eq:h-prime}
  h'(x_1,x_2) := \sum_{x_3\zd x_{n+1}\in\{0,1\}} h_\rho(\vc{x}{n+1}) \prod_{j=3}^{n+1} u_j(x_j)
 \end{equation}
 implies $h'(0,0)h'(1,1)-h'(0,1)h'(1,0)\neq 0$.
 Note that
 \begin{align*}
  h'(x_1,x_2) &= \sum_{x_3\zd x_{n+1}\in\{0,1\}} (K^{-1}\circ f_\rho)(\vc{x}{n+1}) \prod_{j=3}^{n+1} u_j(x_j) \\
  &= K^{-1}\circ\left(\sum_{x_3\zd x_{n+1}\in\{0,1\}} f_\rho(\vc{x}{n+1}) \prod_{j=3}^{n+1} ((K^{-1})^T\circ u_j)(x_j)\right)
 \end{align*}
 and holographic transformations do not affect whether a function is decomposable.
 Thus, by applying Proposition~\ref{prop:popescu-rohrlich_gadget} to $f_\rho$, we find that it suffices to take
 \[
  (K^{-1})^T\circ u_j\in\{\dl_0,\dl_1,\dl_+,\dl_-\} \quad\Leftrightarrow\quad u_j\in\{\dl_+,\dl_-,\dl_i,\dl_{-i}\},
 \]
 where $\dl_i:=[1,i]$ and $\dl_{-i}:=[1,-i]$ and we have ignored some scalar factors which, by Lemma~\ref{lem:scaling}, do not affect the complexity.
 
 For each $v\in\{\dl_+,\dl_-,\dl_i,\dl_{-i}\}$ define
 \[
  h_v(\vc{x}{n}) := \sum_{x_{n+1}\in\{0,1\}} h_\rho(\vc{x}{n+1}) v(x_{n+1}).
 \]
 We now argue that for at least one $v\in\{\dl_+,\dl_-,\dl_i,\dl_{-i}\}$, the function $h_v$ is not in $\ang{\cM}$.
 Write $v=[1,\alpha]$, where $\alpha\in\{1,-1,i,-i\}$.
 
 First, consider the value $h_v(1,1,a_3\zd a_{n})$, where $a_3\zd a_{n}$ are the above values resulting from applying Lemma~\ref{lem:not_in_cM} to $h$.
 Then
 \begin{equation}\label{eq:linear_alpha}
  h_v(1,1,a_3\zd a_{n}) = h_\rho(1,1,a_3\zd a_{n},0)+\alpha h_\rho(1,1,a_3\zd a_{n},1)
 \end{equation}
 is a linear polynomial in the variable $\alpha$.
 Furthermore, this polynomial is not identically zero since $h_\rho(1,1,a_3\zd a_{n+1})\neq 0$.
 Hence, this expression vanishes for at most one value of $\alpha$, i.e.\ one of the $h_v$.
 
 Secondly, let
 \[
  h_u'(x_1,x_2) := \sum_{x_3\zd x_{n+1}\in\{0,1\}} h_\rho(\vc{x}{n+1}) v(x_{n+1}) \prod_{j=3}^{n} u_j(x_j),
 \]
 where $u_j\in\{\dl_+,\dl_-,\dl_i,\dl_{-i}\}$ are the unary functions determined by applying Lemma~\ref{lem:not_in_cM} to $h$ as in \eqref{eq:h-prime}.
 Then
 \begin{equation}\label{eq:quadratic_alpha}
  h_v'(0,0)h_v'(1,1)-h_v'(0,1)h_v'(1,0)
 \end{equation}
 is a quadratic polynomial in $\alpha$ which is not identically zero.
 Thus this polynomial vanishes for at most two distinct values of $\alpha$, i.e.\ two of the $h_v$.
 
 Therefore, there must be at least one $h_v$ such that both \eqref{eq:linear_alpha} and \eqref{eq:quadratic_alpha} are non-zero.
 By Lemma~\ref{lem:not_in_cM}, this $h_v$ is not in $\ang{\cM}$.
 Furthermore,
 \[
  (K\circ h_v)(\vc{x}{n}) = \sum_{x_{n+1}\in\{0,1\}} f_\rho(\vc{x}{n+1}) v'(x_{n+1}),
 \]
 where $v'=(K^{-1})^T\circ v\in\{\dl_0,\dl_1,\dl_+,\dl_-\}$.
 Thus, $K\circ h_v\in S(\{f,\dl_0,\dl_1,\dl_+,\dl_-\})\setminus\ang{K\circ\cM}$ and we are done by the inductive hypothesis.
 
 {Now suppose $\tM = KX\circ\cM$.
 The argument of Case~1 is the same as before, with $KX\circ\cM$ instead of $K\circ\cM$.
 In Case~2, again define $h:=K^{-1}\circ f$, then the properties of this function differ from the ones written out above by a bit flip on all inputs.
 Note that, up to scalar factor, the set $\{\dl_+,\dl_-,\dl_i,\dl_{-i}\}$ is invariant under bit flip.
 Thus, the argument is analogous to before.}
\end{proof}

\subsection{Interreducing planar holant problems and planar counting CSPs}
\label{s:interreducing_planar}

The interreducibility of certain bipartite holant problems and counting CSPs, as in Theorem~\ref{thm:GHZ-state}, will be a crucial ingredient in our hardness proof.
We therefore need to ensure this result holds in the planar case.

Recall that $\#\mathsf{R_3\text{-}CSP}\left(\cF\right) \equiv_T \holp{\cF\mid\{\EQ_1,\EQ_2,\EQ_3\}}$.
The following lemma will be useful.

\begin{lemma}[{\cite[Lemma~6.1]{cai_complexity_2014}}]
 Let $g\in\allf_2$ be a non-degenerate binary function.
 Then, for any finite $\cF\sse\allf$ containing $g$, we have $\#\mathsf{R_3\text{-}CSP}\left(\cF\cup\{\EQ_2\}\right) \leq_T \#\mathsf{R_3\text{-}CSP}\left(\cF\right)$.
\end{lemma}

In fact, this result can straightforwardly be extended to the following.

\begin{lemma}\label{lem:interpolate_equality2}
 Let $g\in\allf_2$ be a non-degenerate binary function.
 Suppose $\cG_1,\cG_2\sse\allf$ are finite, then
 \begin{multline*}
  \plhol\left(\cG_1\cup\{g,\EQ_2\}\mid\cG_2\cup\{\EQ_1,\EQ_2,\EQ_3\}\right) \\
  \leq_T \plhol\left(\cG_1\cup\{g\}\mid\cG_2\cup\{\EQ_1,\EQ_2,\EQ_3\}\right).
 \end{multline*}
\end{lemma}
\begin{proof}
 The proof is analogous to that of \cite[Lemma~6.1]{cai_complexity_2014}, noting that the constructions used in gadgets and in polynomial interpolation are planar, and that adding more functions on the RHS does not affect the argument.
\end{proof}

The following result about a polynomial-time reduction from any counting CSP to some problem in \numP{} is known, see e.g.\ \cite[p.~212]{cai_complexity_2017}. Nevertheless, we have not been able to find an explicit proof, so we give one here for completeness.
This proof is based on a similar reduction for graph homomorphism problems \cite[Lemma~7.1]{cai_graph_2010}.

First, we define the field within which we will be working, and the computational problem to which the counting CSP will be reduced.
For any finite $\cF\sse\allf$, let $A_\cF$ be the set of algebraic complex numbers that appear as a value of some function in $\cF$:
\begin{equation}\label{eq:function_values}
 A_\cF := \left\{ z\in\AA \,\middle|\, \exists f\in\cF \text{ and } \bx\in\{0,1\}^{\ari(f)} \text{ such that } f(\bx)=z \right\}.
\end{equation}
Since $\cF$ is fixed and finite, the set $A_\cF$ is also fixed and finite.
Let $\QQ(A_\cF)$ be the algebraic extension of the field of rational numbers by the numbers in $A_\cF$.
Note that, given an instance $(V,C)$ of $\csp(\cF)$, the weight $\wt_{(V,C)}(\xi)$ associated with any assignment $\xi:V\to\{0,1\}$ is in $\QQ(A_\cF)$, as is the total weight $Z_{(V,C)}$.
We may thus define the following counting problem:
 
 \begin{description}[noitemsep]
  \item[Name] $\mathsf{COUNT}(\cF)$
  \item[Instance] A tuple $((V,C),x)$, where $V$ is a finite set of variables, $C$ is a finite set of constraints over $\cF$, and $x\in\QQ(A_\cF)$.
  \item[Output] $\#_\cF((V,C), x) = \abs{ \{\text{assignments } \xi: V\to\{0,1\} \mid \wt_{(V,C)}(\xi) = x\} }$.
 \end{description}
 
$\mathsf{COUNT}(\cF)$ is the problem of counting the number of accepting paths of a Turing machine that accepts an input $((V,C),x,\xi)$ if and only if $\wt_{(V,C)}(\xi) = x$.
Here, $\xi:V\to\{0,1\}$ has to be an assignment of values to the variables $V$, otherwise the input is rejected.
Given an assignment $\xi$, computing the associated weight in the fixed algebraic extension field $\QQ(A_\cF)$ can be done in time polynomial in the size of $(V,C)$.
Furthermore, numbers within $\QQ(A_\cF)$ can be compared efficiently.
Therefore the Turing machine runs in time polynomial in the size of its input, and $\mathsf{COUNT}(\cF)$ is in \numP.

\begin{lemma}\label{lem:NCSP_to_numP}
 Let $\cF\sse\allf$ be finite.
 Then $\NCSP(\cF) \leq_T \mathsf{COUNT}(\cF)$.
\end{lemma}
\begin{proof}
 Consider an instance $(V,C)$ of $\csp(\cF)$, where $V$ is a finite set of variables and $C$ is a finite set of constraints over $\cF$.
 Since $\cF$ is a fixed finite set, its elements can be enumerated in some order $\vc{f}{m}$, where $m:=\abs{\cF}$ is a constant.
 For each $j\in [m]$, define $a_j:=\ari(f_j)$ as a short-hand. 
 Let $n:=\abs{C}$ be the number of constraints, and define the following set of algebraic complex numbers:
 \[
  {\mathcal{X}} = \left\{ \prod_{j\in[m]} \prod_{\bx_j\in\{0,1\}^{a_j}} (f_j(\bx_j))^{k_{j,\bx_j}} \;\middle|\; k_{j,\bx_j}\in\ZZ_{\geq 0} \text{ and } \sum_{j\in[m]} \sum_{\bx_j\in\{0,1\}^{a_j}} k_{j,\bx_j} = n \right\}.
 \]
 Each element of $\mathcal{X}$ is uniquely determined by the integers $k_{j,\bx_j}$, and there are $M:=\sum_{j\in[m]}2^{a_j}$ such integers.
 Thus, the elements of $\mathcal{X}$ are in bijection with the $M$-tuples of non-negative integers satisfying the property that the sum of all elements of the tuple is $n$.
 The set of all such $M$-tuples is exactly the support set of a multinomial distribution with $n$ trials and $M$ possible outcomes for each trial; therefore
 \[
  \abs{\mathcal{X}} = \binom{n+M-1}{M-1} \leq \frac{(n+M-1)^{M-1}}{(M-1)!},
 \]
 where the inequality uses the straightforward-to-derive bound $\binom{n}{k} \leq \frac{n^k}{k!}$ on binomial coefficients.
 Now the parameter $(M-1)$ is constant (it depends only on the properties of the fixed finite set $\cF$), so $\abs{\mathcal{X}}$ is polynomial in $n$ and thus polynomial in the instance size.
 
 Note that $\mathcal{X}\sse \QQ(A_\cF)$.
 Consider an element of $\mathcal{X}$ of the form
 \[
  \prod_{j\in[m]} \prod_{\bx_j\in\{0,1\}^{a_j}} (f_j(\bx_j))^{k_{j,\bx_j}}.
 \]
 The condition on the sum of the integers $k_{j,\bx_j}$, together with non-negativity, implies that at most $n$ of these integers are non-zero.
 Thus, each element of $\mathcal{X}$ can be computed in time polynomial in $n$.
 Since the size of $\mathcal{X}$ is also polynomial in $n$, this means the elements of $\mathcal{X}$ can be enumerated in time polynomial in $n$.
 
 Recall from Section~\ref{s:csp} that, for any assignment $\xi:V\to\{0,1\}$, we have $\wt_{(V,C)}(\xi) = \prod_{c\in C}f_c(\xi|_c)$, where $f_c$ is the function associated with the constraint $c$ and $\xi|_c$ is the restriction of the assignment $\xi$ to the scope of $c$.
 Now, for any $j\in [m]$ and $\bx_j\in\{0,1\}^{a_j}$, define $\kappa_{j,\bx_j} := \abs{\{c\in C \text{ such that } f_c=f_j \text{ and } \xi|_c=\bx_j\}}$, then
 \[
  \wt_{(V,C)}(\xi) = \prod_{c\in C}f_c(\xi|_c) = \prod_{j\in[m]} \prod_{\bx_j\in\{0,1\}^{a_j}} (f_j(\bx_j))^{\kappa_{j,\bx_j}}
 \]
 and
 \[
  \sum_{j\in[m]} \sum_{\bx_j\in\{0,1\}^{a_j}} \kappa_{j,\bx_j} = \abs{C} = n.
 \]
 Hence $\wt_{(V,C)}(\xi)\in\mathcal{X}$.
 This in turn implies that
 \begin{equation}\label{eq:CSP-to-COUNT}
  Z_\cF(V,C) = \sum_{x\in\mathcal{X}} x\cdot \#_\cF((V,C), x).
 \end{equation}
 {Recall that $\abs{\mathcal{X}}$ is polynomial in the instance size, and that the elements of $\mathcal{X}$ are generated by a straightforward procedure.
 Therefore, the elements of $\mathcal{X}$ can be enumerated in polynomial time.}
 Multiplication and addition within $\QQ(A)$ are also efficient, {hence} \eqref{eq:CSP-to-COUNT} gives the desired reduction $\NCSP(\cF) \leq_T \mathsf{COUNT}(\cF)$.
\end{proof}

We are now ready to prove the planar version of Theorem~\ref{thm:GHZ-state}.

\begin{theorem}\label{thm:GHZ-state-planar}
 Let $\mathcal{G}_1,\mathcal{G}_2\sse\allf$ be finite.
 Let $[y_0,y_1,y_2]\in\allf_2$ be an $\omega$-normalised and non-degenerate function.
 In the case $y_0=y_2=0$, further assume that $\mathcal{G}_1$ contains a unary function $[a,b]$ which is $\omega$-normalised and satisfies $ab\neq 0$.
 Then:
 \[
  \plholp{\{[y_0,y_1,y_2]\}\cup\mathcal{G}_1 \mid \{\EQ_3\}\cup\mathcal{G}_2} \equiv_T \plcsp(\{[y_0,y_1,y_2]\}\cup\mathcal{G}_1\cup\mathcal{G}_2).
 \]
\end{theorem}
\begin{proof}
 First, consider the reduction from the holant problem to the counting CSP.
 We have
 \begin{align*}
  \plholp{\{[y_0,y_1,y_2]\}\cup\mathcal{G}_1 \mid \{\EQ_3\}\cup\mathcal{G}_2}
  &\leq_T \plholp{\{[y_0,y_1,y_2]\}\cup\mathcal{G}_1 \cup \{\EQ_3\}\cup\mathcal{G}_2} \\
  &\leq_T \plcsp(\{[y_0,y_1,y_2]\}\cup\mathcal{G}_1\cup\mathcal{G}_2),
 \end{align*}
 where the first step is by forgetting the bipartition and the second step is by Definition~\ref{dfn:planar_CSP}.
 
 The reduction from the counting CSP to the holant problem is more complicated and separates into multiple cases.
 
 \textbf{Case~1}: Assume $\plholp{\{[y_0,y_1,y_2]\}\mid\{\EQ_3\}}$ is \numP-hard.
  
 Then the more general counting problem $\plholp{\{[y_0,y_1,y_2]\}\cup\mathcal{G}_1 \mid \{\EQ_3\}\cup\mathcal{G}_2}$ is also \numP-hard.
 The set $\{[y_0,y_1,y_2]\}\cup\mathcal{G}_1\cup\mathcal{G}_2$ is finite, hence $\mathsf{COUNT}(\{[y_0,y_1,y_2]\}\cup\mathcal{G}_1\cup\mathcal{G}_2)$ is in \numP.
 Therefore,
   \[
    \mathsf{COUNT}(\{[y_0,y_1,y_2]\}\cup\mathcal{G}_1\cup\mathcal{G}_2) \leq_T \plholp{\{[y_0,y_1,y_2]\}\cup\mathcal{G}_1 \mid \{\EQ_3\}\cup\mathcal{G}_2}.
   \]
 Furthermore, by Lemma~\ref{lem:NCSP_to_numP} with $\cF=\{[y_0,y_1,y_2]\}\cup\mathcal{G}_1\cup\mathcal{G}_2$,
   \[
    \plcsp(\{[y_0,y_1,y_2]\}\cup\mathcal{G}_1\cup\mathcal{G}_2) \leq_T \mathsf{COUNT}(\{[y_0,y_1,y_2]\}\cup\mathcal{G}_1\cup\mathcal{G}_2).
   \]
 Combining the two reductions yields the desired result.
   
 \textbf{Case~2}: Assume $\plholp{\{[y_0,y_1,y_2]\}\mid\{\EQ_3\}}$ is not \numP-hard.
 By Corollary~\ref{cor:pl-holant_binary}, this implies at least one of the following properties holds:
 \begin{enumerate}
  \item $[y_0,y_1,y_2]\in\ang{\cE}$, or
  \item $[y_0,y_1,y_2]\in\cA$, or
  \item $y_0,y_1,y_2\neq 0$ and $y_0^3=y_2^3$.
 \end{enumerate}
 We have dropped the holographic transformation from Subcase~2 because $[y_0,y_1,y_2]$ is required to be $\om$-normalised, which forces the holographic transformation to be trivial.
 
 For Properties~1 and 2, the desired reduction follows from Lemmas~2--4 of \cite{cai_holant_2012} since all the gadget constructions in those proofs are planar.
 
 For Property~3, note that the equation $y_0^3=y_2^3$ implies that $y_2 = e^{2ik\pi/3}y_0$ for some $k\in\{0,1,2\}$.
 Since $[y_0,y_1,y_2]$ is $\om$-normalised, we must have $k=0$ and thus $y_2=y_0$.
 We will now prove the following chain of reductions, where $g:=[y_0,y_1,y_0]$:
 \begin{align*}
  \plcsp\left(\{g\}\cup \mathcal{G}_1 \cup \mathcal{G}_2\right)
  &\leq_T \plhol\left(\{g, \EQ_3\}\cup \mathcal{G}_1 \cup \mathcal{G}_2\right) \\
  &\leq_T \plhol\left(\{g, \EQ_2\}\cup \mathcal{G}_1 \mid \{\EQ_2,\EQ_3\} \cup \mathcal{G}_2\right) \\
  &\leq_T \plhol\left(\{g,\EQ_2\}\cup \mathcal{G}_1 \mid \{\EQ_1,\EQ_2,\EQ_3\} \cup \mathcal{G}_2\right) \\
  &\leq_T \plhol\left(\{g\}\cup \mathcal{G}_1 \mid \{\EQ_1,\EQ_2,\EQ_3\} \cup \mathcal{G}_2\right) \\
  &\leq_T \plhol\left(\{g\}\cup \mathcal{G}_1 \mid \{\EQ_3\} \cup \mathcal{G}_2\right).
 \end{align*}
 The first reduction is the definition of $\plcsp$, the second step is by Proposition~\ref{prop:bipartite}, and the third step is because adding an additional function on the RHS cannot make the problem easier.
 The fourth reduction step is by Lemma~\ref{lem:interpolate_equality2}.
 
 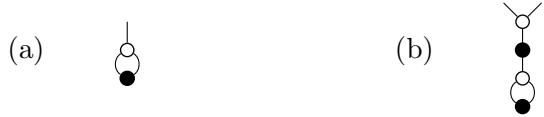
\begin{figure}
  \centering
  (a) \qquad \begin{tikzpicture}
	\begin{pgfonlayer}{nodelayer}
		\node [style=hollown] (0) at (0, -0) {};
		\node [style=solidn] (1) at (0, -0.75) {};
		\node [style=none] (2) at (0, 0.75) {};
	\end{pgfonlayer}
	\begin{pgfonlayer}{edgelayer}
		\draw (2.center) to (0);
		\draw [bend left=60, looseness=1.00] (0) to (1);
		\draw [bend right=60, looseness=1.00] (0) to (1);
	\end{pgfonlayer}
\end{tikzpicture} \qquad\qquad\qquad\qquad (b) \qquad \begin{tikzpicture}
	\begin{pgfonlayer}{nodelayer}
		\node [style=hollown] (0) at (0, -0.75) {};
		\node [style=solidn] (1) at (0, -1.5) {};
		\node [style=solidn] (2) at (0, -0) {};
		\node [style=hollown] (3) at (0, 0.75) {};
		\node [style=none] (4) at (-0.5, 1.25) {};
		\node [style=none] (5) at (0.5, 1.25) {};
	\end{pgfonlayer}
	\begin{pgfonlayer}{edgelayer}
		\draw [bend left=60, looseness=1.00] (0) to (1);
		\draw [bend right=60, looseness=1.00] (0) to (1);
		\draw (4.center) to (3);
		\draw (3) to (0);
		\draw (3) to (5.center);
	\end{pgfonlayer}
\end{tikzpicture}
  \caption{(a) A gadget for $y_0\cdot\EQ_1$ and (b) a gadget for $y_0(y_0+y_1)\cdot\EQ_2$, where each black degree-2 vertex is assigned $[y_0,y_1,y_0]$ and each white degree-3 vertex is assigned $\EQ_3$.}
  \label{fig:equality_gadgets}
 \end{figure}
 
 The first three reduction steps do not use any of the specific properties of $g$, and the fourth step only uses its property of being non-degenerate.
 It is only the fifth (and last) reduction step -- which we will now prove -- that uses the specific symmetry properties of $g$.
 
 Consider the gadgets in Figure~\ref{fig:equality_gadgets}, which can both be used on the RHS of the problem $\plhol\left(\{g\}\cup \mathcal{G}_1 \mid \{\EQ_3\} \cup \mathcal{G}_2\right)$.
 The first gadget has effective function $y_0\cdot\EQ_1$ and the second gadget has effective function $y_0(y_0+y_1)\cdot\EQ_2$.
 Recall that, $y_0\neq 0$ by the assumption of the subcase and $y_0^2=y_0y_2\neq y_1^2$ by non-degeneracy of $g$.
 The latter implies that $y_0+y_1\neq 0$.
 
 We thus have non-zero scalings of $\EQ_1$ and $\EQ_2$ on the RHS.
 Therefore, by Lemmas~\ref{lem:scaling} and~\ref{lem:realisable},
 \[
  \plhol\left(\{g\}\cup \mathcal{G}_1 \mid \{\EQ_1,\EQ_2,\EQ_3\} \cup \mathcal{G}_2\right)
  \leq_T \plhol\left(\{g\}\cup \mathcal{G}_1 \mid \{\EQ_3\} \cup \mathcal{G}_2\right).
 \]
 This establishes the desired result.
\end{proof}

\subsection{Proof of the \textsf{Holant}\texorpdfstring{\textsuperscript{+}}{\textasciicircum +} dichotomy theorem}
\label{s:hardness}

\newcommand{\stateholantplusdichotomy}{
 Let $\cF\sse\allf$ be finite.
 $\Holp[+]{\cF}$ can be computed in polynomial time if $\cF$ satisfies one of the following conditions:
 \begin{itemize}
  \item $\cF\subseteq\avg{\cT}$, or
  \item there exists $O\in\cO$ such that $\cF\subseteq\avg{O\circ\mathcal{E}}$, or
  \item $\cF\subseteq\avg{K\circ\mathcal{E}}=\avg{KX\circ\mathcal{E}}$, or
  \item $\cF\subseteq\avg{K\circ\mathcal{M}}$ or $\cF\subseteq\avg{KX\circ\mathcal{M}}$, or
  \item $\cF\subseteq\cA$.
 \end{itemize}
 In all other cases, the problem is \sP-hard.
 The same dichotomy holds for $\plhol^+(\cF)$.
}

\begin{theorem}\label{thm:holant_plus}
 \stateholantplusdichotomy
\end{theorem}

\begin{proof}
 Define $\cF':=\cF\cup\{\dl_0,\dl_1,\dl_+,\dl_-\}$.
 The tractability part of the theorem follows by reduction to a conservative holant problem or to $\csp$, respectively: if $\cF$ is a subset of one of the tractable sets of Theorem~\ref{thm:Holant-star}, then $\cF'$ is also a subset of one of the tractable sets of Theorem~\ref{thm:Holant-star}, and thus $\Holp[+]{\cF}$ can be solved in polynomial time.
 Similarly, if $\cF\subseteq\cA$, then $\cF'\subseteq\cA$.
 Furthermore,
 \[
  \Holp[+]{\cF} = \Holp{\cF'} \leq_T \holp{\cF'\cup\{\EQ_3\}} \leq_T \csp(\cF'),
 \]
 where the first reduction holds because adding a function cannot make the problem easier, and the second reduction is Proposition~\ref{prop:CSP_holant}.
 Now, by Theorem~\ref{thm:csp}, $\cF'\sse\cA$ implies that \new{$\csp(\cF')$ can be solved in polynomial time.
 Thus, by the above reduction, $\Holp[+]{\cF}$ can be solved in polynomial time.}
 
 Hence from now on we may assume that we are not in one of the known tractable cases.
 We will then prove the hardness of $\Holp[+]{\cF}$ via Theorem~\ref{thm:GHZ-state} (or Theorem~\ref{thm:GHZ-state-planar} in the planar case), which requires ternary and binary symmetric non-decomposable functions.
 
 Not being in one of the known tractable cases implies in particular that $\cF\not\sse\ang{\cT}$, i.e.\ there is some function $f\in\cF$ having at least one factor which is a non-decomposable function of arity $\geq 3$.
 Thus, we can apply Theorem~\ref{thm:three-qubit-gadget} to realise a non-decomposable ternary function $f'\in S(\{f,\dl_0,\dl_1,\dl_+,\dl_-\})$ via a planar gadget.
 This function has either $W$ or GHZ type, we distinguish cases accordingly.
 \begin{enumerate}
  \item\label{c:W} Suppose $f'$ has $W$ type.
   There are several subcases.
   \begin{itemize}
    \item If $f'\notin (K\circ\cM)\cup (KX\circ\cM)$, then there exists a non-decomposable symmetric ternary function $h\in S(\{f'\})$ by Lemma~\ref{lem:W_symmetrise}.
    \item If $f'\in K\circ\cM$, since $\cF\nsubseteq\ang{K\circ\cM}$, there exists $g\in\cF\setminus\ang{K\circ\cM}$.
    We can realise a non-decomposable binary function $g'\in S(\{g,\dl_0,\dl_1,\dl_+,\dl_-\})\setminus\ang{K\circ\cM}$ via a planar gadget by Lemma~\ref{lem:binary_notin_KcircM}.
    Then Lemma~\ref{lem:W_symmetrise-K} can be applied, yielding a non-decomposable symmetric ternary function $h\in S(\{f',g'\})$.
    \item If $f'\in KX\circ\cM$, the process is analogous to the subcase $f'\in K\circ\cM$.
   \end{itemize}
   In each subcase, by Lemma~\ref{lem:W_symmetrise} or Lemma~\ref{lem:W_symmetrise-K}, the gadget for $h$ is planar, the non-decomposable symmetric ternary function $h$ is in $S(\cF')$, and $h$ has GHZ type.
  \item Suppose $f'$ has GHZ type.
   Again, there are several subcases.
   \begin{itemize}
    \item If $f'$ is already symmetric, let $h:=f'$.
    \item If $f'$ is not symmetric, we can realise a non-decomposable symmetric ternary function $f''\in S(\{f',\dl_0,\dl_1,\dl_+,\dl_-\})$ by Lemma~\ref{lem:GHZ_symmetrise}.
    The gadget for $f''$ is planar.
     \begin{itemize}
      \item If $f''$ has GHZ type, let $h:=f''$.
      \item If $f''$ has $W$ type, go back to Case~\ref{c:W} with $f''$ in place of $f'$ and apply the symmetrisation procedure for $W$-type functions to get a {symmetric} GHZ-type function.
     \end{itemize}
   \end{itemize}
 \end{enumerate}
 
 To summarise, if $\cF$ is not one of the tractable sets, then there exists a non-decomposable symmetric ternary function $h\in S(\cF')$ which can be realised via a planar gadget and which has GHZ type.
 
 \new{Recall from Section~\ref{s:results_ternary_symmetric} that this means} there exists $M\in\GL$ such that $h=M\circ\EQ_3$ and either $M^T\circ\EQ_2$ is $\om$-normalised, or $M^T\circ\EQ_2=c\cdot\NEQ$ for some $c\in\AAnz$.
 In the latter case, since $\smm{1&0\\0&\ld}\circ\NEQ=\ld\cdot\NEQ$ for any $\ld$, \new{recall that} we can choose $M$ such that $M^T\circ\dl_0$ is $\om$-normalised.
 We may thus apply the following chain of interreductions:
 \begin{align*}
  \holp[+]{\cF} &= \holp{\cF'} \\
  &\equiv_T \holp{\cF'\cup\{h\}} \\
  &\equiv_T \holp{\cF'\cup\{\EQ_2\}\mid\{h,\EQ_2\}} \\
  &\equiv_T \holp{M^T\circ(\cF'\cup\{\EQ_2\})\mid\{\EQ_3,M^{-1}\circ\EQ_2\}} \\
  &\equiv_T \csp\left( M^T\circ(\cF'\cup\{\EQ_2\})\cup\{M^{-1}\circ\EQ_2\} \right)
 \end{align*}
 where the first step is the definition of $\hol^+$, the second step is by Lemma~\ref{lem:realisable}, the third step is by Proposition~\ref{prop:bipartite} with $\cG_1=\cF'$ and $\cG_2=\{h\}$, the fourth step is by Theorem~\ref{thm:Valiant_Holant}, and the last step is by Theorem~\ref{thm:GHZ-state}.
 To prove $\holp[+]{\cF}$ is hard, it therefore suffices to show that the counting CSP is hard whenever $\cF$ is not one of the tractable families of Theorem~\ref{thm:holant_plus}.
 
 We show the contrapositive: if the counting CSP is polynomial-time computable according to Theorem~\ref{thm:csp}, then $\cF$ is one of the tractable families of Theorem~\ref{thm:holant_plus}.
 The argument is split into cases according to the tractable cases of Theorem~\ref{thm:csp}.
 \begin{itemize}
  \item Suppose $M^T\circ(\cF\cup\{\EQ_2,\dl_0,\dl_1,\dl_+,\dl_-\})\cup\{M^{-1}\circ\EQ_2\}\sse\cA$.
  
   The condition $M^T\circ\{\EQ_2,\dl_0,\dl_1\}\sse\cA$ is equivalent to $M\in\cS$ by \eqref{eq:cS_definition}.
   The remaining conditions of the case are $M^T\circ\dl_+,M^T\circ\dl_-,M^{-1}\circ\EQ_2\in\cA$.
   Denote by $f_L$ the binary function corresponding to a matrix $L\in\GL$.
   By Lemma~\ref{lem:cS-cA}, $M\in\cS$ and $M^T\circ\dl_+,M^T\circ\dl_-\in\cA$ together imply that $M\in\cS_\cA:=\{L\in\GL\mid f_L\in\cA\}$, i.e.\ $M$ is a matrix corresponding to a binary function in $\cA$.
   Furthermore, by Lemma~\ref{lem:cS_cA-group}, $\cS_\cA$ is a group, so $M^{-1}\in\cS_\cA$.
   Now, by Lemma~\ref{lem:affine_closed}, $\cA$ is closed under taking gadgets, so by Lemma~\ref{lem:hc_gadget}, $M^{-1}\circ\EQ_2\in\cA$.
   Transposition of a matrix permutes the inputs of the corresponding function, so $M\in\cS_\cA$ also implies $M^T, (M^T)^{-1}\in\cS_\cA$.
   Thus, $M^T\circ\cF\sse\cA$ implies that $\cF\sse(M^T)^{-1}\circ\cA\sse\cA$, one of the known tractable cases.
  \item Suppose $M^T\circ(\cF\cup\{\EQ_2,\dl_0,\dl_1,\dl_+,\dl_-\})\cup\{M^{-1}\circ\EQ_2\}\sse\ang{\cE}$.
  
   Now, $M^T\circ\EQ_2$ and $M^{-1}\circ\EQ_2$ are non-decomposable symmetric binary functions with matrices $M^TM$ and $M^{-1}(M^{-1})^T$.
   All non-decomposable symmetric binary functions in $\ang{\cE}$ take the form $[\ld,0,\mu]$ or $[0,\ld,0]$ for some $\ld,\mu\in\AA\setminus\{0\}$.
   \begin{itemize}
    \item If $M^TM=\smm{\ld&0\\0&\mu}$, then $M=QD$ for some $Q\in\cO$ and some invertible diagonal matrix $D$ by Lemma~\ref{lem:ATA-D}.
     Now $(QD)^T\circ\cF\sse\ang{\cE}$ implies $\cF\sse\ang{(QD^{-1})\circ\cE}=\ang{Q\circ\cE}$, which is one of the known tractable families.
    \item Similarly, if $M^TM=\ld X$, then $M=KD$ or $M=KXD$ for some invertible diagonal matrix $D$ by Lemma~\ref{lem:ATA-X}.
     Now $K^T\doteq XK^{-1}$, so $(KD)^T\circ\cF\sse\ang{\cE}$ implies $\cF\sse\ang{(KXD^{-1})\circ\cE}=\ang{K\circ\cE}$, which is another of the known tractable families.
     
     An analogous argument holds for $KX$ instead of $K$.
   \end{itemize}
 \end{itemize}

 We have shown that if $\csp\left( M^T\circ(\cF'\cup\{\EQ_2\})\cup\{M^{-1}\circ\EQ_2\} \right)$ can be solved in polynomial time, this implies that $\cF$ is one of the tractable families listed in the Theorem~\ref{thm:holant_plus}.
 By Theorem~\ref{thm:csp}, the counting CSP is \sP-hard in all other cases.
 Thus, if $\cF$ is not one of the tractable families listed in the theorem, then $\holp[+]{\cF}$ is \sP-hard.
 This completes the proof of the theorem for the non-planar case.
 
 As all gadgets used in this proof are planar, the above constructions also work in the planar case.
 The only difference is that, when considering planar holant problems, we need to use Theorem~\ref{thm:GHZ-state-planar} instead of Theorem~\ref{thm:GHZ-state} and apply the planar $\csp$ dichotomy from Theorem~\ref{thm:planar_csp} instead of Theorem~\ref{thm:csp}.
 In addition to the tractable cases from the general $\csp$ dichotomy, which we have already excluded, Theorem~\ref{thm:planar_csp} contains one additional tractable family: $\plcsp\left( M^T\circ(\cF'\cup\{\EQ_2\})\cup\{M^{-1}\circ\EQ_2\} \right)$ can be solved in polynomial time if $M^T\circ(\cF'\cup\{\EQ_2\})\cup\{M^{-1}\circ\EQ_2\}\sse \smm{1&1\\1&-1}\circ\cH$, where $\cH$ is the set of matchgate functions.
 All other cases remain \sP-hard.
 By Lemma~\ref{lem:unary_matchgate}, the only unary matchgate functions are scalings of the pinning functions $\dl_0$ and $\dl_1$, so the only unary functions in $\smm{1&1\\1&-1}\circ\cH$ are $\smm{1&1\\1&-1}\circ\dl_0\doteq\dl_+$ and $\smm{1&1\\1&-1}\circ\dl_1\doteq\dl_-$ (up to scaling).
 Yet $M^T\circ\cF'$ contains at least the four unary functions $M^T\circ\{\dl_0,\dl_1,\dl_+,\dl_-\}$ and, by invertibility of $M$, these four functions have to be mapped to four pairwise linearly-independent functions.
 Thus, $M^T\circ(\cF'\cup\{\EQ_2\})\cup\{M^{-1}\circ\EQ_2\}$ cannot be a subset of $\smm{1&1\\1&-1}\circ\cH$.
 Therefore, the $\hol^+$ dichotomy remains unchanged when restricted to planar instances.
\end{proof}

\section{The full \textsf{Holant}\texorpdfstring{\textsuperscript{c}}{\textasciicircum c} dichotomy}
\label{s:dichotomy}

We now combine the techniques developed in deriving the $\hol^+$ dichotomy with techniques from the real-valued $\hol^c$ dichotomy \cite{cai_dichotomy_2017} to get a full complexity classification for complex-valued $\hol^c$.

As in the $\hol^+$ case, the general proof strategy is to realise a non-de\-com\-posable ternary function and then a non-decomposable symmetric ternary function.
Without $\dl_+$ and $\dl_-$, we can no longer use Theorem~\ref{thm:three-qubit-gadget}.
Instead, we employ a technique using $\dl_0$, $\dl_1$ and self-loops from the proof of Theorem~5.1 in \cite{cai_dichotomy_2017} with some slight modifications.
Self-loops reduce the arity of a function in steps of 2, so sometimes this technique fails to yield a non-decomposable ternary function.
When not yielding a ternary function, the process instead yields a non-decomposable arity-4 function with specific properties.
The complexity classification for the latter case was resolved in \cite{cai_dichotomy_2017} even for complex values.

The symmetrisation constructions for binary and ternary functions, as well as the subsequent hardness proofs, occasionally require a little extra work in the $\hol^c$ setting as compared to the $\hol^+$ setting; we deal with those issues before proving the main theorem.

\subsection{Hardness proofs involving a non-decomposable ternary function}
\label{s:ternary}

First, we prove several lemmas that give a complexity classification for $\hol^c$ problems in the presence of a non-decomposable ternary function.
These results adapt techniques used in the $\hol^+$ complexity classification to the $\hol^c$ setting.
They also replace Lemmas 5.1, 5.3, and 5.5--5.7 of \cite{cai_dichotomy_2017}.
Whereas the last three of those only apply to real-valued functions, our new results work for complex values.

\begin{lemma}\label{lem:case_KM}
 Suppose $\cF\sse\allf$ is finite and $f\in\cF$ is a non-decomposable ternary function.
 If $f\in K\circ\cM$ and {$\cF\nsubseteq\ang{K\circ\cM}$}, then $\holp[c]{\cF}$ is \sP-hard.
 Similarly, if $f\in KX\circ\cM$ and {$\cF\nsubseteq\ang{KX\circ\cM}$}, then $\holp[c]{\cF}$ is \sP-hard.
\end{lemma}
\begin{proof}
 We consider the case $f\in K\circ\cM$ and {$\cF\nsubseteq\ang{K\circ\cM}$}, the proof for the second case is analogous {since $X\circ\cM$ differs from $\cM$ only by a bit flip on all function inputs}.
 
 As {$\cF\nsubseteq\ang{K\circ\cM}$}, we can find $h\in\cF\setminus\ang{K\circ\cM}$.
 Then $h$ has arity at least 2 and is non-degenerate because all unary functions are in $K\circ\cM$.
 We distinguish cases according to the arity and decomposability properties of $h$.
 
 \textbf{Case~1}: Suppose $h$ has arity 2, then non-degeneracy implies $h$ is non-decomposable.
 Thus, by Lemma~\ref{lem:W_symmetrise-K}, there exists a non-decomposable symmetric ternary function $g\in S(\{f,h\})$.
 This function $g$ is guaranteed to have GHZ type by the same lemma.
 We thus have
 \begin{align*}
  \holp[c]{\cF} = \holp{\cF\cup\{\dl_0,\dl_1\}} 
  &\equiv_T \holp{\cF\cup\{\dl_0,\dl_1,g\}} \\
  &\equiv_T \holp{\cF\cup\{\dl_0,\dl_1,\EQ_2\} \mid \{g,\EQ_2\}}
 \end{align*}
 where the equality is the definition of $\hol^c$, the first reduction is by Lemma~\ref{lem:realisable}, and the second reduction is by {Proposition~\ref{prop:bipartite}}.
 Furthermore, there exists $M\in\GL$ such that $g=M\circ\EQ_3$ and either $M^T\circ\EQ_2$ is $\om$-normalised or $M^T\circ\EQ_2=c\cdot\NEQ$ for some $c\in\AAnz$.
 In the latter case, we can choose $M$ such that $M^T\circ\dl_0$ is $\om$-normalised.
 Therefore,
 \begin{align*}
  \holp[c]{\cF}
  &\equiv_T \holp{\cF\cup\{\dl_0,\dl_1,\EQ_2\} \mid \{g,\EQ_2\}} \\
  &\equiv_T \holp{M^T\circ(\cF\cup\{\dl_0,\dl_1,\EQ_2\})\mid\{\EQ_3,M^{-1}\circ\EQ_2\}} \\
  &\equiv_T \csp\left( M^T\circ(\cF\cup\{\dl_0,\dl_1,\EQ_2\})\cup\{M^{-1}\circ\EQ_2\} \right)
 \end{align*}
 where the first step is from above, the second step is Theorem~\ref{thm:Valiant_Holant}, and the third step is by Theorem~\ref{thm:GHZ-state}.
 Now, $f\in\cF\cap K\circ\cM$ is a non-decomposable ternary function.
 By Lemma~\ref{lem:family-types}, any non-decomposable ternary function in $\cM$ has $W$ type, and holographic transformations by definition do not affect the entanglement class.
 Thus, $M^T\circ f$ has $W$ type.
 But Lemma~\ref{lem:family-types} also shows that any non-decomposable ternary function in $\ang{\cE}$ or in $\cA$ has GHZ type.
 Therefore $M^T\circ f\notin\ang{\cE}$ and $M^T\circ f\notin\cA$, hence the counting CSP is \sP-hard.
 
 \textbf{Case~2}: Suppose $h$ is an $n$-ary function with $n>2$, and $h$ is non-decomposable.
 
 Write the ternary function $f$ as $K\circ f'$, where $f'\in\cM$ means that it takes the form $f'=(a,b,c,0,d,0,0,0)$ for some $a,b,c,d\in\AA$.
 Non-decomposability of $f$ implies $bcd\neq 0$.
 Consider the three different gadgets that consist of a vertex assigned function $f$ with a self-loop (where the three gadgets differ in which argument of $f$ corresponds to the dangling edge).
 The gadget where the first edge is dangling has effective function
 \begin{align*}
  \sum_{y\in\{0,1\}} f(x,y,y) &= \sum_{y,z_1,z_2,z_3\in\{0,1\}} K_{x z_1}K_{y z_2} K_{y z_3} f'(z_1,z_2,z_3) \\
  &= 2\sum_{z_1,z_2,z_3\in\{0,1\}} K_{x z_1} \NEQ(z_2,z_3) f'(z_1,z_2,z_3) \\
  &= 2\sum_{z_1\in\{0,1\}} K_{x z_1} \left( f'(z_1,0,1) + f'(z_1,1,0) \right) \\
  &= 2(b+c) (K\circ\dl_0)(x)
 \end{align*}
 Using a self loop on a vertex assigned function $f$, we can therefore realise the unary function $2(b+c) (K\circ\dl_0) = 2(b+c)[1,i]$, {where $i$ is the imaginary unit}.
 The other two gadgets similarly yield $2(b+d)[1,i]$ and $2(c+d)[1,i]$.
 Since $bcd\neq 0$, at least one of those gadgets is non-zero.
 Thus we can realise $\dl_i=[1,i]$ up to irrelevant scaling.
 
 We can now prove the following chain of interreductions:
 \begin{align*}
  \holp[c]{\cF} &= \holp{\cF\cup\{\dl_0,\dl_1\}} \\
  &\equiv_T \holp{\cF\cup\{\dl_0,\dl_1,\dl_i,f,h\}} \\
  &\equiv_T \holp{\cF\cup\{f,h,\EQ_2\}\mid\{\dl_0,\dl_1,\dl_i,\EQ_2\}} \\
  &\equiv_T \holp{ K^{-1}\circ(\cF\cup\{f,h,\EQ_2\}) \mid K^T\circ\{\dl_0,\dl_1,\dl_i,\EQ_2\} }
 \end{align*}
 Here, the equality is the definition of $\hol^c$.
 The first reduction is by Lemma~\ref{lem:realisable} and the above gadget constructions.
 The second reduction is by {Proposition~\ref{prop:bipartite}} and the third reduction is by Theorem~\ref{thm:Valiant_Holant}.
 
 Recall that $K=\smm{1&1\\i&-i}$, so $K^T\circ\dl_0\doteq \dl_+$, $K^T\circ\dl_1\doteq \dl_-$, $K^T\circ\dl_i\doteq \dl_1$ and $K^T\circ\EQ_2\doteq\NEQ$.
 Let $h':=K^{-1}\circ h$ and recall from above that $f'=K^{-1}\circ f = (a,b,c,0,d,0,0,0)$.
 Thus, the effective function of the gadget in Figure~\ref{fig:unaries}a is $k(x,y)=\sum_{z\in\{0,1\}} f'(x,y,z)\dl_\pm(z)$.
 This is a LHS gadget and $k=(a\pm b, c, d, 0)$.
 Since $b\neq 0$, there is a choice of sign such that $k(0,0)\neq 0$.
 
 Then the gadget in Figure~\ref{fig:unaries}b has effective function
 \[
  k'(x,y) = \sum_{z_1,z_2\in\{0,1\}} k(x,z_1)\NEQ(z_1,z_2)k(y,z_2).
 \]
 It is a symmetric LHS gadget, and $k' = (2d(a\pm b),cd,cd,0) \doteq [\frac{2}{c}(a\pm b),1,0]$.
 Note that with the above choice of sign, $z:=k'(0,0)\neq 0$.
 
 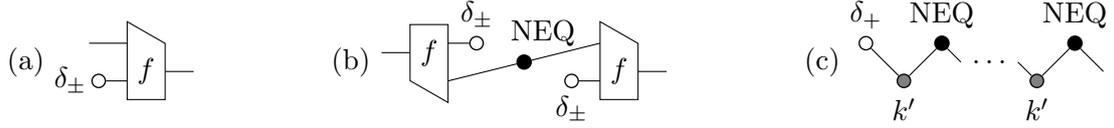
\begin{figure}
  \centering
  (a) \begin{tikzpicture}[rotate=90]
	\begin{pgfonlayer}{nodelayer}
		\node [style=map, shape border rotate=90, minimum height=0.5cm] (0) at (-0.25, -0.75) {$f$};
		\node [style=hollown] (1) at (-0.5, 0.5) {};
		\node [style=none] (2) at (0.5, 0.75) {};
		\node [style=none] (3) at (-0.25, -1.25) {};
		\node [style=none] (4) at (-0.25, -2) {};
		\node [style=none] (5) at (0.5, -0.25) {};
		\node [style=none] (6) at (-0.5, -0.25) {};
		\node [style=none] (7) at (-0.5, 1.25) {$\dl_\pm$};
	\end{pgfonlayer}
	\begin{pgfonlayer}{edgelayer}
		\draw (3.center) to (4.center);
		\draw (2.center) to (5.center);
		\draw (1) to (6.center);
	\end{pgfonlayer}
\end{tikzpicture} \qquad\qquad
  (b) \begin{tikzpicture}[rotate=90]
	\begin{pgfonlayer}{nodelayer}
		\node [style=map, shape border rotate=90, minimum height=0.5cm] (0) at (-0.25, -2.5) {$f$};
		\node [style=hollown] (1) at (-0.5, -1.25) {};
		\node [style=none] (2) at (-0.5, 2) {};
		\node [style=none] (3) at (-0.25, -3) {};
		\node [style=none] (4) at (-0.25, -3.75) {};
		\node [style=none] (5) at (0.5, -2) {};
		\node [style=none] (6) at (-0.5, -2) {};
		\node [style=none] (7) at (-1.25, -1.25) {$\dl_\pm$};
		\node [style=map, shape border rotate=270, minimum height=0.5cm] (8) at (0.25, 2.5) {$f$};
		\node [style=hollown] (9) at (0.5, 1.25) {};
		\node [style=none] (11) at (0.25, 3.75) {};
		\node [style=none] (12) at (0.25, 3) {};
		\node [style=none] (14) at (0.5, 2) {};
		\node [style=none] (15) at (1.25, 1.25) {$\dl_\pm$};
		\node [style=solidn] (16) at (0, 0) {};
		\node [style=none] (17) at (0.75, -0.5) {$\NEQ$};
	\end{pgfonlayer}
	\begin{pgfonlayer}{edgelayer}
		\draw (3.center) to (4.center);
		\draw (1) to (6.center);
		\draw (11.center) to (12.center);
		\draw (9) to (14.center);
		\draw (2.center) to (16);
		\draw (16) to (5.center);
	\end{pgfonlayer}
\end{tikzpicture} \qquad\qquad
  (c) \begin{tikzpicture}
	\begin{pgfonlayer}{nodelayer}
		\node [style=hollown] (0) at (-3.25, 0.5) {};
		\node [style=greyn] (1) at (-2.25, -0.5) {};
		\node [style=greyn] (2) at (1.25, -0.5) {};
		\node [style=solidn] (3) at (-1.25, 0.5) {};
		\node [style=solidn] (4) at (2.25, 0.5) {};
		\node [style=none] (5) at (3, -0.25) {};
		\node [style=none] (6) at (-0.75, 0) {};
		\node [style=none] (7) at (0.75, 0) {};
		\node [style=none] (8) at (0, 0) {$\ldots$};
		\node [style=none] (9) at (-3.25, 1.25) {$\dl_+$};
		\node [style=none] (10) at (-1.25, 1.25) {$\NEQ$};
		\node [style=none] (11) at (2.25, 1.25) {$\NEQ$};
		\node [style=none] (12) at (-2.25, -1.25) {$k'$};
		\node [style=none] (13) at (1.25, -1.25) {$k'$};
	\end{pgfonlayer}
	\begin{pgfonlayer}{edgelayer}
		\draw (0) to (1);
		\draw (1) to (3);
		\draw (2) to (4);
		\draw (4) to (5.center);
		\draw (3) to (6.center);
		\draw (7.center) to (2);
	\end{pgfonlayer}
\end{tikzpicture}
  \caption{(a) The gadget for $k$, where we denote $f$ by a non-symmetric box to indicate that it is not generally a symmetric function. (b) The gadget for $k'$, which is symmetric. (c) A family of gadgets for producing unary functions.}
  \label{fig:unaries}
 \end{figure}

 A chain of $\ell$ of these symmetric gadgets, connected to $\dl_{\pm}$ at one end and connected by copies of $\NEQ$, as shown in Figure~\ref{fig:unaries}c, gives a RHS gadget with function $[1,\ell z\pm 1]$.
 Thus, since $z\neq 0$, we can realise polynomially many different unary functions on the RHS.
 
 Since $h'=K^{-1}\circ h\notin\cM$, there exists a bit string $\ba\in\{0,1\}^n$ of Hamming weight at least 2 such that $h(\ba)\neq 0$.
 Without loss of generality, assume $a_1=a_2=1$.
 Otherwise permute the argument of $h$, the resulting function is in $S(\{h\})$ so it can be added to the LHS of our holant problem without affecting the complexity.
 Let $v_m=\dl_{a_m}$ for $m\in\{3,4\zd n\}$ and define
 \[
  g_1(x_1,x_2) := \sum_{x_3\zd x_n\in\{0,1\}} h'(\vc{x}{n}) \prod_{m=3}^n v_m(x_m).
 \]
 Then
 \begin{equation}\label{eq:not_in_cM}
  g_1(1,1)\neq 0.
 \end{equation}
 Furthermore, by Proposition~\ref{prop:popescu-rohrlich_gadget}, we know that there exist $u_m\in\{\dl_0,\dl_1,\dl_+,\dl_-\}$ for all $m\in\new{\{3,4\zd n\}}$ such that the following function is non-decomposable:
 \[
  g_2(x_1,x_2) := \sum_{x_3\zd x_n\in\{0,1\}} h'(\vc{x}{n}) \prod_{m=3}^n u_m(x_m).
 \]
 The non-decomposability condition for binary functions is
 \begin{equation}\label{eq:entangled}
  g_2(0,0)g_2(1,1)-g_2(0,1)g_2(1,0) \neq 0.
 \end{equation}
 \new{Now consider a third function
 \[
  g_3(x_1,x_2) := \sum_{x_3\zd x_n\in\{0,1\}} h'(\vc{x}{n}) \prod_{m=3}^n w_m(x_m)
 \]
 where $w_m = [1,1+\ell_m z]$ for each $m\in\{3,4\zd n\}$, with the $\ell_m$ being integer variables whose values are yet to be determined.
 Define
 \[
  p(\ell_3,\zd\ell_m) := g_3(1,1) \big( g_3(0,0)g_3(1,1)-g_3(0,1)g_3(1,0) \big).
 \]
 Then $p$ is a multivariate polynomial where the maximum exponent of any variable is 3.
 By \eqref{eq:not_in_cM} and \eqref{eq:entangled}, this polynomial is not identically zero.
 
 Now since the variable $\ell_3$ has degree at most 3 in $p$, there exists a value $\ld_3\in\{0,1,2,3\}$ such that $p(\ld_3,\ell_4\zd\ell_m)$ is not identically zero.
 We may repeat this argument for $\ell_4\zd\ell_m$ until we have found values $\ld_4\zd\ld_m$ for all the variables such that $p(\ld_3\zd\ld_m)\neq 0$.
 Each resulting $w_m$ is realisable by a RHS gadget.
 Thus, $g_3$ is realisable by a LHS gadget.
 The function $g_3$} is binary, non-decomposable, and not in $K\circ\cM$.
 Therefore we can proceed as in the case where $h$ is binary.
 
 \textbf{Case~3}: Suppose $h$ is an $n$-ary function with $n>2$, and $h$ is decomposable.
 Since $h\notin\ang{K\circ\cM}$, in any decomposition of $h$ there must be one factor $h'$ which is not in $\ang{K\circ\cM}$.
 This factor has arity at least 2 since all unary functions are in $\ang{K\circ\cM}$, and its arity is strictly smaller than that of $h$.
 Furthermore, by Lemma~\ref{lem:decomposable}, $h'\in S(\{h,\dl_0,\dl_1\})$.
 We may thus apply the argument to $h'$ instead of $h$.
 
 If $h$ satisfies the conditions of Case~1, we are immediately done.
 Case~2 straightforwardly reduces to Case~1.
 Finally, Case~3 yields a function of smaller arity than the original one, to which the case distinction can then be applied.
 Because the arity decreases every time we hit Case~3, the argument terminates.
 Thus the proof is complete.
\end{proof}

\begin{lemma}\label{lem:arity3_hardness}
 Suppose $\cF\sse\allf$ is finite and contains a non-decomposable ternary function $f$.
 Then $\holp[c]{\cF}$ is \sP-hard unless:
 \begin{enumerate}
  \item\label{c:orthogonal} There exists $O\in\cO$ such that $\cF\subseteq\avg{O\circ\mathcal{E}}$, or
  \item\label{c:KcE} $\cF\subseteq\avg{K\circ\mathcal{E}}=\avg{KX\circ\mathcal{E}}$, or
  \item\label{c:KcM} $\cF\subseteq\avg{K\circ\mathcal{M}}$ or $\cF\subseteq\avg{KX\circ\mathcal{M}}$, or
  \item\label{c:McA} $\cF\subseteq M\circ\cA$ for some $M\in\cS$, as defined in \eqref{eq:cS_definition}.
 \end{enumerate}
 In all the exceptional cases, the problem $\holp[c]{\cF}$ can be solved in polynomial time.
\end{lemma}

\begin{rem}
 The case $\cF\sse\ang{\cT}$ does not appear here because $f$ is a non-decomposable ternary function and $f\in\cF$, which implies $\cF\nsubseteq\ang{\cT}$.
\end{rem}

\begin{proof}
 We distinguish two cases according to whether the ternary function $f$ is symmetric.
 
 \textbf{Case~1}: Suppose $f$ is symmetric.
 By the entanglement classification (cf.\ Section~\ref{s:entanglement}), $f$ has either GHZ type or $W$ type.
 We treat these two subcases separately.
 
 \textbf{Subcase~a}: Suppose $f$ is of GHZ type.
 Given an instance of $\holp[c]{\cF}$, by Proposition~\ref{prop:make_bipartite},
 \[
  \holp[c]{\cF} = \holp{\cF\cup\{\dl_0,\dl_1\}} \equiv_T \holp{\cF\cup\{\dl_0,\dl_1\}\mid\{\EQ_2\}}.
 \]
 Furthermore, we have
 \[
  \holp{\cF\cup\{\dl_0,\dl_1\}\mid\{\EQ_2\}} \equiv_T \holp{\cF\cup\{\dl_0,\dl_1\}\mid\{\EQ_2,\dl_0,\dl_1\}}.
 \]
 The $\leq_T$ direction is immediate; for the other direction, note that any occurrence of $\dl_0$ or $\dl_1$ on the RHS can be replaced by a gadget consisting of a LHS copy of the unary function connected to $\EQ_2$.
 
 By Theorem~\ref{thm:Valiant_Holant}, for any $M\in\GL$,
 \begin{multline}\label{eq:Holant-c_holographic}
  \holp{\cF\cup\{\dl_0,\dl_1\}\mid\{\EQ_2,\dl_0,\dl_1\}} \\
  \equiv_T \holp{M^{-1}\circ(\cF\cup\{\dl_0,\dl_1\}) \,\middle|\, M^T\circ\{\EQ_2,\dl_0,\dl_1\}}.
 \end{multline}
 In the following, we aim to identify a matrix $M$ such that Theorem~\ref{thm:GHZ-state} can be applied to {show that the RHS of \eqref{eq:Holant-c_holographic} is equivalent to a counting CSP.
 To apply the theorem, the following three properties must hold for $M$ and $\cF$:
 \begin{itemize}
  \item $\EQ_3\in M^{-1}\circ\cF$,
  \item $M^T\circ\EQ_2$ is $\om$-normalised, and
  \item if $M^T\circ\EQ_2=[0,\ld,0]$ for some $\ld\in\AAnz$, then there must exist an $\om$-normalised function in $M^T\circ\{\dl_0,\dl_1\}$ which is not a pinning function
 \end{itemize}}
 We now show that it is always possible to choose $M$ so these conditions are satisfied.
 
 As $f$ has GHZ type, there exists $A\in\GL$ such that $f=A\circ\EQ_3$.
 The matrix $A$ is not unique {(cf.\ Section~\ref{s:results_ternary_symmetric})}; for now we pick an arbitrary one among all matrices that satisfy $f=A\circ\EQ_3$.
 \begin{itemize}
  \item Suppose $A^T\circ\EQ_2\neq \ld\cdot\NEQ$ for any $\ld\in\AA$.
   If $A^T\circ\EQ_2$ is $\om$-normalised, let $M:=A$.
   Otherwise, by the argument in Section~\ref{s:ternary}, there exists $D_\om:=\smm{1&0\\0&\om}$ with $\om\in\{e^{2i\pi/3}, e^{4i\pi/3}\}$ such that $(D_\om A^T)\circ\EQ_2$ is $\om$-normalised.
   Let $M:=A D_\om$, then $f=M\circ\EQ_3$.
   Now, since $\EQ_3\in M^{-1}\circ\cF$ and $M^T\circ\EQ_2$ is an $\om$-normalised symmetric binary function, Theorem \ref{thm:GHZ-state} can be applied.
  \item Suppose $A^T\circ\EQ_2=\ld\cdot\NEQ$ for some $\ld\in\AA$, then $A^T A\doteq X$.
   Thus, by Lemma~\ref{lem:ATA-X}, $A=KD$ or $A=KXD$ for some invertible diagonal matrix $D$.
   In either of these cases, all {entries} of $A$ are non-zero, so $A^T\circ\dl_0=[a,b]$ for some $a,b,\in\AA\setminus\{0\}$.
   Thus there exists $M:=A D_\om$ for some $\om^3=1$ such that $f=M\circ\EQ_3$ and $M^T\circ\dl_0$ is $\om$-normalised, i.e.\ the conditions of Theorem \ref{thm:GHZ-state} are satisfied.
 \end{itemize}
 In either case, Theorem~\ref{thm:GHZ-state} yields
 \begin{equation}\label{eq:Holant-c_csp}
  \holp[c]{\cF} \equiv_T \csp\left( M^{-1}\circ\left( \cF \cup \left\{ \dl_0, \dl_1 \right\} \right) \cup M^T \circ\{ \EQ_2, \dl_0, \dl_1 \} \right),
 \end{equation}
 where $\EQ_3\in M^{-1}\circ\cF$.
 The matrix $M$ may still not be uniquely defined, but the remaining ambiguity does not affect any of the subsequent arguments.
   
 By \eqref{eq:Holant-c_csp} and Theorem~\ref{thm:csp}, $\holp[c]{\cF}$ is \sP-hard unless
 \[
  \cF':= M^{-1}\circ\left( \cF \cup \left\{ \dl_0, \dl_1 \right\} \right) \cup M^T \circ\{\EQ_2, \dl_0, \dl_1\}
 \]
 is a subset of either $\avg{\cE}$ or $\cA$.
 Again, we treat these two cases separately.
 \begin{itemize}
  \item  Suppose $\cF'\sse\ang{\cE}$.
 All non-decomposable binary functions in $\avg{\cE}$ are of the form
 \[
  \begin{pmatrix}\alpha&0\\0&\beta\end{pmatrix} \qquad\text{or}\qquad \begin{pmatrix}0&\alpha\\\beta&0\end{pmatrix}
 \]
 for some $\alpha,\beta\in\AA\setminus\{0\}$.
 Note that $M^T\circ\EQ_2$ corresponds to the matrix $M^T M$.
 As $M^T\circ\EQ_2\in\cF'$, there are two subcases.
 \begin{itemize}
  \item Suppose $M^T M = \smm{\alpha&0\\0&\beta}$.
   Then by Lemma~\ref{lem:ATA-D}, $M=QD$ for some orthogonal 2 by 2 matrix $Q$ and some invertible diagonal matrix $D$.
   Now, $\cF'\sse\ang{\cE}$ implies $M^{-1}\circ\cF\sse\ang{\cE}$, which is equivalent to $\cF\sse\ang{M\circ\cE}$ since holographic transformations and closure under tensor products commute.
   But $D\circ\cE=\cE$ for any invertible diagonal matrix $D$.
   Hence $\cF\sse\ang{Q\circ\cE}$, where $Q$ is orthogonal, i.e.\ $\cF$ satisfies Item~\ref{c:orthogonal} of this lemma.
  \item Suppose $M^T M = \smm{0&\alpha\\\beta&0}$
   Then by Lemma~\ref{lem:ATA-X}, $M=KD$ or $M=KXD$ for some invertible diagonal matrix $D$.
   As in the previous case, $\cF'\sse\ang{\cE}$ implies $\cF\sse\ang{M\circ\cE}$.
   Since $\cE$ is invariant under holographic transformations by diagonal matrices and under bit flips, this is the same as $\cF\sse\ang{K\circ\cE}$.
   Hence $\cF$ satisfies Item~\ref{c:KcE} of this lemma.
 \end{itemize}
 This completes the analysis {of the subcase $\cF'\sse\ang{\cE}$.}
 \item Suppose $\cF'\sse\cA$.
 This implies in particular
   \begin{equation}\label{eq:condition_cA}
    M^T\circ\left\{\EQ_2,\dl_0,\dl_1\right\}\sse\cA,
   \end{equation}
 which is the definition of $M\in\cS$, cf.\ \eqref{eq:cS_definition}.
 
 The assumption $\cF'\in\cA$ also implies $M^{-1}\circ\{\dl_0,\dl_1\}\sse\cA$.
 This does not yield any further restrictions on $M=\smm{a&b\\c&d}$.
 To see this, note we already know that $M^T\circ\{\dl_0,\dl_1\}\sse\cA$, i.e.\ $[a,b],[c,d]\in\cA$.
 Now, up to an irrelevant scalar factor, $M^{-1}\circ\{\dl_0,\dl_1\}$ is $\{[d,-c],[-b,a]\}$.
 Note that $[-b,a]=-\smm{0&1\\-1&0}\circ [a,b]$ and $[d,-c]=\smm{0&1\\-1&0}\circ [c,d]$.
 The function corresponding to $\smm{0&1\\-1&0}$ is affine and $\cA$ is closed under taking gadgets (cf.\ Lemma~\ref{lem:affine_closed}) and under scalings.
 Hence, by Lemma~\ref{lem:hc_gadget}, $[a,b],[c,d]\in\cA$ implies $[d,-c],[-b,a]\in\cA$.
 
 Thus $\cF'\sse\cA$ implies $M\in\cS$.
 Furthermore, $\cF'\sse\cA$ implies $M^{-1}\circ\cF\sse\cA$, or equivalently $\cF\sse M\circ\cA$, i.e.\ $\cF$ satisfies Item~\ref{c:McA} of this lemma.
 \end{itemize}
   
 To summarise, if $f$ has GHZ type, the problem is tractable if $\cF\subseteq\avg{Q\circ\cE}$ for some orthogonal 2 by 2 matrix $Q$, if $\cF\subseteq\avg{K\circ\cE}$, or if there exists $M\in\cS$ such that $\cF\subseteq M\circ\cA$.
 In all other cases, the problem is \sP-hard by reduction from $\csp$.

 \textbf{Subcase~b}: Suppose $f$ is of $W$ type, then:
   \begin{itemize}
    \item If $f\notin (K\circ\cM)\cup (KX\circ\cM)$, $\holp{f}$ is \sP-hard by Theorem \ref{thm:W-state}.
    \item If {$\cF\subseteq\ang{K\circ\cM}$ or $\cF\subseteq\ang{KX\circ\cM}$}, the problem is polynomial-time computable by Theorem~\ref{thm:Holant-star}; this is Item~\ref{c:KcM} of this lemma.
    \item If $f\in K\circ\cM$ but {$\cF\nsubseteq\ang{K\circ\cM}$}, the problem is \sP-hard by Lemma~\ref{lem:case_KM}, and analogously with $KX$ instead of $K$.
   \end{itemize}
 
 \textbf{Case~2}: Suppose $f$ is not symmetric.
 \begin{itemize}
  \item If $f\notin (K\circ\cM)\cup (KX\circ\cM)$, we can realise a non-decomposable symmetric ternary function by Lemmas \ref{lem:GHZ_symmetrise} and \ref{lem:W_symmetrise} and then proceed as in Case~1.
  \item If {$\cF\subseteq\ang{K\circ\cM}$ or $\cF\subseteq \ang{KX\circ\cM}$}, the problem is polynomial-time computable by Theorem~\ref{thm:Holant-star}; this is Item~\ref{c:KcM} of the lemma.
  \item Finally, if $f\in K\circ\cM$ but {$\cF\nsubseteq\ang{K\circ\cM}$}, or $f\in KX\circ\cM$ but {$\cF\nsubseteq\ang{KX\circ\cM}$}, the problem is \sP-hard by Lemma~\ref{lem:case_KM}.
 \end{itemize}
 This covers all cases.
\end{proof}

\subsection{Main theorem}
\label{s:main_theorem}

We now have all the components required to prove the main dichotomy for $\hol^c$.
The following theorem generalises Theorem 5.1 of \cite{cai_dichotomy_2017}, which applies only to real-valued functions.
The proof follows the one in \cite[Theorem~5.1]{cai_dichotomy_2017} fairly closely.

\begin{theorem}\label{thm:holant-c}
 Let $\cF\sse\allf$ be finite.
 Then $\holp[c]{\cF}$ is \sP-hard unless:
 \begin{itemize}
  \item $\cF\subseteq\avg{\cT}$, or
  \item there exists $O\in\cO$ such that $\cF\subseteq\avg{O\circ\mathcal{E}}$, or
  \item $\cF\subseteq\avg{K\circ\mathcal{E}}=\avg{KX\circ\mathcal{E}}$, or
  \item $\cF\subseteq\avg{K\circ\mathcal{M}}$ or $\cF\subseteq\avg{KX\circ\mathcal{M}}$, or
  \item there exists $B\in\cS$ such that $\cF\subseteq B\circ\cA$, or
  \item $\cF\subseteq\cL$.
 \end{itemize}
 In all of the exceptional cases, $\holp[c]{\cF}$ is polynomial-time computable.
\end{theorem}
\begin{proof}
 If $\cF$ is in one of the first four exceptional cases, then polynomial-time computability follows from Theorem~\ref{thm:Holant-star}.
 In the \new{penultimate exceptional case, let $\cF':= B^{-1}\circ\cF\sse\cA$, then
 \begin{align*}
  \holp[c]{\cF} &\leq_T \holp[c]{\cF\mid\{\EQ_2\}} \\
  &\leq_T \holp[c]{\cF'\mid\{B^T\circ\EQ_2\}} \\
  &\leq_T \holp[c]{\cF'\cup\{B^T\circ\EQ_2\}} \\
  &\leq_T \holp[c]{\cF'\cup\{B^T\circ\EQ_2,\EQ_3\}} \\
  &\leq_T \csp(\cF'\cup\{B^T\circ\EQ_2\})
 \end{align*}
 where the first reduction is Proposition~\ref{prop:make_bipartite} and the second reduction is Theorem~\ref{thm:Valiant_Holant} with $M=B^{-1}$.
 The third and fourth reductions hold because dropping the bipartite restriction and adding a function cannot make the problem easier.
 The final reduction combines Proposition~\ref{prop:CSP_holant} and Lemma~\ref{lem:csp^c}.
 Now, $B^T\circ\EQ_2\in\cA$ by \eqref{eq:cS_definition}.
 A counting CSP with affine constraint functions is polynomial-time computable by Theorem~\ref{thm:csp}.
 
 In the final exceptional case, we have
 \[
  \holp[c]{\cF} \leq_T \holp[c]{\cF\cup\{\EQ_4\}} \leq_T \csp_2^c(\cF)
 \]
 where the final reduction can be proved analogously to Proposition~\ref{prop:CSP_holant}.
 Polynomial-time computability thus follows from Theorem~\ref{thm:csp_2^c}.}
 
 From now on, assume $\cF$ does not satisfy any of the tractability conditions.
 In particular, this implies that $\cF\not\subseteq\avg{\cT}$, i.e.\ $\cF$ has multipartite entanglement.
 
 By Lemma~\ref{lem:decomposable}, without loss of generality, we may focus on non-decomposable functions.
 So assume that there is some non-decomposable function $f\in S(\cF\cup\{\dl_0,\dl_1\})$ of arity $n\geq 3$.
 If the function has arity 3, we are done by Lemma~\ref{lem:arity3_hardness}.
 Hence assume $n\geq 4$.
 
 Now if there was a bit string $\ba\in\{0,1\}^n$ such that $f(\bx)=0$ for all $\bx\in\{0,1\}^n\setminus\{\ba\}$, then $f$ would be a scaling of some function in $\ang{\{\dl_0,\dl_1\}}$, so it would be degenerate and thus decomposable.
 Hence, $f$ being non-decomposable implies that there exist two distinct $n$-bit strings $\ba,\bb\in\{0,1\}^n$ such that $f(\ba)f(\bb)\neq 0$.
 
 As in the proof of \cite[Theorem~5.1]{cai_dichotomy_2017}, let:
 \begin{equation}\label{eq:D0}
  D_0 = \min \left\{ d(\bx,\by) \mid \bx\neq\by, f(\bx)\neq 0, f(\by)\neq 0 \right\},
 \end{equation}
 where $d(\cdot,\cdot)$ is the Hamming distance.
 By the above argument, $D_0$ exists.
 We now distinguish cases according to the values of $D_0$.
 
 \textbf{Case $D_0\geq 4$ and $D_0$ is even}: Pick a pair of bit strings $\ba,\bb$ such that $f(\ba)f(\bb)\neq 0$ with {minimum} Hamming distance $d(\ba,\bb)$. 
 Pin all inputs where the two bit strings agree (without loss of generality, we always assume bit strings agree on the last $n-D_0$ bits; otherwise permute the arguments).
 This realises a function $f'\in\allf_{D_0}$ of the form
 \begin{equation}\label{eq:generalised_eq}
  f'(\bx) = \begin{cases} \alpha & \text{if } x_j=a_j \text{ for {all} } j\in [D_0] \\ \beta & \text{if } x_j=\bar{a}_j \text{ for {all} } j\in [D_0] \\ 0 & \text{otherwise}, \end{cases}
 \end{equation}
 where $\alpha\beta\neq 0$.
 {If $D_0=4$, this is a generalised equality of arity 4.
 If $D_0>4$, there must exist $j,k\in[D_0]$ with $j\neq k$ but $a_j=a_k$.
 Via a self-loop, we can realise the function $\sum_{x_j,x_k\in\{0,1\}}\EQ_2(x_j,x_k)f(\bx)$, which has arity $(D_0-2)$ and still satisfies a property like \eqref{eq:generalised_eq} for the remaining arguments.
 This process can be repeated, reducing the arity in steps of 2, to finally realise an arity-4 generalised equality function.}
 Then \sP-hardness follows by Lemma~\ref{lem:generalised_equality4}.
 
 \textbf{Case $D_0\geq 3$ and $D_0$ is odd}: Pin and use self-loops analogously to the previous case to realise a function of the form \eqref{eq:generalised_eq} with arity 3.
 Then \sP-hardness follows by Lemma~\ref{lem:arity3_hardness}.
 
 \textbf{Case $D_0=2$}: We can realise by pinning either a function $g=[\alpha,0,\beta]$ or a function $\tilde{g} = \smm{0&\alpha\\\beta&0}$ (up to permutation of the arguments), i.e.\ a generalised equality or a generalised disequality.
 In the second subcase, suppose the indices of $f$ that were not pinned in order to realise $\tilde{g}$ are $j$ and $k$.
 Replace $f$ by the gadget $\sum_{z\in\{0,1\}}\tilde{g}(x_j,z)f(\vc{x}{j-1},z,x_{j+1}\zd x_n)$.
 The new function has the same entanglement and the same $D_0$ as the old one, but pinning all inputs except $j$ and $k$ now gives rise to a generalised equality function as in the first subcase.
 \new{If necessary, redefine $\alpha,\beta$ according to this new function.
 In the following, we assume $j=1$, $k=2$; this is without loss of generality as permuting arguments is a gadget operation.}
 
 Following the proof of \cite[Theorem~5.1]{cai_dichotomy_2017}, define $f_\bx (z_1,z_2):=f(z_1,z_2,\bx)$ for any $\bx\in\{0,1\}^{n-2}$ and let
 \begin{align*}
  A_1 &:= \{ \bx\in\{0,1\}^{n-2} \mid f_\bx(z_1,z_2) = c[\alpha,0,\beta] \text{ for some } c\in\AAnz \}, \\
  B_1 &:= \{ \bx\in\{0,1\}^{n-2} \mid f_\bx(z_1,z_2) \neq c[\alpha,0,\beta] \text{ for any } c\in\AA \}.
 \end{align*}
 \new{Note that $A_1$ is non-empty since we can realise $g$ by pinning.
 The set $B_1$ must then be non-empty since otherwise $f$ would be decomposable as $f(z_1,z_2,\bx) = g(z_1,z_2)f'(\bx)$ for some $f'\in\allf_{n-2}$.
 Furthermore,} note that $A_1\cap B_1=\emptyset$ \new{and that} if $f_\by$ is identically 0, then $\by$ is not in $A_1\cup B_1$.
 Thus we can define:
 \[
  D_1 = \min\{ d(\bx,\by) \mid \bx \in A_1, \by\in B_1 \}.
 \]
 Pick a pair $\ba\in A_1,\bb\in B_1$ with {minimum} Hamming distance and pin wherever they are equal, as in the cases where $D_0\geq 3$.
 This realises a function
 \[
  h(\vc{x}{D_1+2}) = g(x_1,x_2)\prod_{k=1}^{D_1}\dl_{a_k}(x_{k+2}) + g'(x_1,x_2)\prod_{k=1}^{D_1}\dl_{\bar{a}_k}(x_{k+2}),
 \]
 where $g'(x_1,x_2):=f_\bb(x_1,x_2)$.
 Note that the assumption $D_0=2$ implies that either $f_\by(0,1)=f_\by(1,0)=0$ or $f_\by(0,0)=f_\by(1,1)=0$ for all $\by\in B_1$.
 {(Assume otherwise, i.e.\ each of the sets $\{f_\by(0,1),f_\by(1,0)\}$ and $\{f_\by(0,0),f_\by(1,1)\}$ contains at least one non-zero value. Then it is straightforward to see that there are two inputs for which $f$ is non-zero and whose Hamming distance is 1, so by \eqref{eq:D0} we should have $D_0=1$, a contradiction.)}
 Thus, $g'$ is either $\smm{\ld&0\\0&\mu}$ or $\smm{0&\ld\\\mu&0}$ for some $\ld,\mu\in\AA$ such that $\ld,\mu$ are not both zero.
 If $g'=\smm{\ld&0\\0&\mu}$, then the definition of $B_1$ furthermore implies $\alpha\mu-\beta\ld\neq 0$.
 
 We now distinguish cases according to the values of $D_1$.
 \begin{itemize}
  \item If $D_1\geq 3$, \new{then $\ari(h)=D_1+2\geq 5$}.
   Distinguish cases according to $g'$.
   \begin{itemize}
    \item Suppose $g'=\smm{\ld&0\\0&\mu}$.
     If $\ld\neq 0$, we can pin the first two inputs of $h$ to 00 to get a function ${\alpha}(\prod_{k=1}^{D_1}\dl_{a_k}(x_{k+2}))+{\ld}(\prod_{k=1}^{D_1}\dl_{\bar{a}_k}(x_{k+2}))$.
     \new{The resulting function still has arity at least 3, so we can} proceed as in the cases $D_0\geq 4 $ or $D_0\geq 3$.
     If $\ld=0$ then $\mu\neq 0$ and we can pin to 11 instead {to get the function $\beta(\prod_{k=1}^{D_1}\dl_{a_k}(x_{k+2}))+\mu\prod_{k=1}^{D_1}\dl_{\bar{a}_k}(x_{k+2}))$, and proceed analogously}.
    \item Suppose $g'=\smm{0&\ld\\\mu&0}$.
     If $\ld\neq 0$, we can pin the first input of $h$ to 0 to realise:
     \[
      {\alpha}\delta_0(x_2)\prod_{k=1}^{D_1}\dl_{a_k}(x_{k+2}) + \ld\delta_1(x_2)\prod_{k=1}^{D_1}\dl_{\bar{a}_k}(x_{k+2})
     \]
     at which point we again proceed as in the cases $D_0\geq 4 $ or $D_0\geq 3$.
     If $\ld=0$ then $\mu\neq 0$ and we can pin the first input to 1 instead {to get the function
     \[
      \beta\delta_1(x_2)\prod_{k=1}^{D_1}\dl_{a_k}(x_{k+2}) + \mu\delta_0(x_2)\prod_{k=1}^{D_1}\dl_{\bar{a}_k}(x_{k+2}),
     \]
     and proceed analogously.}
   \end{itemize}
  \item If $D_1=2$, there are eight subcases, depending on the form of $g'$ and the value of $\ba$.
  They can be considered pairwise, grouping $\ba$ with $\bar{\ba}$.
   \begin{itemize}
    \item If $g'=\smm{\ld&0\\0&\mu}$ {and $\ba=00$}, the function after pinning is
     \[
      \begin{pmatrix} \alpha&0&0&\ld \\ 0&0&0&0 \\ 0&0&0&0 \\ \beta&0&0&\mu \end{pmatrix}
     \]
     with $\alpha\beta\neq 0$ and $\alpha\mu-\beta\ld\neq 0$.
     In this case, apply Lemma~\ref{lem:interpolate_equality4} {to show
     \[
      \hol^c(\cF\cup\{\EQ_4\})\leq_T\hol^c(\cF).
     \]
     Now by Lemma~\ref{lem:generalised_equality4}, $\csp_2^c(\cF\cup\{\EQ_4\})\leq_T\hol^c(\cF\cup\{\EQ_4\})$.
     Therefore hardness follows by Theorem~\ref{thm:csp_2^c}.}
     Here, the original proof in \cite{cai_dichotomy_2017} used a different technique requiring real values.
     {If $\ba=11$, the argument is analogous with the first and last columns swapped.}
    \item {If $g'=\smm{\ld&0\\0&\mu}$ and $\ba=01$, the function after pinning is
     \[
      \begin{pmatrix} 0&\alpha&\ld&0 \\ 0&0&0&0 \\ 0&0&0&0 \\ 0&\beta&\mu&0 \end{pmatrix}
     \]
     where again $\alpha\beta\neq 0$ and $\alpha\mu-\beta\ld\neq 0$.
     Call this function $h'$.
     Then the gadget $\sum_{y,z\in\{0,1\}} h'(x_1,x_2,y,z)h'(x_3,x_4,y,z)$ corresponds to taking the product of the above matrix with its transpose; it takes values
     \[
      \begin{pmatrix} \alpha^2+\lambda^2&0&0&\alpha\beta+\ld\mu \\ 0&0&0&0 \\ 0&0&0&0 \\ \alpha\beta+\ld\mu&0&0&\beta^2+\mu^2 \end{pmatrix}
     \]
     Now
     \[
      \det \pmm{\alpha^2+\lambda^2&\alpha\beta+\ld\mu \\ \alpha\beta+\ld\mu&\beta^2+\mu^2} = (\alpha\mu-\beta\ld)^2 \neq 0,
     \]
     which means Lemma~\ref{lem:interpolate_equality4} can be applied again.
     Thus hardness follows as in the first subcase.
     
     If $\ba=10$, the middle two columns of the first matrix are swapped, but the argument is analogous.}
    \item If $g'=\smm{0&\ld\\\mu&0}$ {and $\ba=00$}, the function after pinning is
     \[
      \begin{pmatrix} \alpha&0&0&0 \\ 0&0&0&\ld \\ 0&0&0&\mu \\ \beta&0&0&0 \end{pmatrix}.
     \]
     If $\ld\neq 0$, pin the first input of $h$ to 0 to get the function $[\alpha,0,0,\ld]$.
     If $\ld=0$ then $\mu\neq 0$, so pin the first input of $h$ to 1 to get the function {$\smm{0&0&0&\mu\\\beta&0&0&0}$.
     
     If $\ba=11$, then the first and last columns are swapped, the argument is otherwise analogous.
     In each case, the resulting function has arity~3 and is non-decomposable (in fact, it is a generalised equality), so we can show hardness via Lemma~\ref{lem:arity3_hardness}.}
    \item {If $g'=\smm{0&\ld\\\mu&0}$ and $\ba=01$, the function after pinning is
     \[
      \begin{pmatrix} 0&\alpha&0&0 \\ 0&0&\ld&0 \\ 0&0&\mu&0 \\ 0&\beta&0&0 \end{pmatrix}
     \]
     and we can pin as in the previous subcase.     
     If $\ba=10$, the middle columns are swapped and the process is analogous.
     All resulting functions have arity~3 and are non-decomposable (in fact, they are generalised equalities), so we can show hardness via Lemma~\ref{lem:arity3_hardness}.}
   \end{itemize}
  \item If $D_1=1$, then $h(x_1,x_2,x_3)=g(x_1,x_2)\dl_{a_1}(x_3) + g'(x_1,x_2)\dl_{\bar{a}_1}(x_3)$ is a non-decomposable ternary function since $g$ and $g'$ are linearly independent.
   Thus, we can apply Lemma~\ref{lem:arity3_hardness} to show hardness.
 \end{itemize}
   
 \textbf{Case $D_0=1$}: By pinning, we can realise $g=[\alpha,\beta]$ for some $\alpha,\beta\neq 0$.
 \new{Without loss of generality, assume this function arises by pinning the last $(n-1)$ inputs of $f$; otherwise permute the arguments.}
 Define $f_\bx (z):=f(z,\bx)$ for any $\bx\in\{0,1\}^{n-1}$ and let
 \begin{align*}
  A_2 &:= \{ \bx\in\{0,1\}^{n-1} \mid f_\bx(z) = c[\alpha,\beta] \text{ for some } c\in\AA\setminus\{0\} \}, \\
  B_2 &:= \{ \bx\in\{0,1\}^{n-1} \mid f_\bx(z) \neq c[\alpha,\beta] \text{ for any } c\in\AA \}.
 \end{align*}
 \new{Analogously to $A_1$ and $B_1$, these two sets are non-empty, do not intersect, and do not contain any $\bx$ such that $f_\bx(z)$ is identically zero.}
 Then let:
 \[
  D_2 = \min\{ d(\bx,\by) \mid \bx\in A_2, \by\in B_2 \},
 \]
 \new{which is well-defined.}
 By pinning we can realise a function
 \[
  h(\vc{x}{D_2+1}) = g(x_1)\prod_{k=1}^{D_2}\dl_{a_k}(x_{k+1}) + g'(x_1)\prod_{k=1}^{D_2}\dl_{\bar{a}_k}(x_{k+1}),
 \]
 where $g'=[\ld,\mu]$ with $\alpha\mu-\beta\ld\neq 0$.
 \begin{itemize}
  \item If $D_2\geq 3$, then $h$ has arity greater than 3.
    If $\ld\neq 0$, pin the first input of $h$ to 0, which yields a generalised equality function of arity at least 3, {and then proceed as in the cases $D_0\geq 4$ or $D_0\geq 3$}.
    If $\ld=0$ then $\mu\neq 0$, so we can pin to 1 instead for an analogous result.
  \item If $D_2=2$, $h$ is a non-decomposable ternary function so we are done by Lemma~\ref{lem:arity3_hardness}.
   This is another change compared to the proof in \cite{cai_dichotomy_2017}, where hardness was only shown for a real-valued function of the given form.
  \item If $D_2=1$, $h$ is a non-decomposable binary function $\smm{\alpha&\beta\\\ld&\mu}$ (possibly up to a bit flip of the second argument, which does not affect the proof).
  Unlike in \cite{cai_dichotomy_2017}, we do not attempt to use this binary function for interpolation.
  Instead we immediately proceed to defining $A_3$, $B_3$, and $D_3$ analogous to before.
  \new{Without loss of generality, assume that the two variables of $h$ correspond to the first two variables of $f$ (otherwise permute arguments).
  Let}
  \begin{align*}
   A_3 &:= \{ \bx\in\{0,1\}^{n-2} \mid f_\bx(z_1,z_2) = c\cdot h(z_1,z_2) \text{ for some } c\in\AA\setminus\{0\} \}, \\
   B_3 &:= \{ \bx\in\{0,1\}^{n-2} \mid f_\bx(z_1,z_2) \neq c\cdot h(z_1,z_2) \text{ for any } c\in\AA \},
  \end{align*}
  \new{then $A_3$ and $B_3$ are non-empty, they do not intersect, and they do not contain any $f_\bx$ which is identically zero.
  Thus we may define}
  $D_3 = \min\{ d(\bx,\by) \mid \bx\in A_3, \by\in B_3 \}$.
  Let $\ba\in A_3$ be such that there exists $\bb\in B_3$ with $d(\ba,\bb)=D_3$.
  By pinning in all places that $\ba$ and $\bb$ agree, we realise a function
  \begin{equation}\label{eq:function3}
   h(x_1,x_2)\prod_{k=1}^{D_3}\dl_{a_k}(x_{k+2}) + h'(x_1,x_2)\prod_{k=1}^{D_3}\dl_{\bar{a}_k}(x_{k+2})
  \end{equation}
  where $h'(x_1,x_2):=f_\bb(x_1,x_2)$ is not a scaling of $h$.
  Suppose $h'=\smm{\alpha'&\beta'\\\ld'&\mu'}$.
  Distinguish cases according to $D_3$.
  {It will be useful to consider these in ascending order since the case $D_3=2$ follows straightforwardly from $D_3=1$.}
   \begin{itemize}
    \item {If $D_3=1$, the function in \eqref{eq:function3} is ternary.
     Note that the case $a_1=1$ differs from the case $a_1=0$ only by a bit flip on the third input, which is a SLOCC transformation by $(I\otimes I\otimes X)$ and thus does not affect entanglement.
     So it suffices to consider $a_1=0$; the function then takes values
     \[
      \pmm{\alpha&\alpha'&\beta&\beta'\\\ld&\ld'&\mu&\mu'}.
     \]
     Recall that $\alpha,\beta\neq 0$, the unprimed variables satisfy $\alpha\mu-\beta\ld\neq 0$, and $(\alpha',\beta',\ld',\mu')$ is not a scaling of $(\alpha,\beta,\ld,\mu)$.
     
     By Lemma~\ref{lem:li}, the ternary function is non-decomposable unless at least two of the following three expressions are false:
     \begin{align*}
      (\alpha\beta'\neq\beta\alpha') &\vee (\mu\ld'\neq\ld\mu') \\
      (\ld\alpha'\neq\alpha\ld') &\vee (\mu\beta'\neq\beta\mu') \\
      (\beta'\ld'\neq\alpha'\mu') &\vee (\beta\ld\neq\alpha\mu)
     \end{align*}
     Now, the third expression is always true since $\alpha\mu-\beta\ld\neq 0$ implies $\beta\ld\neq\alpha\mu$.
     So the function being decomposable would imply $\alpha\beta'=\beta\alpha'$, $\mu\ld'=\ld\mu'$, $\ld\alpha'=\alpha\ld'$ and $\mu\beta'=\beta\mu'$.
     But the first of these equations implies $\beta'=\beta\alpha'/\alpha$ since $\alpha\neq 0$.
     Similarly, the third equation implies $\ld'=\ld\alpha'/\alpha$.
     Also, since $\beta\neq 0$, the fourth equation implies $\mu'=\mu\beta'/\beta = \mu\alpha'/\alpha$.
     Thus,
     \[
      h' = \pmm{\alpha'&\beta'\\\ld'&\mu'} = \frac{\alpha'}{\alpha}\pmm{\alpha&\beta\\\ld&\mu} = \frac{\alpha'}{\alpha}\cdot h,
     \]
     which is a contradiction since $h'$ is not a scaling of $h$.
     Thus, the ternary function must be non-decomposable, and hardness follows by Lemma~\ref{lem:arity3_hardness}.
     
     This is a change compared to the proof in \cite{cai_dichotomy_2017}, where multiple cases were distinguished and the hardness lemmas only applied to real-valued functions.}
   
    \item {If $D_3=2$, the function in \eqref{eq:function3} has arity 4.
     By connecting $[\alpha,\beta]$ to the last input, we get one of the functions
     \begin{align*}
      \alpha\cdot h(x_1,x_2)\dl_{a_1}(x_3) + \beta\cdot h'(x_1,x_2)\dl_{\bar{a}_1}(x_3) \\
      \beta\cdot h(x_1,x_2)\dl_{a_1}(x_3) + \alpha\cdot h'(x_1,x_2)\dl_{\bar{a}_1}(x_3)
     \end{align*}
     depending on whether $a_2$ is 0 or 1.
     Note that these functions differ from the function considered in the previous subcase --
     $h(x_1,x_2)\dl_{a_1}(x_3) + h'(x_1,x_2)\dl_{\bar{a}_1}(x_3)$ -- only by a SLOCC with $\smm{\alpha&0\\0&\beta}$ or $\smm{\beta&0\\0&\alpha}$ on the third input (depending on whether $a_1=a_2$).
     Hence they must be in the same entanglement class and hardness follows as in the previous subcase.  

     This is another change compared to \cite{cai_dichotomy_2017}, where the hardness lemma only applied to real values and the construction did not employ the function $[\alpha,\beta]$.}
   
    \item If $D_3\geq 3$, we consider different cases according to the relationships between the values of $h$ and $h'$.
    
     {If $\alpha\alpha'=\beta\beta'=\ld\ld'=\mu\mu'=0$, then $\alpha'=\beta'=0$ since $\alpha,\beta$ are non-zero by assumption.
     Thus, at least one of $\ld'$ and $\mu'$ must be non-zero because $h'$ is not identically zero.
     Also, at least one of $\ld$ and $\mu$ must be non-zero because $h$ is non-decomposable.
     Therefore we have either $\ld=\mu'=0$ and $\ld',\mu\neq 0$, or $\ld'=\mu=0$ and $\ld,\mu'\neq 0$.}
     In the first case, $\ld=\mu'=0$ and $\ld',\mu\neq 0$, pin the first input to 1 to get:
     \[
      \mu\dl_1(x_2)\prod_{k=1}^{D_3}\dl_{a_k}(x_{k+2})+\ld'\dl_0(x_2)\prod_{k=1}^{D_3}\dl_{\bar{a}_k}(x_{k+2}).
     \]
     This is a generalised equality function of arity at least 4, so we can proceed as before.
     In the second case, $\ld'=\mu=0$ and $\ld\mu'\neq 0$, pinning the first input to 1 {works again, albeit resulting in a different generalised equality function}.
       
     Otherwise, there exists a pair of primed and unprimed coefficients of the same label that are both non-zero.
     If these are $\alpha$ and $\alpha'$, pin the first two inputs to 00 to get a generalised equality of arity at least 3.
     If the non-zero pair are $\beta$ and $\beta'$, pin to 01, and so on.
     Given the generalised equality function, we may proceed as in the cases $D_0\geq4$ or $D_0\geq 3$, depending on whether the arity of this function is even or odd.
   \end{itemize}
 \end{itemize}

 We have covered all cases, hence the proof is complete.
\end{proof}

\section*{Acknowledgements}
I would like to thank Pinyan Lu for pointing out a flaw in the original statement of the main theorem, Jin-Yi Cai for interesting discussions about extending the $\hol^+$ result to the planar case, and Mariami Gachechiladze and Otfried G\"uhne for pointing out a significantly shorter and more elegant proof of Theorem~\ref{thm:three-qubit-entanglement}, and for letting me use it here.
Thanks also go to Leslie Ann Goldberg, Ashley Montanaro, and William Whistler for for helpful comments on earlier versions of this paper, as well as to the anonymous referees for their insightful feedback and suggestions.

\bibliographystyle{plain}
\bibliography{refs}

\end{document}